\documentclass[11pt,a4paper]{article}
\usepackage{amsfonts}
\usepackage[ansinew]{inputenc}
\usepackage{verbatim}
\usepackage{graphicx}
\usepackage{tabularx}
\usepackage{array}
\usepackage{delarray}
\usepackage{dcolumn}
\usepackage{float}
\usepackage{amsmath}
\usepackage{amstext}
\usepackage{amssymb}
\usepackage{hyperref}
\usepackage{url}
\usepackage{xr}
\usepackage{amsthm}

\hypersetup{citecolor=true,colorlinks=true,allcolors=blue}
\RequirePackage[round,authoryear,longnamesfirst]{natbib}
\newtheorem{lem}{Lemma}[section]
\newtheorem{thm}{Theorem}[section]
\newtheorem{cor}{Corollary}[section]

\newtheorem{ass}{Assumption}

\newtheorem*{ex*}{Example}
\newtheoremstyle{named}{}{}{\itshape}{}{\bfseries}{.}{.5em}{#1\thmnote{ #3}}
\theoremstyle{named}

\newtheorem{ex}{Example}[section]

\DeclareMathOperator*{\argmin}{arg\,min}

\numberwithin{equation}{section}
\input{tcilatex}
\begin{document}

\title{Strong Consistency of Spectral Clustering for Stochastic Block Models 
}
\author{Liangjun Su\thanks{
Singapore Management University.\ E-mail~address: ljsu@smu.edu.sg. Su
acknowledges the funding support provided by the Lee Kong Chian Fund for
Excellence.} \and Wuyi Wang\thanks{%
Jinan University. \ E-mail~address: wangwuyi@live.com.} \and Yichong Zhang%
\thanks{
Singapore Management University.\ E-mail~address: yczhang@smu.edu.sg. The
corresponding author.} }
\maketitle

\begin{abstract}
In this paper we prove the strong consistency of several methods based on
the spectral clustering techniques that are widely used to study the
community detection problem in stochastic block models (SBMs). We show that
under some weak conditions on the minimal degree, the number of communities,
and the eigenvalues of the probability block matrix, the K-means algorithm
applied to the eigenvectors of the graph Laplacian associated with its first
few largest eigenvalues can classify all individuals into the true community
uniformly correctly almost surely. Extensions to both regularized spectral
clustering and degree-corrected SBMs are also considered. We illustrate the
performance of different methods on simulated networks. \medskip

\vspace{2mm}
\end{abstract}

\section{Introduction}

\label{sec:intro} Community detection is one of the fundamental problems in
network analysis, where communities are groups of nodes that are, in some
sense, more similar to each other than to the other nodes. The stochastic
block model (SBM) that was first proposed by \cite{HLL83} is a common tool
for model-based community detection that has been widely studied in the
statistics literature. Within the SBM framework, the most essential task is
to recover the community membership of the nodes from a single observation
of the network. Various procedures have been proposed to solve this problem
in the last decade or so. These include method of moments \citep{BCL11},
modularity maximization \citep{NG04}, semidefinite programming %
\citep{ABH16,CL15}, spectral clustering %
\citep{JY16,LR15,QR13,RCY11,SB15,V18,YP14,YP16}, likelihood methods %
\citep{ACBL13,BC09,CWA12,ZLZ12}, and spectral embedding %
\citep{LSTAP14,STFP12}. \cite{A18} provides an excellent survey on recent
developments on community detection and stochastic block models. Among the
methods mentioned above, spectral clustering is arguably one of the most
widely used methods due to its computational tractability.

\cite{BC09} introduce the notion of strong consistency of community
detection as the number of nodes, $n,$ grows.\footnote{\cite{BC09} use the
terminology \textquotedblleft asymptotic consistency\textquotedblright\ in
place of strong consistency.} By strong consistency, they mean that one can
identify the members of the block model communities perfectly in large
samples. Based on the parameters of the block model, properties of the
modularities, and expected degree of the graph ($\lambda _{n}$), \cite{BC09}
give the sufficient conditions for strong consistency, which is $\lambda
_{n}/\log (n)\rightarrow \infty.$ \cite{ZLZ12} define weak consistency of
community detection, which essentially means that the number of
misclassified nodes is of smaller order than the number of nodes. \cite{BC12}
find that weak consistency requires that $\lambda _{n}\rightarrow \infty $
for the SBM. Similarly, under the conditions that $\lambda _{n}/\log
(n)\rightarrow \infty $ ($\lambda _{n}\rightarrow \infty )$, \cite{ZLZ12}
establish the strong (weak) consistency under both standard SBMs and
degree-corrected SBMs.

If the community detection method is strongly consistent, then it means that
the communities are exactly recoverable. From an information-theory
perspective, \cite{AS15}, \cite{ABH16}, \cite{MNS14}, and \cite{V18} study
the phase transition threshold for exact recovery, which requires $\lambda_n
= \Omega(\log(n))$. It is well known that some methods like the modularity
maximization of \cite{NG04} and the likelihood method of \cite{BC09} yield
strongly consistent community recovery, but they either rely on
combinatorial methods that are computationally demanding or are guaranteed
to be successful only when the starting values are well-chosen. \cite{ABH16}
show that semidefinite programming can achieve exact recovery when there are
two equal-sized communities. \cite{YP14}, \cite{YP16}, and \cite{V18}
establish strong consistency for the variants of spectral method, which
involve graph splitting, trimming, and a final improvement step. The pure
spectral clustering method has been shown to enjoy weak consistency under
standard or degree-corrected SBMs by various researchers; see \cite{JY16}, 
\cite{LR15}, \cite{QR13}, and \cite{RCY11}. Weak consistency here means that
the fraction of misclassified nodes decreases to zero as $n$ grows. Because
the decrease rates established in above papers are usually slower than $n$,
the above weak consistency results imply that the number of misclassified
nodes still increases to infinity as $n$ grows. On the contrary, strong
consistency implies that the number of misclassified nodes is zero for
sufficiently large $n$, which greatly improves upon weak consistency.



The aim of this paper is to formally establish the strong consistency of
spectral clustering for standard/regular SBMs without any extra refinement
steps, under a set of conditions on the minimal degree of nodes ($\mu _{n}$%
), the number of communities ($K$), the minimal value of the nonzero
eigenvalue of the normalized block probability matrix, and some other
parameters of the block model. In the special case where $K$ is fixed and
the normalized block probability matrix has minimal eigenvalue bounded away
from zero in absolute value, we show that $\mu _{n}/\log (n)$ being
sufficiently large can ensure strong consistency. In other words, the
spectral clustering method achieves the optimal rate for exact recovery, as
pointed out in \cite{ABH16} and \cite{AS15}.

As demonstrated by \cite{ACBL13}, the performance of spectral clustering can
be considerably improved via regularization. \cite{JY16} provide an attempt
at quantifying this improvement through theoretical analysis and find that
the typical minimal degree assumption for the consistency of spectral
clustering can potentially be removed with suitable regularization. In this
paper, we also establish the strong consistency of regularized spectral
clustering.

The SBM is limited by its assumption that all nodes within a community are
stochastically equivalent and thus provides a poor fit to real-world
networks with hubs or highly varying node degrees within communities. For
this reason, \cite{KN11} propose a degree-corrected SBM (DC-SBM) to allow
variation in node degrees within a community while preserving the overall
block community structure. The DC-SBM greatly enhances the flexibility of
modeling degree heterogeneity and enables us to fit network data with
varying degree distributions. We also prove the strong consistency of
spectral clustering for regularized DC-SBMs.

Our paper is mostly related to \cite{abbe2017}. \cite{abbe2017} derive the $%
L_{\infty }$ bound for the entrywise eigenvector of random matrices with low
expected rank. Then they apply their general results to SBM with two
communities, where both within- and cross-community probabilities are of
order $\log (n)/n$ and show that classifying nodes based on the sign of the
entries in the second eigenvector can achieve exact recovery. Our paper
complements theirs in the following three aspects. First, we consider the
eigenvectors of normalized graph Laplacian $L$ rather than the adjacency
matrix $A$. Therefore, the entrywise bound of the eigenvectors derived in 
\cite{abbe2017} cannot be directly used in our case. Our proof relies on the
construction of a contraction mapping for the entrywise bound, via which we
can iteratively refine the bound. Such strategy is different from that in 
\cite{abbe2017}.

Second, we consider SBM with a general block probability matrix whereas \cite%
{abbe2017} consider a $2\times 2$ block probability matrix. Even though \cite%
{abbe2017} establish general theories of $L_{\infty }$ bound for the
entrywise eigenvector of random matrices, when applying their theory to
SBMs, they only study the model with the following block probability matrix: 
\begin{equation}
\begin{pmatrix}
\frac{a\log (n)}{n} & \frac{b\log (n)}{n} \\ 
\frac{b\log (n)}{n} & \frac{a\log (n)}{n}%
\end{pmatrix}%
.  \label{eq:bl}
\end{equation}%
Their block probability matrix assumes that there are two groups, the
connection probability within groups are the same for the two groups, and
the within- and cross-group connection probabilities are of the same order
of $\log (n)/n$. In contrast, our paper studies the general SBM with generic 
$K$ groups, where $K$ is allowed to diverge to infinity at a slow rate and
the decay rates for different elements in the block probability matrix can
be different. When there are two communities, \cite{abbe2017} use the sign
of the eigenvector associated with the second largest eigenvalue (in
absolute value) to identify the node's membership. When $K>2$, just checking
the sign is not sufficient to identify all $K$ groups. Our paper shows that
applying the K-means algorithm to the first $K$ eigenvectors can achieve
strong consistency.

Third, we consider SBM with both regularization and degree correction. We
show that, by regularization, the strong consistency is still possible even
when the minimal degree does not diverge at all. For the DC-SBM with
regularization, we also derive the conditions for strong consistency.
Neither regularization nor degree-corrected SBM is discussed in \cite%
{abbe2017}.

In the simulation, we consider both standard SBMs and DC-SBMs. For standard
SBMs, we adopt \cite{JY16}'s regularization method and choose the tuning
parameter $\tau $ according to their recommendation. The results show that
in terms of classification, spectral clustering tends to outperform the
unconditional pseudo-likelihood (UPL) method, which also has the strong
consistency property \citep{ACBL13}. In contrast, for the DC-SBMs our
simulations suggest that the regularized spectral clustering tends to
slightly underperform the conditional pseudo-likelihood (CPL) method even
though both are strongly consistent under some conditions. We also show that
an adaptive procedure helps the regularized spectral clustering to achieve
much better performance than the CPL method.

The rest of the paper is organized as follows. We study the strong
consistency of spectral clustering for the basic SBMs in Section \ref{sec:SC}%
. We consider the extensions to regularized spectral clustering and
degree-corrected SBMs in Section \ref{sec:ext}. Section \ref{sec:sim}
reports the numerical performance of various spectral-clustering-based
methods for a range of simulated networks. Section \ref{sec:strategy}
describes the proof strategy of the key theorem in our paper. Section \ref%
{sec:concl} concludes. The proofs of the main results are relegated to the
mathematical appendix.

\textit{Notation.} Throughout the paper, we use $[M]_{ij}$ and $[M]_{i\cdot
} $ to denote the $(i,j)$-th entry and $i$-th row of matrix $M$,
respectively. Without confusion, we sometimes simplify $[M]_{ij}$ as $M_{ij}$%
. $\Vert M\Vert $ and $\Vert M\Vert _{F}$ denote the spectral norm and
Frobenius norm of $M,$ respectively. Note that $\Vert M\Vert =\Vert M\Vert
_{F}$ when $M$ is a vector. In addition, let $\Vert M\Vert _{2\rightarrow
\infty }=\sup_{i}\Vert \lbrack M]_{i\cdot }\Vert .$ We use $\mathbf{1}%
\left\{ \cdot \right\} $ to denote the indicator function which takes value
1 when $\cdot $ holds and 0 otherwise. $C_1$ and $c_1$ denote specific
absolute constants that remain the same throughout the paper.

\section{Strong consistency of spectral clustering}

\label{sec:SC}

\subsection{Basic setup}

Let $A\in \{0,1\}^{n\times n}$ be the adjacency matrix. By convention, we do
not allow self-connection, i.e., $A_{ii}=0$. Let $\hat{d}_{i}=%
\sum_{j=1}^{n}A_{ij}$ denote the degree of node $i$, $D=\text{diag}(\hat{d}%
_{1},\ldots ,\hat{d}_{n})$, and $L=D^{-1/2}AD^{-1/2}$ be the graph
Laplacian. The graph is generated from a SBM with $K$ communities. We assume
that $K$ is known and potentially depends on the number of nodes $n$. We
omit the dependence of $K$ on $n$ for notation simplicity. If $K$ is
unknown, it can be determined by either \citeauthor{Lei16}'s \citeyear{Lei16}
sequential goodness-of-fit testing procedure, the likelihood-based model
selection method proposed by \cite{WB16}, or the network cross-validation
method proposed by \cite{CL17}. The communities, which represent a partition
of the $n$ nodes, are assumed to be fixed beforehand. Denote these by $%
C_{1},\ldots ,C_{K}$. Let $n_{k}$, for $k=1,\ldots ,K$, be the number of
nodes belonging to each of the clusters.

Given the communities, the edge between nodes $i$ and $j$ are chosen
independently with probability depending on the communities $i$ and $j$
belong to. In particular, for nodes $i$ and $j$ belonging to cluster $%
C_{k_{1}}$ and $C_{k_{2}}$, respectively, the probability of edge between $i$
and $j$ is given by $P_{ij}=B_{k_{1}k_{2}}$, where the \textit{block
probability matrix} $B=\{B_{k_{1}k_{2}}\}$, $k_{1},k_{2}=1,\ldots ,K$, is a
symmetric matrix with each entry between $[0,1]$. The $n\times n$ edge
probability matrix $P=\{P_{ij}\}$ represents the population counterpart of
the adjacency matrix $A$. Frequently we suppress the dependence of matrices
and their elements on $n.$

Denote $Z=\{Z_{ik}\}$ as the $n\times K$ binary matrix providing the cluster
membership of each node, i.e., $Z_{ik}=1$ if node $i$ is in $C_{k}$ and $%
Z_{ik}=0$ otherwise. Then we have $P=ZBZ^{T}.$ Let $\mathcal{D}=\text{diag}%
(d_{1},\ldots ,d_{n})$ where $d_{i}=\sum_{j=1}^{n}P_{ij}$. The population
version of the graph Laplacian is $\mathcal{L}=\mathcal{D}^{-1/2}P\mathcal{D}%
^{-1/2}.$ The standard spectral clustering corresponds to classifying the
eigenvectors of $L$ by K-means algorithm. In this paper, we focus on the
strong consistency of both the standard spectral clustering and its variant.


\subsection{Identification of the group membership}

\label{sec:id} Let $\pi _{kn}=n_{k}/n$, $W_{k}=[B]_{k\cdot }Z^{T}\iota
_{n}/n=\sum_{l=1}^{K}B_{kl}\pi _{ln}$, $\mathcal{D}_{B}=\text{diag}%
(W_{1},\ldots ,W_{K})$, and $B_{0}=\mathcal{D}_{B}^{-1/2}B\mathcal{D}%
_{B}^{-1/2}$, where $\iota _{n}$ is a vector of ones in $\Re ^{n}$. We can
view $W_{k}$ as the weighted average of the $k$-th row of $B$ with weights
given by $\pi _{kn}.$ Similarly, $B_{0}$ is a normalized version of $B$.
Note that $B_{0}$ is symmetric as $B$ is. Let $\Pi _{n}=\text{diag}(\pi
_{1n},\ldots ,\pi _{Kn})$. Throughout the paper, we allow for the elements
in the block probability matrix $B$ to depend on $n$ and decay to zero as $n$
grows, which leads to a sparse graph.

\begin{ass}
\label{ass:id} $B_{0}$ has rank $K$ and the spectral decomposition of $\Pi
_{n}^{1/2}B_{0}\Pi _{n}^{1/2}$ is $S_{n}\Omega _{n}S_{n}^{T}$, in which $%
S_{n}$ is a $K\times K$ matrix such that $S_{n}^{T}S_{n}=I_{K}$ and $\Omega
_{n}=\text{diag}(\omega _{1n},\ldots ,\omega _{Kn})$ such that $|\omega
_{1n}|\geq \cdots \geq |\omega _{Kn}|>0$.
\end{ass}

Assumption \ref{ass:id} implies that $B=\mathcal{D}_{B}^{1/2}\Pi
_{n}^{-1/2}S_{n}\Omega _{n}S_{n}^{T}\Pi _{n}^{-1/2}\mathcal{D}_{B}^{1/2}$
and $B_{0}=\Pi _{n}^{-1/2}S_{n}\Omega _{n}S_{n}^{T}\Pi _{n}^{-1/2}$. The
full-rank assumption is also made in \cite{RCY11}, \cite{LR15}, and \cite%
{JY16} and can be relaxed at the cost of more complicated notation.\footnote{%
The first version of our paper only requires that $B_{0}$ has distinct rows
and rank $K^{*}$, which can be less than $K$. Then, researchers need to
apply K-means algorithm to the first $K^{*}$ eigenvectors. By modifying the
corresponding assumptions accordingly, the strong consistency result in this
paper still holds. We stick to the full rank case mainly for notation
simplicity.} In addition, we allows for the possibility that $K\rightarrow
\infty $ and/or $\omega _{Kn}\rightarrow 0$ as $n\rightarrow \infty $ below.
This also mitigates concern of the full-rank condition. Assumption \ref%
{ass:id} implies that $\mathcal{L}$ has rank $K$ and the following spectral
decomposition: 
\begin{equation*}
\mathcal{L}=U_{n}\Sigma _{n}U_{n}^{T}=U_{1n}\Sigma _{1n}U_{1n}^{T},
\end{equation*}%
where $\Sigma _{n}=\text{diag}(\sigma _{1n},\ldots ,\sigma _{Kn},0,\ldots
,0) $ is a $n\times n$ matrix that contains the eigenvalues of $\mathcal{L}$
such that $|\sigma _{1n}|\geq |\sigma _{2n}|\geq \cdots \geq |\sigma
_{Kn}|>0 $, $\Sigma _{1n}=\text{diag}(\sigma _{1n},\ldots ,\sigma _{Kn})$,
the columns of $U_{n}$ contain the eigenvectors of $\mathcal{L}$ associated
with the eigenvalues in $\Sigma _{n}$, $U_{n}=(U_{1n},U_{2n})$, and $%
U_{n}^{T}U_{n}=I_{n}$. As shown in Theorem \ref{thm:id} below, $\sigma
_{kn}=\omega _{kn}$ for $k=1,\ldots ,K$.

\begin{ass}
\label{ass:nk}There exist some constants $C_1$ and $c_1$ such that 
\begin{equation*}
\infty >C_1 \geq \limsup_{n}\sup_{k}n_{k}K/n\geq
\liminf_{n}\inf_{k}n_{k}K/n\geq c_1>0.
\end{equation*}
\end{ass}

Assumption \ref{ass:nk} implies that the network has balanced communities.
It is commonly assumed in the literature on strong consistency of community
detection; see, e.g., \cite{BC09}, \cite{ZLZ12}, \cite{ACBL13}, and \cite%
{AS15}. %

\begin{thm}
\label{thm:id}Let $z_{i}^{T}=\left[ Z\right] _{i\cdot },$ the $i$-th row of $%
Z $. If Assumptions \ref{ass:id} and \ref{ass:nk} hold, then $\Omega
_{n}=\Sigma _{1n}$, $U_{1n}=Z(Z^{T}Z)^{-1/2}S_{n}$ and 
\begin{equation*}
\sup_{1\leq i\leq n}(n/K)^{1/2}\Vert z_{i}^{T}(Z^{T}Z)^{-1/2}S_{n}\Vert \leq
c_1^{-1/2}.
\end{equation*}%
In addition, for $n$ sufficiently large, if $z_{i}\neq z_{j}$, then 
\begin{equation*}
(n/K)^{1/2}\Vert (z_{i}^{T}-z_{j}^{T})(Z^{T}Z)^{-1/2}S_{n}\Vert \geq
C_1^{-1/2}\sqrt{2}>0.
\end{equation*}
\end{thm}

Noting that the $i$th row of $U_{1n}$ is given by $%
z_{i}^{T}(Z^{T}Z)^{-1/2}S_{n}$. Theorem \ref{thm:id} indicates that the rows
of $U_{1n}$ contain the same community information as $Z$ for all nodes in
the network. Therefore, we can infer each node's community membership based
on the eigenvector matrix $U_{1n}$ if $\mathcal{L}$ is observed.

In practice, $\mathcal{L}$ is not observed. But we can estimate it by $L.$
We show below that the eigenvectors of $L$ associated with its $K$ largest
eigenvalues in absolute value consistently estimate those of $\mathcal{L}$
up to an orthogonal matrix so that the rows of the eigenvector matrix of $L$
also contains the useful community information.

\subsection{Uniform bound for the estimated eigenvectors}

\label{sec:unif} To study the upper bound of the eigenvectors of $L$
associated with its $K$ largest eigenvalues, we add the following assumption.

\begin{ass}
\label{ass:rate}Let $\mu _{n}=\min_{i}d_{i}$ and $\rho _{n}=\max
(\sup_{k_{1}k_{2}}[B_{0}]_{k_{1}k_{2}},1)$. Then, for $n$ being sufficiently
large, 
\begin{align*}
\frac{K\rho _{n}\log ^{1/2}(n)}{\mu _{n}^{1/2}\sigma _{Kn}^{2}}%
\left(1+\rho_n + \left(\frac{1}{K} + \frac{\log(5)}{\log(n)}%
\right)^{1/2}\rho_n^{1/2}\right) \leq 10^{-8}C_1^{-1}c_1^{1/2}.
\end{align*}
\end{ass}


Several remarks are in order. First, $\rho _{n}$ is a measure of
heterogeneity of the normalized block probability matrix $B_{0}$. If all the
entries in $B$ are of the same order of magnitude, then $\rho _{n}$ is
bounded. In addition, by Assumption \ref{ass:nk} and the fact that 
\begin{equation*}
(\pi _{k_{1}n}\pi _{k_{2}n})^{1/2}[B_{0}]_{k_{1}k_{2}}=\frac{(\pi
_{k_{1}n}\pi _{k_{2}n})^{1/2}B_{k_{1}k_{2}}}{(\sum_{l=1}^{K}\pi
_{ln}B_{k_{1}l})^{1/2}(\sum_{l=1}^{K}\pi _{ln}B_{k_{2}l})^{1/2}}\leq 1,
\end{equation*}%
we have $\limsup_{n} \rho _{n}\leq c_1^{-1}K$. Therefore, if the number of
blocks is fixed, then $\rho _{n}$ is also bounded.

Second, if $K$ is fixed and $\liminf_{n}|\sigma _{Kn}|$ is bounded away from
zero, then Assumption \ref{ass:rate} reduces to the requirement that $\mu_n
\geq \underline{C}\log(n)$ for some constant $\underline{C}$. Therefore,
Assumption \ref{ass:rate} allows for $\mu_n = \Omega(\log(n))$. Such
condition is the minimal requirement for strong consistency (exact
recovery), as established in \cite{ABH16} and \cite{AS15}. Our results in
Theorem \ref{thm:strong} based on Assumption \ref{ass:rate} imply that, in
the baseline case, the spectral clustering method achieve strong consistency
under this minimal rate requirement.

Third, to provide a more detailed comparison between Assumption \ref%
{ass:rate} and the phase transition threshold, let us consider the special
case where there are two equal sized communities and the block probability
matrix is 
\begin{equation*}
B=%
\begin{pmatrix}
\frac{a\log (n)}{n} & \frac{b\log (n)}{n} \\ 
\frac{b\log (n)}{n} & \frac{a\log (n)}{n}%
\end{pmatrix}%
,
\end{equation*}%
where $a>b.$ In this case, $K=2$, $\Pi _{n}=\text{diag}(0.5,0.5)$, $\mathcal{%
D}_{B}=\text{diag}(\frac{(a+b)\log (n)}{2n},\frac{(a+b)\log (n)}{2n})$, and%
\begin{equation*}
B_{0}=\mathcal{D}_{B}^{-1/2}B\mathcal{D}_{B}^{-1/2}=%
\begin{pmatrix}
\frac{2a}{a+b} & \frac{2b}{a+b} \\ 
\frac{2b}{a+b} & \frac{2a}{a+b}%
\end{pmatrix}%
.
\end{equation*}%
Note that $\mu _{n}=\frac{(a+b)\log (n)}{2}$, $\rho _{n}=\frac{2a}{a+b}\in
(1,2)$, and $\sigma _{2n}$, the second eigenvalue of $\Pi _{n}^{1/2}B_{0}\Pi
_{n}^{1/2}$, is $\frac{a-b}{a+b}$. Then, Assumption \ref{ass:rate} boils
down to 
\begin{equation*}
\left( \frac{2a}{a+b}\right) ^{2}\sqrt{\frac{2}{a+b}}\left( \frac{a+b}{a-b}%
\right) ^{2}\leq \underline{c}
\end{equation*}%
for some small constant $0.0001>\underline{c}>0$. Since $\frac{2a}{a+b}\geq
1 $ and $\frac{a+b}{a-b}>1$, the above condition implies that 
\begin{equation*}
\underline{c}\geq \left( \frac{2a}{a+b}\right) ^{2}\sqrt{\frac{2}{a+b}}%
\left( \frac{a+b}{a-b}\right) ^{2}\geq \frac{\sqrt{2(a+b)}}{a-b}\geq \frac{%
\sqrt{a}+\sqrt{b}}{a-b}=\frac{1}{\sqrt{a}-\sqrt{b}},
\end{equation*}%
or equivalently, 
\begin{equation*}
\sqrt{a}-\sqrt{b}\geq \underline{c}^{-1}>\sqrt{2}.
\end{equation*}%
Because $\sqrt{2}$ is the information-theoretic threshold for exact recovery
established in \cite{ABH16}, Assumption \ref{ass:rate} ensures that the SBM
under our consideration is in the region that exact recovery is solvable.

Fourth, the constants in Assumption \ref{ass:rate}, and thus, $\underline{c}$
in the above remark, are not optimal. We choose these constants purely for
their technical ease. We conjecture that more sophisticated arguments such
as those in \cite{AS15}, \cite{ABH16}, and \cite{abbe2017} are needed to
establish the optimal constant for the exact recovery of spectral clustering
method. On the other hand, although our method cannot show the exact
recovery all the way down to the information-theoretic threshold, it can be
easily extended to handle degree-corrected and/or regularized SBM, as shown
in Section \ref{sec:ext}.

Consider the spectral decomposition 
\begin{equation*}
L=\hat{U}_{n}\widehat{\Sigma }_{n}\hat{U}_{n}^{T},
\end{equation*}%
where $\hat{\Sigma}_{n}=\text{diag}(\hat{\sigma}_{1n},\ldots ,\hat{\sigma}%
_{nn})$ with $|\hat{\sigma}_{1n}|\geq |\hat{\sigma}_{2n}|\geq \cdots \geq |%
\hat{\sigma}_{nn}|\geq 0$, and $\hat{U}_{n}$ is the corresponding
eigenvectors. Let $\hat{\Sigma}_{1n}=\text{diag}(\hat{\sigma}_{1n},\ldots ,%
\hat{\sigma}_{Kn})$, $\hat{\Sigma}_{2n}=\text{diag}(\hat{\sigma}%
_{K+1,n},\ldots ,\hat{\sigma}_{nn})$, and $\hat{U}_{n}=(\hat{U}_{1n},\hat{U}%
_{2n})$, where $\hat{U}_{1n}$ contains the eigenvectors associated with
eigenvalues $\hat{\sigma}_{1n},\ldots ,\hat{\sigma}_{Kn}$. Then, $\hat{U}%
_{1n}^{T}\hat{U}_{1n}=I_{K}$, $\hat{U}_{2n}^{T}\hat{U}_{1n}=0$, and 
\begin{equation*}
L=\hat{U}_{1n}\hat{\Sigma}_{1n}\hat{U}_{1n}^{T}+\hat{U}_{2n}\hat{\Sigma}_{2n}%
\hat{U}_{2n}^{T}.
\end{equation*}%
The following lemma indicates that $L$ and $\hat{U}_{1n}$ are close to their
population counterparts, and up to an orthogonal matrix in the latter case.

\begin{lem}
\label{lem:dk} If Assumptions \ref{ass:id}--\ref{ass:rate} hold, then there
exists a $K\times K$ orthogonal (random) matrix $\hat{O}_{n}$ such that 
\begin{equation*}
\Vert \mathcal{L}-L\Vert \leq 7\log ^{1/2}(n)\mu _{n}^{-1/2}\quad a.s.
\end{equation*}%
and 
\begin{equation*}
\Vert \hat{U}_{1n}\hat{O}_{n}-U_{1n}\Vert \leq 10\log ^{1/2}(n)\mu
_{n}^{-1/2}|\sigma _{Kn}^{-1}|\quad a.s.
\end{equation*}
\end{lem}

Two variants of Lemma \ref{lem:dk} have been derived in \cite{JY16} and \cite%
{QR13} as special cases. The main difference is that we obtain the almost
sure bound for the objects of interest instead of the probability bound in
those papers. As illustrated in \cite{abbe2017}, 
\begin{equation*}
\hat{O}_{n}=\bar{U}\bar{V}^{T},
\end{equation*}%
where $\bar{U}\bar{\Sigma}\bar{V}^{T}$ is the singular value decomposition
of $\hat{U}_{1n}^{T}U_{1n}$. Apparently, $\hat{O}_{n}$ is random.

In order to study the strong consistency, we have to derive the uniform
bound for $\Vert \hat{u}_{1i}^{T}\hat{O}_{n}-u_{1i}^{T}\Vert $, where $\hat{u%
}_{1i}^{T}$ and $u_{1i}^{T}$ are the $i$-th rows of $\hat{U}_{1n}$ and $%
U_{1n}$, respectively.

\begin{thm}
\label{thm:main} If Assumptions \ref{ass:id}--\ref{ass:rate} hold, then 
\begin{equation*}
\sup_{i}\sqrt{n/K}\Vert \hat{u}_{1i}^{T}\hat{O}_{n}-u_{1i}^{T}\Vert \leq C^*%
\frac{\rho _{n}\log ^{1/2}(n)}{\mu _{n}^{1/2}\sigma _{Kn}^{2}}\left(1+\rho_n
+ \left(\frac{1}{K} + \frac{\log(5)}{\log(n)}\right)^{1/2}\rho_n^{1/2}%
\right)\quad a.s.,
\end{equation*}
where $C^*$ is the same absolute constant as in Theorem \ref{thm:main_DC}.
\end{thm}

We consider the four-parameter SBM studied in \cite{RCY11} to illustrate the
upper bound in Theorem \ref{thm:main}.

\begin{ex}[\textbf{2.1}]
The SBM is parametrized by $K,$ $s,$ $r$ and $p,$ where the $K$ communities
contain $s$ nodes each, and $r$ and $r+p$ denote the probability of a
connection between two nodes in two separate blocks and in the same block,
respectively. For this model, $\rho _{n}=\frac{(p+r)K}{p+rK} $, $\sigma
_{Kn}=\frac{p}{Kr+p}$, and $\mu _{n}=\frac{ n(p+rK)}{K}-\left( p+r\right) $.
Therefore, the probability bound of $\sup_{i}\sqrt{n/K}\Vert \hat{u}%
_{1i}-O_{n}^{T}u_{1i}\Vert $ is of order 
\begin{equation}  \label{eq:ex}
\biggl(\frac{K\log (n)}{n(p+rK)}\biggr)^{1/2}\biggl(\frac{(p+r)^{2}K^{2}}{
p^{2}}\biggr).
\end{equation}
The above display is small if $K^{5}\log \left( n\right) /(np)$ is small and 
$rK/p\rightarrow c\in \left( 0,\infty \right)$, or if $K^{4}\log \left(
n\right) /(nr)$ is small and $r/p\rightarrow c\in \left( 0,\infty \right).$
If we further restrict our attention to the dense SBM with both $r$ and $p$
bounded away from zero, then the displayed item in \eqref{eq:ex} becomes
small as long as $K^{4}\log \left( n\right) /n$ is small.
\end{ex}

Since both $U_{1n}$ and $\hat{U}_{1n}$ have orthonormal columns, they have a
typical element of order $(n/K)^{-1/2}.$ This explains why we need the
normalization constant $(n/K)^{1/2}$ in Theorem \ref{thm:main}. An important
implication of Theorem \ref{thm:main} is that like $U_{1n},$ the rows of $%
\hat{U}_{1n}$ also contain the community membership information. Let $\hat{%
\beta}_{in}=(n/K)^{1/2}\hat{u}_{1i}^{T}.$ Let $g_{i}^{0}\in \{1,\ldots ,K\}$
denote the true community that node $i$ belongs to. Theorems \ref{thm:id}-%
\ref{thm:main} and the fact that $\hat{O}_{n}\hat{O}_{n}^{T}=I_{K}$ imply
that there exist $\beta_{kn} = (K\pi_{kn})^{-1/2}[S_n\hat{O}_n^T]_{k\cdot}$, 
$k = 1,\cdots,K$ such that 
\begin{equation*}
(n/K)^{1/2}u_{1i}^{T}\hat{O}_{n}^{T} = \beta _{g_{i}^{0}n}, \quad
||\beta_{kn}|| \leq c_1^{-1/2},
\end{equation*}
and 
\begin{equation*}
\sup_i\Vert \hat{\beta}_{in}-\beta _{g_{i}^{0}n}\Vert \leq C^*\frac{\rho
_{n}\log ^{1/2}(n)}{\mu _{n}^{1/2}\sigma _{Kn}^{2}}\left(1+\rho_n + \left(%
\frac{1}{K} + \frac{\log(5)}{\log(n)}\right)^{1/2}\rho_n^{1/2}\right) \quad
a.s.
\end{equation*}%
If the distance between $\hat{\beta}_{in}$ and $\beta _{g_{i}^{0}n}$ is much
smaller than that among distinctive $\{\beta_{kn}\}_{k=1}^K$, then K-means
algorithm applying to $\{\hat{\beta}_{in}\}_{i=1}^n$ are expected to recover
the true community memberships. The statistical properties of K-means method
are studied in the next two sections.

\subsection{Strong consistency of the K-means algorithm}

\label{sec:kmeans} With a little abuse of notation, let $\hat{\beta}_{in} \in \Re^K$ be
a generic estimator of $\beta _{g_{i}^{0}n} \in \Re^K$ for $i=1,\ldots ,n.$
To recover the community membership structure (i.e., to estimate $g_{i}^{0}$%
), it is natural to apply the K-means clustering algorithm to $\{\hat{\beta}%
_{in}\}$. Specifically, let $\mathcal{A}=\{\alpha _{1},\ldots ,\alpha _{K}\}$
be a set of $K$ arbitrary $K\times 1$ vectors: $\alpha _{1},\ldots ,\alpha
_{K}$. Define 
\begin{equation*}
\widehat{Q}_{n}(\mathcal{A})=\frac{1}{n}\sum_{i=1}^{n}\min_{1\leq l\leq
K}\Vert \hat{\beta}_{in}-\alpha _{l}\Vert ^{2}
\end{equation*}%
and $\widehat{\mathcal{A}}_{n}=\{\widehat{\alpha }_{1},\ldots ,\widehat{%
\alpha }_{K}\}$, where $\widehat{\mathcal{A}}_{n}=\argmin_{\mathcal{A}}%
\widehat{Q}_{n}(\mathcal{A}).$ Then we compute the estimated cluster
identity as 
\begin{equation*}
\hat{g}_{i}=\argmin_{1\leq l\leq K}\Vert \hat{\beta}_{in}-\widehat{\alpha }%
_{l}\Vert ,
\end{equation*}%
where if there are multiple $l$'s that achieve the minimum, $\hat{g}_{i}$
takes value of the smallest one. Next, we consider the case in which the
estimates $\{\hat{\beta}_{in}\}_{i=1}^{n}$ and the true vectors $\{\beta
_{kn}\}_{k=1}^{K}$ satisfy the following restrictions.

\begin{ass}
\label{ass:theta}

\begin{enumerate}
\item There exists a constant $M$ such that 
\begin{equation*}
\limsup_{n} \sup _{1\leq k\leq K} \Vert \beta _{kn}\Vert \leq M<\infty .
\end{equation*}

\item There exist some deterministic sequences $c_{1n}$ and $c_{2n}$ such
that $\sup_{i}\Vert \hat{\beta} _{in}-\beta _{g_{i}^{0}n}\Vert \leq c_{2n}
\leq M$ a.s. and $\inf_{1\leq k<k^{\prime }\leq K}\Vert \beta _{kn}-\beta
_{k^{\prime }n}\Vert \geq c_{1n}>0$.

\item $(2c_{2n}c_1^{1/2} + 16 K^{3/4} M^{1/2} c^{1/2}_{2n})^2 \leq
c_1c_{1n}^2.$
\end{enumerate}
\end{ass}

Assumption \ref{ass:theta}.1 requires that the centroids are uniformly
bounded. Assumption \ref{ass:theta}.2 requires that the centroids are
well-separated and the vectors to be classified (i.e., $\{\hat{\beta}_{in}\}$%
) are sufficiently close to one of the centroids. Assumption \ref{ass:theta}%
.3 requires that the distance between the estimated vector and the
corresponding centroid is smaller than that among any of the two distinctive
centroids. When the number of clusters $K$ is fixed and the gap $c_{1n}$
between the centroids is bounded away from zero, Assumption \ref{ass:theta}%
.3 holds as long as $c_{2n}$ is sufficiently small. Note here, we do not
necessarily need $c_{2n} = o(1)$, i.e., $\hat{ \beta}_{in}$ is not
necessarily consistent.

Let $H(\cdot ,\cdot )$ denote the Hausdorff distance between two sets and $%
\mathcal{B}_{n}=\{\beta _{1n},\ldots ,\beta _{Kn}\}.$ The following lemma
shows that the K-means algorithm can estimate the true centroids $\{\beta
_{kn}\}_{k=1}^{K}$ up to the rate $O_{a.s.}(c_{2n}^{1/2}K^{3/4}).$

\begin{lem}
\label{lem:kmeans1} Suppose that Assumptions \ref{ass:nk} and \ref{ass:theta}
hold. Then 
\begin{equation*}
H(\widehat{\mathcal{A}}_{n},\mathcal{B} _{n}) \leq
(15M/c_1)^{1/2}c_{2n}^{1/2}K^{3/4} \quad a.s.
\end{equation*}
\end{lem}

\begin{thm}
\label{thm:strong} Suppose that Assumptions \ref{ass:nk} and \ref{ass:theta}
hold. Then for sufficiently large $n$, we have 
\begin{equation*}
\sup_{1\leq i\leq n}\mathbf{1}\{\hat{g}_{i}\neq g_{i}^{0}\}=0\quad a.s.
\end{equation*}
\end{thm}

Theorem \ref{thm:strong} establishes that, under the given conditions, the
K-means algorithm yields perfect classification in large samples.
Intuitively, as long as the estimated vectors $\{\hat{\beta}_{in}\}_{i=1}^n$
are uniformly much closer to the true centroid $\beta_{g_i^0n}$ rather than
others, the K-means algorithm can divide each individual into the right
group. To achieve strong consistency for our SBM, we need the following
condition.

\begin{ass}
\label{ass:ratestrong} For $n$ sufficiently large, 
\begin{align*}
C^*\frac{K^{3/2}\rho _{n}\log ^{1/2}(n)}{\mu _{n}^{1/2}\sigma _{Kn}^{2}}%
\left(1+\rho_n + \left(\frac{1}{K} + \frac{\log(5)}{\log(n)}%
\right)^{1/2}\rho_n^{1/2}\right) \leq \frac{2c_1^{3/2}C_1^{-1}}{257},
\end{align*}
where $C^*$ is the absolute constant in Theorem \ref{thm:main}.
\end{ass}

\begin{cor}
\label{cor:sc} Suppose that Assumptions \ref{ass:id}--\ref{ass:rate} and \ref%
{ass:ratestrong} hold and the K-means algorithm is applied to $\hat{\beta}%
_{in}=(n/K)^{1/2}\hat{u}_{1i}$ and $\beta _{g_{i}^{0}n}=(n/K)^{1/2}\hat{O}%
_{n}u_{1i}$ Then, 
\begin{equation*}
\sup_{1\leq i\leq n}\mathbf{1}\{\hat{g}_{i}\neq g_{i}^{0}\}=0\quad a.s.
\end{equation*}
\end{cor}

Corollary \ref{cor:sc} shows that the spectral-clustering-based K-means
algorithm consistently recovers the community membership for all nodes
almost surely in large samples.

\begin{ex}[\textbf{2.1 (cont.)}]
For the four-parameter model in Example 2.1, Assumption \ref{ass:rate} is
equivalent to 
\begin{equation}
\frac{(p+r)^{4}K^{8}\log (n)}{p^{4}n(p+rK)}  \label{eq:fourparameter}
\end{equation}%
being sufficiently small. If $rK/p$ is bounded, then the above display
further reduces to $K^{8}\log (n)/\left( np\right)$, which allows $%
K=O((np/\log (n))^{1/8})$. As long as $p$ decays to zero no faster than $%
\log (n)/n$, Assumption \ref{ass:rate} holds even when $K$ grows slowly to
infinity. On the other hand, if $r/p\rightarrow c\in \left( 0,\infty \right),
$ \eqref{eq:fourparameter} reduces to $K^{7}\log (n)/\left( nr\right) $. In
addition, if both $p$ and $r$ are bounded away from zero, then %
\eqref{eq:fourparameter} requires that $K^{7}\log (n)/n$ is sufficiently
small. In contrast, \cite{RCY11} find that when $K=O\left( n^{1/4}/\log
\left( n\right) \right) $ and $p$ is bounded away from $0,$ the number of
misclassified nodes from the K-means algorithm in the four-parameter SBM is
of order $o\left( K^{3}\log ^{2}\left( n\right) \right) =o\left(
n^{3/4}\right) .$
\end{ex}

\subsection{Strong consistency of the modified K-means algorithm}

\label{sec:med} It is possible to improve the rate requirement for the
number of communities in Assumption \ref{ass:ratestrong} by considering a
modified K-means algorithm: 
\begin{equation*}
\widetilde{Q}_{n}(\mathcal{A})=\frac{1}{n}\sum_{i=1}^{n}\min_{1\leq l\leq
K}\Vert \hat{\beta}_{in}-\alpha _{l}\Vert
\end{equation*}%
and $\widetilde{\mathcal{A}}_{n}=\argmin_{\mathcal{A}}\widetilde{Q}_{n}(%
\mathcal{A})$, where $||\cdot||$ still denote the Euclidean distance. Denote 
$\widetilde{\mathcal{A}}$ as $\{\widetilde{\alpha}_1,\cdots,\widetilde{\alpha%
}_K\}$. Then, we compute the estimated cluster identity as 
\begin{equation*}
\tilde{g}_{i}=\argmin_{1\leq l\leq K}\Vert \hat{\beta}_{in}-\widetilde{%
\alpha }_{l}\Vert ,
\end{equation*}%
where if there are multiple $l$'s that achieve the minimum, $\tilde{g}_{i}$
takes value of the smallest one.

\begin{ass}
\label{ass:theta2}

\begin{enumerate}
\item There exist some deterministic sequences $c_{1n}$ and $c_{2n}$ such
that $\sup_{i}\Vert \hat{\beta} _{in}-\beta _{g_{i}^{0}n}\Vert \leq c_{2n}$
a.s. and $\inf_{1\leq k<k^{\prime }\leq K}\Vert \beta _{kn}-\beta
_{k^{\prime }n}\Vert \geq c_{1n}>0$.

\item $15Kc_{2n} \leq c_1c_{1n}.$
\end{enumerate}
\end{ass}

The following two results parallel Lemma \ref{lem:kmeans1} and Theorem \ref%
{thm:strong}.

\begin{lem}
\label{lem:kmed} Suppose that Assumptions \ref{ass:nk} and \ref{ass:theta2}
hold. Then 
\begin{equation*}
H(\widetilde{\mathcal{A}}_{n},\mathcal{B} _{n}) \leq 3Kc_1^{-1}c_{2n} \quad
a.s.
\end{equation*}
\end{lem}

\begin{thm}
\label{thm:strong2} Suppose that Assumptions \ref{ass:nk} and \ref%
{ass:theta2} hold. Then for sufficiently large $n$, we have 
\begin{equation*}
\sup_{1\leq i\leq n}\mathbf{1}\{\tilde{g}_{i}\neq g_{i}^{0}\}=0\quad a.s.
\end{equation*}
\end{thm}

In order to apply the modified K-means algorithm in spectral clustering, we
only need to verify conditions in Assumption \ref{ass:theta2}.

\begin{ass}
\label{ass:ratestrong2} Suppose there exists some constant $c^*$ such that,
for $n$ sufficiently large, 
\begin{align*}
15C^*\frac{K\rho _{n}\log ^{1/2}(n)}{\mu _{n}^{1/2}\sigma _{Kn}^{2}}%
\left(1+\rho_n + \left(\frac{1}{K} + \frac{\log(5)}{\log(n)}%
\right)^{1/2}\rho_n^{1/2}\right)\leq c_1C_1^{-1/2}\sqrt{2},
\end{align*}
where $C^*$ is the absolute constant in Theorem \ref{thm:main}.
\end{ass}

\begin{cor}
\label{cor:sc2} Suppose that Assumptions \ref{ass:id}--\ref{ass:rate} and %
\ref{ass:ratestrong2} hold and the K-means algorithm is applied to $\hat{%
\beta}_{in}=(n/K)^{1/2}\hat{u}_{1i}$ and $\beta _{g_{i}^{0}n}=(n/K)^{1/2}%
\hat{O}_{n}u_{1i}$ Then, 
\begin{equation*}
\sup_{1\leq i\leq n}\mathbf{1}\{\tilde{g}_{i}\neq g_{i}^{0}\}=0\quad a.s.
\end{equation*}
\end{cor}

Corollary \ref{cor:sc2} implies that the community memberships estimated by
the modified K-means can recover the truth. Assumption \ref{ass:ratestrong2}
implies a weaker requirement on the rate of $K$ than Assumption \ref%
{ass:ratestrong}, as the exponent for $K$ is reduced from 1.5 in Assumption %
\ref{ass:ratestrong} to 1 in Assumption \ref{ass:ratestrong2}. To derive the
optimal rate for $K$ may be much more difficult. We leave it as one topic
for future research. We investigate the performance of the 
K-means algorithm in Section \ref{sec:sim}.

Like spectral clustering, semidefinite programming (SDP) has also become
very popular in the community detection literature. Numerically, SDP
relaxation enjoys the computational feasibility that spectral clustering
has, and various efficient algorithms have been proposed to solve different
types of SDP. Theoretically, under the ordinary SBM, SDP methods have been
shown to be capable in detecting communities; see, \cite{ABH16}, \cite{A14}, 
\cite{BBV16}, \cite{c12}, \cite{c14}, \cite{CL15}, \cite{HWX16a}, and \cite%
{HWX16b}, among others, and \cite{LCX18} for an excellent survey. In
particular, \cite{ABH16} propose an efficient SDP\ algorithm to solve a
standard SBM with two communities, and show that it succeeds in recovering
the true communities with high probability when certain threshold conditions
are satisfied; \cite{CL15} propose a new SDP-based convex optimization
method for a generalized SBM and show that a SDP relaxation followed by a
K-means clustering can accurately detect the communities with small
misclassification rate and the method is both computationally fast and
robust to different kinds of outliers. In contrast, \cite{CL15} and \cite%
{JY16} show that the standard spectral clustering applied to the graph
Laplacian may not work due to the existence of small and weak clusters. The
possible presence of weak clusters in SBMs motivates the use of
regularization to be studied in the following section.

\section{Extensions}

\label{sec:ext} In this section we consider two extensions of the above
results: regularized spectral clustering of the standard and
degree-corrected SBMs.

\subsection{Regularized spectral clustering analysis for standard SBMs}

\label{sec:ext1} The SBM is the same as considered in the previous section.
Following \cite{ACBL13} and \cite{JY16}, we regularize the adjacency matrix $%
A$ to be $A_{\tau }=A+\tau n^{-1}\iota _{n}\iota _{n}^{T},$ where $\tau \leq
n$ is the regularization parameter and $\iota _{n}$ is the $n\times 1$
vector of ones. Given the regularized adjacency matrix, we can compute the
regularized degree for each node as $\hat{d}_{i}^{\tau }=\hat{d}_{i}+\tau $
and $D_{\tau }=\text{diag}(\hat{d}_{1}+\tau ,\ldots ,\hat{d}_{n}+\tau )$.
The regularized version of $P$ and $\mathcal{D}$ are denoted as $P_{\tau }$
and $\mathcal{D}_{\tau }$ and defined as 
\begin{equation*}
P_{\tau }=P+\tau n^{-1}\iota _{n}\iota _{n}^{T}\quad \text{and}\quad 
\mathcal{D}_{\tau }=\text{diag}(d_{1}+\tau ,\ldots ,d_{n}+\tau ),
\end{equation*}%
respectively. Consequently, the regularized graph Laplacian and its
population counterpart are denoted as $L_{\tau }$ and $\mathcal{L}_{\tau }$
and written as 
\begin{equation*}
L_{\tau }=D_{\tau }^{-1/2}A_{\tau }D_{\tau }^{-1/2}\quad \text{and}\quad 
\mathcal{L}_{\tau }=\mathcal{D}_{\tau }^{-1/2}P_{\tau }\mathcal{D}_{\tau
}^{-1/2},
\end{equation*}%
respectively. Noting that $\iota _{n}=Z\iota _{K},$ we have 
\begin{equation*}
P_{\tau }=P+\tau n^{-1}\iota _{n}\iota _{n}^{T}=ZBZ^{T}+\tau n^{-1}Z\iota
_{K}\iota _{K}^{T}Z^{T}=ZB^{\tau }Z^{T},
\end{equation*}%
where $B^{\tau }=B+\tau n^{-1}\iota _{K}\iota _{K}^{T}$. Apparently, the
block model structure is preserved after regularization. Given $B^{\tau }$,
we can define $B_{0}^{\tau }$, the normalized version of $B^{\tau }$ as in
the previous section. Let $W_{k}^{\tau }=[B^{\tau }]_{k\cdot }Z^{T}\iota
_{n}/n=\sum_{l=1}^{K}[B^{\tau }]_{{kl}}\pi _{ln}$, $\mathcal{D}_{B}^{\tau }=%
\text{diag}(W_{1}^{\tau },\ldots ,W_{K}^{\tau })$, and $B_{0}^{\tau }=(%
\mathcal{D}_{B}^{\tau })^{-1/2}B^{\tau }(\mathcal{D}_{B}^{\tau })^{-1/2}$.

In order to follow the identification analysis in the previous section, we
need to modify Assumption \ref{ass:id} as follows.

\begin{ass}
\label{ass:id2} Suppose $B_{0}^{\tau }$ has rank $K$ and the spectral
decomposition of $\Pi _{n}^{1/2}B_{0}^{\tau }\Pi _{n}^{1/2}$ is $S_{n}^{\tau
}\Omega _{n}^{\tau }(S_{n}^{\tau })^{T}$, in which $S_{n}^{\tau } $ is a $%
K\times K$ matrix such that $(S_{n}^{\tau })^{T}S_{n}^{\tau }=I_{K} $ and $%
\Omega _{n}^{\tau }=\text{diag} (\omega _{1n}^{\tau },\ldots ,\omega
_{Kn}^{\tau })$ such that $|\omega _{1n}^{\tau }|\geq \cdots \geq |\omega
_{Kn}^{\tau }|>0$.
\end{ass}

We consider the eigenvalue decomposition of $\mathcal{L}_{\tau }$ as 
\begin{equation*}
\mathcal{L}_{\tau }=U_{n}^{\tau }\Sigma _{n}^{\tau }(U_{n}^{\tau
})^{T}=U_{1n}^{\tau }\Sigma _{1n}^{\tau }(U_{1n}^{\tau })^{T}
\end{equation*}%
where $\Sigma _{n}^{\tau }=\text{diag}(\sigma _{1n}^{\tau },\ldots ,\sigma
_{Kn}^{\tau },0,\ldots ,0)$ is an $n \times n$ matrix that contains the
eigenvalues of $\mathcal{L}_{\tau }$ such that $|\sigma _{1n}^{\tau }|\geq
|\sigma _{2n}^{\tau }|\geq \cdots \geq |\sigma _{Kn}^{\tau }|>0$, $\Sigma
_{1n}^{\tau }=\text{diag}(\sigma _{1n}^{\tau },\ldots ,\sigma _{Kn}^{\tau })$%
, the columns of $U_{n}^{\tau }$ contain the eigenvectors of $\mathcal{L}%
_{\tau }$ associated with the eigenvalues in $\Sigma _{n}^{\tau }$, $%
U_{n}^{\tau }=(U_{1n}^{\tau },U_{2n}^{\tau })$, and $(U_{n}^{\tau
})^{T}U_{n}^{\tau }=I_{n}$.

The following theorem parallels Theorem \ref{thm:id} in Section \ref{sec:id}.

\begin{thm}
\label{thm:id2}If Assumptions \ref{ass:nk} and \ref{ass:id2} hold, then $%
\Omega _{n}^{\tau }=\Sigma _{n}^{\tau }$, $U_{1n}^{\tau
}=Z(Z^{T}Z)^{-1/2}S_{n}^{\tau }$ and 
\begin{equation*}
\sup_{1\leq i\leq n}(n/K)^{1/2}\Vert z_{i}^{T}(Z^{T}Z)^{-1/2}S_{n}^{\tau
}\Vert \leq c_1^{-1/2}.
\end{equation*}
In addition, there exists a constant $c$ independent of $n$ such that if $%
z_{i}\neq z_{j}$, 
\begin{equation*}
(n/K)^{1/2}\Vert (z_{i}^{T}-z_{j}^{T})(Z^{T}Z)^{-1/2}S_{n}^{\tau }\Vert \geq
C_1^{-1/2}\sqrt{2}>0.
\end{equation*}
\end{thm}

Since $\mathcal{L}_{\tau }=n^{-1}ZB_{0}^{\tau }Z$, the proof of Theorem \ref%
{thm:id2} is exactly the same as that of Theorem \ref{thm:id} with obvious
modifications. Theorem \ref{thm:id2} indicates that we can infer each node's
community membership based on the eigenvector matrix $U_{1n}^{\tau }$ if $%
\mathcal{L}_{\tau }$ is observed.

As before, we consider the spectral decomposition of $L_{\tau }:$ 
\begin{equation*}
L_{\tau }=\hat{U}_{n}^{\tau }\hat{\Sigma}_{n}^{\tau }(\hat{U}_{n}^{\tau
})^{T}=\hat{U}_{1n}^{\tau }\hat{\Sigma}_{1n}^{\tau }(\hat{U}_{1n}^{\tau
})^{T}+\hat{U}_{2n}^{\tau }\hat{\Sigma}_{2n}^{\tau }(\hat{U}_{2n}^{\tau
})^{T}.
\end{equation*}%
where $\hat{\Sigma}_{n}^{\tau }=\text{diag}(\hat{\sigma}_{1n}^{\tau },\ldots
,\hat{\sigma}_{nn}^{\tau })=\text{diag}(\hat{\Sigma}_{1n}^{\tau },\hat{\Sigma%
}_{2n}^{\tau })$ with $|\hat{\sigma}_{1n}^{\tau }|\geq |\hat{\sigma}%
_{2n}^{\tau }|\geq \cdots \geq |\hat{\sigma}_{nn}^{\tau }|\geq 0$, $\hat{%
\Sigma}_{1n}^{\tau }=\text{diag}(\hat{\sigma}_{1n}^{\tau },\ldots ,\hat{%
\sigma}_{Kn}^{\tau })$, and $\hat{\Sigma}_{2n}^{\tau }=\text{diag}(\hat{%
\sigma}_{K+1,n}^{\tau },\ldots ,\hat{\sigma}_{nn}^{\tau })$; $\hat{U}%
_{n}^{\tau }=(\hat{U}_{1n}^{\tau },\hat{U}_{2n}^{\tau })$ is the
corresponding eigenvectors such that $(\hat{U}_{1n}^{\tau })^{T}\hat{U}%
_{1n}=I_{K}$ and $\hat{U}_{2n}^{T}\hat{U}_{1n}=0$. Note that $\hat{U}%
_{1n}^{\tau }$ contains the eigenvectors associated with eigenvalues $\hat{%
\sigma}_{1n}^{\tau },\ldots ,\hat{\sigma}_{Kn}^{\tau }$. To study the
asymptotic properties of $\hat{U}_{1n}^{\tau }$, we modify Assumption \ref%
{ass:rate} as follows.

\begin{ass}
\label{ass:rate2} Denote $\mu _{n}^{\tau }=\min_{i}d_{i}+\tau $ and $\rho
_{n}^{\tau }=\max (\sup_{k_{1}k_{2}}[B_{0}^{\tau }]_{k_{1}k_{2}},1)$. Then,
for $n$ sufficiently large, 
\begin{align*}
\frac{K\rho^\tau _{n}\log ^{1/2}(n)}{(\mu ^\tau_{n})^{1/2}(\sigma^\tau
_{Kn})^{2}}\left(1+\rho^\tau_n + \left(\frac{1}{K} + \frac{\log(5)}{\log(n)}%
\right)^{1/2}(\rho^\tau_n)^{1/2}\right) \leq 10^{-8}C_1^{-1}c_1^{1/2}.
\end{align*}
\end{ass}

The above modification is natural because node $i$'s degree becomes $%
d_{i}^{\tau }\equiv d_{i}+\tau $ after regularization. $\mu _{n}^{\tau }$
can be interpreted as the effective minimum expected degree after
regularization.

Let $(u_{1i}^{\tau })^{T}$ and $(\hat{u}_{1i}^{\tau })^{T}$ be the $i$-th
row of $U_{1n}^{\tau }$ and $\hat{U}_{1n}^{\tau }$, respectively.

\begin{thm}
\label{thm:regularization} Suppose that Assumptions \ref{ass:nk}, \ref%
{ass:id2}, and \ref{ass:rate2} hold. Then there exists a $K\times K$
orthonormal matrix $\hat{O}_{n}^{\tau }$ such that 
\begin{equation*}
\sup_{1\leq i\leq n}\sqrt{n/K}\Vert (\hat{u}_{1i}^{\tau })^{T}\hat{O}%
_{n}^{\tau }-(u_{1i}^{\tau })^{T}\Vert \leq C^{\ast }\frac{\rho _{n}^{\tau
}\log ^{1/2}(n)}{(\mu _{n}^{\tau })^{1/2}(\sigma _{Kn}^{\tau })^{2}}\left(
1+\rho _{n}^{\tau }+\left( \frac{1}{K}+\frac{\log (5)}{\log (n)}\right)
^{1/2}(\rho _{n}^{\tau })^{1/2}\right) \quad a.s.,
\end{equation*}%
where $C^{\ast }$ is the same absolute constant defined in Theorem \ref%
{thm:main}.
\end{thm}

The following assumption parallels Assumptions \ref{ass:ratestrong} and \ref%
{ass:ratestrong2}. The following theorem parallels Theorem \ref{thm:main}.

\begin{ass}
\label{ass:ratestrong3}

\begin{enumerate}
\item For $n$ sufficiently large, 
\begin{align*}
C^*\frac{K^{3/2}\rho^\tau _{n}\log ^{1/2}(n)}{(\mu
^\tau_{n})^{1/2}(\sigma^\tau _{Kn})^{2}}\left(1+\rho^\tau_n + \left(\frac{1}{%
K} + \frac{\log(5)}{\log(n)}\right)^{1/2}(\rho^\tau_n)^{1/2}\right) \leq 
\frac{2c_1^{3/2}C_1^{-1}}{257},
\end{align*}
where $C^*$ is the absolute constant in Theorem \ref{thm:regularization}.

\item For $n$ sufficiently large, 
\begin{align*}
15C^*\frac{K\rho^\tau _{n}\log ^{1/2}(n)}{(\mu ^\tau_{n})^{1/2}(\sigma^\tau
_{Kn})^{2}}\left(1+\rho^\tau_n + \left(\frac{1}{K} + \frac{\log(5)}{\log(n)}%
\right)^{1/2}(\rho^\tau_n)^{1/2}\right) \leq c_1C_1^{-1/2}\sqrt{2},
\end{align*}
where $C^*$ is the absolute constant in Theorem \ref{thm:regularization}.
\end{enumerate}
\end{ass}

The following theorem parallels Corollaries \ref{cor:sc} and \ref{cor:sc2}
in Section \ref{sec:unif}.

\begin{thm}
\label{thm:regularization2} Suppose that Assumptions \ref{ass:nk}, \ref%
{ass:id2}, and \ref{ass:rate2} hold. If Assumption \ref{ass:ratestrong3}.1
holds and the K-means algorithm defined in Section \ref{sec:kmeans} is
applied to $\hat{\beta}_{in}=\sqrt{n/K}(\hat{u}_{1i}^{\tau })^{T}$ and $%
\beta _{g_{i}^{0}n}=(n/K)^{1/2}\hat{O}_{n}^{\tau }u_{1i}^{\tau }$. Denote
the estimated community identities as $\{\hat{g}_i\}_{i=1}^n$. Then for
sufficiently large $n,$ we have 
\begin{equation*}
\sup_{1\leq i\leq n}1\{\hat{g}_{i}\neq g_{i}^{0}\}=0\quad a.s.
\end{equation*}
If Assumption \ref{ass:ratestrong3}.2 holds and the modified K-means
algorithm defined in Section \ref{sec:med} is applied to $\hat{\beta}_{in}=%
\sqrt{n/K}(\hat{u}_{1i}^{\tau })^{T}$ and $\beta _{g_{i}^{0}n}=(n/K)^{1/2}%
\hat{O}_{n}^{\tau }u_{1i}^{\tau }$. Denote the estimated community
identities as $\{\tilde{g}_i\}_{i=1}^n$. Then, for sufficiently large $n,$
we have 
\begin{equation*}
\sup_{1\leq i\leq n}1\{\tilde{g}_{i}\neq g_{i}^{0}\}=0\quad a.s.
\end{equation*}
\end{thm}

As in the standard SBM case, $\hat{O}_{n}^{\tau }=\bar{U}^{\tau }(\bar{V}%
^{\tau })^{T},$ where $\bar{U}^{\tau }\bar{\Sigma}^{\tau }(\bar{V}^{\tau
})^{T}$ is the singular value decomposition of $(\hat{U}_{1n}^{\tau
})^{T}U_{1n}^{\tau }.$Theorem \ref{thm:regularization2} indicates that the
regularized spectral clustering, in conjunction with the standard or
modified K-means algorithm, consistently recovers the community membership
for all nodes almost surely in large samples.

To see the effect of regularization, let $K$ be fixed and $|\sigma
_{Kn}^{\tau }|$ be bounded away from zero. Then, Assumption \ref{ass:rate2}
boils down to $\log (n)/\mu _{n}^{\tau }\leq \underline{c}$ for some
sufficiently small $\underline{c}$. Even if $\min_{i}d_{i}$ grows slower
than $\log (n)$ or does not grow to infinity at all, we can still choose $%
\tau $ with $\tau /\log (n)=\Omega (1)$ such that Assumption \ref{ass:rate2}
holds. This implies that we can obtain strong consistency for some SBMs in
which some nodes have very limited number of links.

In addition, regularization introduces a trade-off between $|\sigma
_{Kn}^{\tau }|$ and $\mu _{n}^{\tau }$. As $\tau $ increases, $\mu
_{n}^{\tau }$ increases and the rows of $B_{0}^{\tau }$ become more similar,
which means that $|\sigma _{Kn}^{\tau }|$ decreases. \cite{RCY11} and \cite{JY16} explore such intuition to choose the
regularizer. Following their leads, we choose over a grid of $\tau $ and
find the one that minimizes 
\begin{equation*}
Q(\tau )\equiv ||L_{\tau }-\hat{\mathcal{L}}_{\tau }||/|\hat{\sigma}%
_{Kn}^{\tau }|,
\end{equation*}%
where $\hat{\mathcal{L}}_{\tau }$ is an estimator of $\mathcal{L}_{\tau }$.
We refer to our Section \ref{sec:sim} for more details.

The following is a non-trivial SBM which does not satisfy Assumption \ref%
{ass:rate} but satisfies Assumption \ref{ass:rate2}.

\begin{ex}[3.1]
Consider a SBM with two groups such that $n_{1}=n_{2}=n/2$ and 
\begin{equation*}
B=%
\begin{pmatrix}
0.4 & 2/n \\ 
2/n & 4/n%
\end{pmatrix}%
.
\end{equation*}%
\textit{In this case, }$d_{i}=0.4(\frac{n}{2}-1)+\frac{2}{n}\cdot \frac{n}{2}%
=0.2n+0.6$ for node $i$ in cluster 1 and $d_{i}=\frac{2}{n}\cdot \frac{n}{2}+%
\frac{4}{n}(\frac{n}{2}-1)=3-\frac{4}{n}$ for node $i$ in cluster 2.
Therefore, Assumption \ref{ass:rate} does not hold. However, for some $\tau $
such that $\tau =\Omega (\log (n))$, we have 
\begin{equation*}
B^{\tau }=%
\begin{pmatrix}
0.4+\tau /n & (2+\tau )/n \\ 
(2+\tau )/n & (4+\tau )/n%
\end{pmatrix}%
\end{equation*}%
\textit{and }$d_{i}^{\tau }=0.2n+0.6+\tau (1-n^{-1})$\textit{\ for node }$i$
\ in cluster 1 and $d_{i}^{\tau }=3-4n^{-1}+\tau (1-n^{-1})$ for node $i$ in
cluster 2. In addition, it is easy to see that 
\begin{equation*}
B_{0}^{\tau }=%
\begin{pmatrix}
\frac{0.4+\tau n^{-1}}{0.2+(1+\tau )n^{-1}} & \frac{2+\tau }{[0.2n+(1+\tau
)]^{1/2}(3+\tau )^{1/2}} \\ 
\frac{2+\tau }{[0.2n+(1+\tau )]^{1/2}(3+\tau )^{1/2}} & \frac{4+\tau }{%
3+\tau }%
\end{pmatrix}%
\rightarrow 
\begin{pmatrix}
\frac{0.4+c_{0}}{0.2+c_{0}} & \sqrt{\frac{c_{0}}{0.2+c_{0}}} \\ 
\sqrt{\frac{c_{0}}{0.2+c_{0}}} & 1%
\end{pmatrix}%
,
\end{equation*}%
when $c_{0}=\lim_{n\rightarrow \infty }\tau /n\in \lbrack 0,1).$ Apparently, 
$B_{0}^{\tau }$ has full rank and Assumption \ref{ass:rate2} holds.
Therefore, the strong consistency of the regularized spectral clustering
still holds.

Let $\sigma _{2,n}^{\tau }$ denote the second eigenvalue of $\Pi
_{n}^{1/2}B_{0}^{\tau }\Pi _{n}^{1/2}$. Then as $n\rightarrow \infty ,$ 
\begin{equation*}
\sigma _{2,n}^{\tau }\rightarrow \frac{0.3+c_{0}-\sqrt{%
c_{0}^{2}+0.2c_{0}+0.01}}{2(c_{0}+0.2)}=\frac{0.2}{0.3+c_{0}+\sqrt{%
c_{0}^{2}+0.2c_{0}+0.01}},
\end{equation*}%
where $c_{0}\in \lbrack 0,1)$. The minimal degree $\mu _{n}^{\tau }\asymp
\tau $. Then, $Q(\tau )=O(\frac{1}{\sigma _{2,n}^{\tau }(\mu _{n}^{\tau
})^{1/2}})$ where 
\begin{equation*}
\frac{1}{\sigma _{2,n}^{\tau }(\mu _{n}^{\tau })^{1/2}}\asymp \frac{%
0.3+c_{0}+\sqrt{c_{0}^{2}+0.2c_{0}+0.01}}{0.2\tau ^{1/2}}.
\end{equation*}%
In order to achieve maximal convergence rate, we need $c_{0}\neq 0$. For
simplicity, we just assume $\tau =c_{0}n$. Then, the constant $\frac{%
0.3+c_{0}+\sqrt{c_{0}^{2}+0.2c_{0}+0.01}}{c_{0}^{1/2}}$ achieves minimum on $%
(0,1)$ at $c_{0}=0.2.$
\end{ex}

The previous example illustrates that the regularization works for the case
where one cluster has strong links and the other one has weak links. However,
if both clusters have weak links, it is hard to separate them.

\begin{ex}[3.2]
Consider the above example with $B$ replaced by 
\begin{equation*}
B=%
\begin{pmatrix}
4/n & 2/n \\ 
2/n & 4/n%
\end{pmatrix}%
,
\end{equation*}%
and $\tau /\log (n)=\Omega (1)$. Then we can verify that 
\begin{equation*}
B_{0}^{\tau }=%
\begin{pmatrix}
(4+\tau )/(3+\tau ) & (2+\tau )/(3+\tau ) \\ 
(2+\tau )/(3+\tau ) & (4+\tau )/(3+\tau )%
\end{pmatrix}%
\end{equation*}%
such that $B_{0}^{\tau }$ has two eigenvalues given by $2$ and $2/\left(
3+\tau \right) $. But Assumption \ref{ass:rate2} cannot be satisfied in this
case because $\mu _{n}^{\tau }|\sigma _{Kn}^{\tau }|^{4}/\log (n)$ is
converging to zero at rate $1/(\tau ^{3}\log (n))$. Consequently, we cannot
show that $\sup_{i}\sqrt{n}\Vert (\hat{O}_{n}^{\tau })^{T}\hat{u}_{1i}^{\tau
}-u_{1i}^{\tau }\Vert $ is sufficiently small or prove strong consistency in
this case.
\end{ex}

The above example shows that the regularization may not work for the case in
which we have multiple clusters with weak links.

\subsection{Regularized spectral clustering analysis for degree-corrected
SBMs}

\label{sec:ext2} In this subsection, we extend our early analyses to the
spectral clustering for a degree-corrected stochastic block model (DC-SBM).

\subsubsection{Degree-corrected SBMs}

Since \cite{KN11}, degree-corrected SBMs have become widely used in
communication detection. The major advantage of a DC-SBM lies in the fact
that it allows variation in node degrees within a community while preserving
the overall block community structure. Given the $K$ communities, the edge
between nodes $i$ and $j$ are chosen independently with probability
depending on the communities that nodes $i$ and $j$ belong to. In
particular, for nodes $i$ and $j$ belonging to clusters $C_{k_{1}}$ and $%
C_{k_{2}}$, respectively, the probability of edge between $i$ and $j$ is
given by 
\begin{equation*}
P_{ij}=\theta _{i}\theta _{j}B_{k_{1}k_{2}},
\end{equation*}%
where the block probability matrix $B=\{B_{k_{1}k_{2}}\}$, $%
k_{1},k_{2}=1,\ldots ,K$, is a symmetric matrix with each entry between $%
[0,1]$. The $n\times n$ edge probability matrix $P=\{P_{ij}\}$ represents
the population counterpart of the adjacency matrix $A$. We continue to use $%
Z=\{Z_{ik}\}$ to denote the cluster membership matrix for all $n$ nodes. Let 
$\Theta =\text{diag}(\theta _{1},\ldots ,\theta _{n})$. Then we have 
\begin{equation*}
P=\Theta ZBZ^{T}\Theta ^{T}.
\end{equation*}%
Note $\Theta $ and $B$ are only identifiable up to scale. We adopt the
following normalization rule: 
\begin{equation}
\sum_{i\in C_{k}}\theta _{i}=n_{k},\quad k=1,\ldots ,K.
\label{eq:thetanormalization}
\end{equation}%
Alternatively, one can follow the literature (e.g., \citep{QR13,ZLZ12}) and
apply the following normalization $\sum_{i\in C_{k}}\theta _{i}=1,\quad
k=1,\ldots ,K.$ We use the normalization in \eqref{eq:thetanormalization}
because it nests the standard SBM as a special case when $\theta _{i}=1$ for 
$i=1,\ldots ,n$.

We first observe that, if we regularize both the adjacency matrix $A$ and
the degree matrix $D$, we are unable to preserve the DC-SBM structure unless 
$\Theta $ is homogeneous. To see this, note that when $A$ is regularized to $%
A_{\tau }=A+\tau n^{-1}\iota _{n}\iota _{n}^{T},$ its population counterpart
is 
\begin{equation*}
P_{\tau }=P+\tau n^{-1}\iota _{n}\iota _{n}^{T}=\Theta ZBZ^{T}\Theta +\tau
n^{-1}Z\iota _{k}\iota _{k}^{T}Z.
\end{equation*}%
Since $\Theta $ does not have the block structure, we are unable to find a $%
K\times K$ matrix $B^{\tau }$ and an $n\times n$ diagonal matrix $\Theta
^{\tau }$ such that $P_{\tau }=\Theta ^{\tau }ZB^{\tau }Z^{T}\Theta ^{\tau
}. $ For this reason, we follow the lead of \cite{QR13} and only regularize
the degree matrix $D$ as $D_{\tau }=D+\tau I_{n}$. To differentiate from the
regularized graph Laplacian $L_{\tau }$ considered in \cite{JY16}, we denote
the new regularized graph Laplacian as 
\begin{equation*}
L_{\tau }^{\prime }=D_{\tau }^{-1/2}AD_{\tau }^{-1/2},
\end{equation*}%
and its population counterpart as 
\begin{equation*}
\mathcal{L}_{\tau }^{\prime }=\mathcal{D}_{\tau }^{-1/2}P\mathcal{D}_{\tau
}^{-1/2},
\end{equation*}%
where $P=\Theta ZBZ^{T}\Theta ,$ $\mathcal{D}_{\tau }=\mathcal{D}+\tau
I_{n}, $ and $\mathcal{D}=\text{diag}(d_{1},\ldots ,d_{n})$ with $%
d_{i}=\sum_{j=1}^{n}P_{ij}$.

\subsubsection{Identification of the group membership}

Let $\pi _{kn},$ $W_{k},$ $\mathcal{D}_{B}$ and $B_{0}$ be as defined in
Section \ref{sec:id}. To facilitate the asymptotic study, we assume the
following:

\begin{ass}
\label{ass:id3}

\begin{enumerate}
\item There exists a sequence $\rho _{n}$ such that $\rho _{n}\geq 1$ and $%
B_{0}\leq \rho _{n}$ element-wise.

\item $B_{0}$ has full rank $K$.
\end{enumerate}
\end{ass}

As before, we consider the spectral decomposition of $\mathcal{L}_{\tau
}^{\prime }:$ 
\begin{equation*}
\mathcal{L}_{\tau }^{\prime }=U_{1n}\Sigma _{n}U_{1n}^{T},
\end{equation*}%
where $\Sigma _{n}=\text{diag}(\sigma _{1n},\ldots ,\sigma _{Kn})$ is a $%
K\times K$ matrix that contains the eigenvalues of $\mathcal{L}_{\tau
}^\prime$ such that $|\sigma _{1n}|\geq |\sigma _{2n}|\geq \cdots \geq
|\sigma _{Kn}|>0 $ and $U_{1n}^{T}U_{1n}=I_{K}$. Note that we suppress the
dependence of $U_{1n}$ and $\Sigma _{n}$ on $\tau .$ Let $\Theta _{\tau }=%
\text{diag}(\theta _{1}^{\tau },\ldots ,\theta _{n}^{\tau })$ where $\theta
_{i}^{\tau }=\theta _{i}d_{i}/(d_{i}+\tau )$ for $i=1,\ldots ,n$. Let $%
n_{k}^{\tau }=\sum_{i\in C_{k}}\theta _{i}^{\tau }.$

\begin{thm}
\label{thm:id3} Suppose Assumptions \ref{ass:id3} holds and let $g_{i}^{0}$
and $u_{i}^{T}$ be the node $i$'s true community identity and the $i$-th row
of $U_{1n}$, respectively. Then, (1) there exists a $K\times K$ matrix $%
S_{n}^{\tau }$ such that $U_{1n}=\Theta _{\tau }^{1/2}Z(Z^{T}\Theta _{\tau
}Z)^{-1/2}S_{n}^{\tau }$, (2) $(n_{g_{i}^{0}}^{\tau })^{1/2}(\theta
_{i}^{\tau })^{-1/2}\Vert u_{i}^{T}\Vert =1$, and (3) if $z_{i}=z_{j}$, then 
$\Vert \frac{u_{i}}{\Vert u_{i}\Vert }-\frac{u_{j}}{\Vert u_{j}\Vert } \Vert
=0;$ if $z_{i}\neq z_{j},$ then $\Vert \frac{u_{i}^{T}}{\Vert u_{i}^{T}\Vert 
}-\frac{u_{j}^{T}}{\Vert u_{j}^{T}\Vert }\Vert =\sqrt{2}$.
\end{thm}

Theorem \ref{thm:id3} follows \citet[][Lemma 3.3]{QR13}. In particular,
Theorem \ref{thm:id3}(3) provides useful facts about the rows of $U_{1n}.$
First, if two nodes $i$ and $j$ belong to the same cluster, then the
corresponding rows of $U_{1n}$ point to the same direction so that $%
u_{i}/\Vert u_{i}\Vert =u_{j}/\Vert u_{j}\Vert .$ Second, if two nodes $i$
and $j$ belong to the different clusters, then the corresponding rows of $%
U_{1n}$ are orthogonal to each other. As a result, we can detect the
community membership based on a feasible version of $\{u_{i}/\Vert
u_{i}\Vert \}.$

\subsubsection{Uniform consistency of the estimated eigenvectors and strong
consistency of the spectral clustering}

To proceed, we add the following assumptions.

\begin{ass}
\label{ass:nk3} There exist two constants $C_1$ and $c_1$ such that 
\begin{equation*}
\infty >C_1\geq \limsup_{n}\sup_{1\leq i\leq n}n_{g_{i}^{0}}^{\tau
}d_{i}^{\tau }K/(nd_{i})\geq \liminf_{n}\inf_{1\leq i\leq
n}n_{g_{i}^{0}}^{\tau }d_{i}^{\tau }K/(nd_{i})\geq c_1>0.
\end{equation*}
\end{ass}

Assumption \ref{ass:nk3} holds for the simplest case in which the degrees
are homogeneous within the same cluster. Note that in this case, $%
n_{g_{i}^{0}}^{\tau }=n_{g_{i}^{0}}d_{i}/d_{i}^{\tau }$, which may be of
smaller order of magnitude of $n/K$ if $d_{i}/\tau \rightarrow 0$. However,
Assumption \ref{ass:nk3} still holds because the factor $d_{i}/d_{i}^{\tau }$
is removed. In general, Assumption \ref{ass:nk3} holds if $d_{i}$ is of the
same order of magnitude for all $i$ in the same cluster.

\begin{ass}
\label{ass:rate3}Denote $\mu _{n}=\min_{i}d_{i}$, $\mu _{n}^{\tau }=\mu
_{n}+\tau $, $\overline{\theta }=\max_{i}\theta _{i}$, and $\underline{%
\theta }=\min_{i}\theta _{i}$. Then, for n sufficiently large,

\begin{enumerate}
\item $\frac{\bar{\theta}^{1/2}\log^{1/2}(n)}{\underline{\theta}%
^{1/2}(\mu_{n}^\tau)^{1/2}\rho_n} \leq 10^{-4},$

\item 
\begin{align*}
\biggl(K\frac{\rho _{n}\log ^{1/2}(n)}{\left( \mu _{n}^{\tau }\right)
^{1/2}\sigma _{Kn}^{2}}\biggr)\biggl(\frac{\left(\frac{1}{K} + \frac{\log(5)%
}{\log(n)}\right)^{1/2}\rho _{n}^{1/2}\overline{\theta }^{1/4}}{\underline{%
\theta }^{1/4}}+\rho _{n}+1\biggr) \leq 10^{-8}C_1^{-1}c_1^{1/2}, \quad 
\text{and}
\end{align*}

\item there exists a positive constant $c$ such that $\underline{\theta }
\geq n^{-c}$.
\end{enumerate}
\end{ass}

Assumption \ref{ass:rate3} specifies conditions on $d_{i},$ $\theta _{i},$
and $\sigma _{Kn}.$ The same remarks after Assumption \ref{ass:rate} apply.
Admittedly, the constants in Assumption \ref{ass:rate3} are not optimal. We
choose them purely for technical ease. If $0<\underline{\theta }\leq 
\overline{\theta }<\infty $, then Assumption \ref{ass:rate3}.1 is nested by
Assumption \ref{ass:rate3}.2, which is similar to Assumption \ref{ass:rate}.
If in addition, $K$ is fixed and $\liminf_{n}|\sigma _{Kn}|>0$, then
Assumption \ref{ass:rate3}.2 further boils down to $\log (n)/\mu _{n}^{\tau
}\leq \underline{c}$ for some sufficiently small $\underline{c}$. This
indicates that even if the minimal degree $\mu _{n}$ is bounded, Assumption %
\ref{ass:rate3}.2 still holds if $\tau =\Omega (\log (n))$.

Consider the spectral decomposition of $L_{\tau }^{\prime }$, the sample
counterpart of $\mathcal{L}_{\tau }^{\prime }$, as 
\begin{equation*}
L_{\tau }^{\prime }=\hat{U}_{n}\hat{\Sigma}_{n}\hat{U}_{n}^{T}=\hat{U}_{1n}%
\hat{\Sigma}_{1n}\hat{U}_{1n}^{T}+\hat{U}_{2n}\hat{\Sigma}_{2n}\hat{U}%
_{2n}^{T},
\end{equation*}%
where $\hat{\Sigma}_{n}=\text{diag}(\hat{\sigma}_{1n},\ldots ,\hat{\sigma}%
_{nn})=\text{diag}(\hat{\Sigma}_{1n},\hat{\Sigma}_{2n})$ with $|\hat{\sigma}%
_{1n}|\geq |\hat{\sigma}_{2n}|\geq \cdots \geq |\hat{\sigma}_{nn}|\geq 0,$ $%
\hat{\Sigma}_{1n}=\text{diag}(\hat{\sigma}_{1n},\ldots ,\hat{\sigma}_{Kn})$, 
$\hat{\Sigma}_{2n}=\text{diag}(\hat{\sigma}_{K+1,n},\ldots ,\hat{\sigma}%
_{nn})$, and $\hat{U}_{n}=(\hat{U}_{1n},\hat{U}_{2n})$ is the corresponding
eigenvectors such that $\hat{U}_{1n}^{T}\hat{U}_{1n}=I_{K}$ and $\hat{U}%
_{2n}^{T}\hat{U}_{1n}=0.$

The following lemma parallels Lemma \ref{lem:dk}.

\begin{lem}
\label{lem:dk3}If Assumptions \ref{ass:id3}--\ref{ass:rate3} hold, then 
\begin{equation*}
\Vert \mathcal{L}_{\tau }^{\prime }-L_{\tau }^{\prime }\Vert \leq 7(\log
(n)/\mu _{n}^{\tau })^{1/2}\quad a.s.
\end{equation*}%
and 
\begin{equation*}
\Vert \hat{U}_{1n}\hat{O}_{n}-U_{1n}\Vert \leq 10(\log (n)/\mu _{n}^{\tau
})^{1/2}|\sigma _{Kn}|^{-1}\quad a.s.,
\end{equation*}%
where $\hat{O}_{n}=\bar{U}\bar{V}^{T}$ is a $K\times K$ orthogonal matrix
and $\bar{U}\bar{\Sigma}\bar{V}^{T}$ for some diagonal matrix $\bar{\Sigma}$
is the singular value decomposition of $\hat{U}_{1n}^{T}U_{1n}.$
\end{lem}

In order to obtain the strong consistency, we need to derive the uniform
bound for $\Vert \hat{u}_{i}^{T}\hat{O}_{n}-u_{i}^{T}\Vert $, where $\hat{u}%
_{i}^{T}$ and $u_{i}^{T}$ are the $i$-th rows of $\hat{U}_{1n}$ and $U_{1n}$%
, respectively.

\begin{thm}
\label{thm:main_DC} If Assumptions \ref{ass:id3}--\ref{ass:rate3} hold, then 
\begin{equation*}
\sup_{i}(n_{g_{i}^{0}}^{\tau })^{1/2}(\theta _{i}^{\tau })^{-1/2}\Vert \hat{u%
}_{i}^{T}\hat{O}_{n}-u_{i}^{T}\Vert \leq C^*\eta _{n}\quad a.s.,
\end{equation*}%
where $C^*$ is an absolute constant specified in the proof and 
\begin{equation*}
\eta _{n}=\biggl(\frac{\rho _{n}\log ^{1/2}(n)}{\left( \mu _{n}^{\tau
}\right) ^{1/2}\sigma _{Kn}^{2}}\biggr)\biggl(\frac{\left(\frac{1}{K} + 
\frac{\log(5)}{\log(n)}\right)^{1/2}\rho _{n}^{1/2}\overline{\theta }^{1/4}}{%
\underline{\theta }^{1/4}}+\rho _{n}+1\biggr).
\end{equation*}
\end{thm}

Theorem \ref{thm:main_DC} is essential to establish the strong consistency
result. The following Assumption specifies the rate requirement for strong
consistency depending on whether the standard or modified K-means algorithm
is used.

\begin{ass}
\label{ass:ratestrong4}Let $C^{\ast }$ denote the absolute constant in
Theorem \ref{thm:main_DC}. For $n$ sufficiently large we have

\begin{enumerate}
\item $C^{\ast }K^{3/2}\eta _{n}\leq \frac{c_{1}}{257},$

\item $30C^{\ast }K\eta _{n}\leq c_{1}\sqrt{2}$.
\end{enumerate}
\end{ass}

\begin{cor}
\label{cor:DC} If Assumptions \ref{ass:id3}--\ref{ass:rate3} hold, then 
\begin{equation}
\sup_{i}\biggl\Vert\frac{\hat{u}_{i}}{\Vert \hat{u}_{i}\Vert }-\frac{\hat{O}%
_{n}u_{i}}{\Vert \hat{O}_{n}u_{i}\Vert }\biggr\Vert\leq 2C^*\eta_n \quad a.s.
\label{eq:uhat}
\end{equation}%
If Assumption \ref{ass:ratestrong4}.1 holds and the K-means algorithm is
applied to $\hat{\beta}_{in}=\hat{u}_{1i}/\Vert \hat{u}_{1i}\Vert $ and $%
\beta _{g_{i}^{0}n}=\hat{O}_{n}u_{1i}/\Vert u_{1i}\Vert $. Denote the
obtained community memberships as $\{\hat{g}_i\}_{i=1}^n$. Then, 
\begin{equation*}
\sup_{1\leq i\leq n}\mathbf{1}\{\hat{g}_{i}\neq g_{i}^{0}\}=0\quad a.s.
\end{equation*}
If Assumption \ref{ass:ratestrong4}.2 holds and the modified K-means
algorithm is applied to $\hat{\beta}_{in}=\hat{u}_{1i}/\Vert \hat{u}%
_{1i}\Vert $ and $\beta _{g_{i}^{0}n}=\hat{O}_{n}u_{1i}/\Vert u_{1i}\Vert $.
Denote the obtained community memberships as $\{\tilde{g}_i\}_{i=1}^n$.
Then, 
\begin{equation*}
\sup_{1\leq i\leq n}\mathbf{1}\{\tilde{g}_{i}\neq g_{i}^{0}\}=0\quad a.s.
\end{equation*}
\end{cor}

Corollary \ref{cor:DC} justifies the use of standard and modified K-means
algorithms on $\hat{u}_{in}/\Vert \hat{u}_{in}\Vert $ provided the bound on
the right hand side of \eqref{eq:uhat} is $O\left( 1/K^{3/2}\right)$ and $%
O(K)$, respectively, which is ensured by Assumptions \ref{ass:ratestrong4}.1
and \ref{ass:ratestrong4}.2, respectively.

\subsubsection{An adaptive procedure}

\label{sec:adaptive} Given the strong consistency of the spectral
clustering, it is possible to consistently estimate $\Theta$ by some
estimator, namely $\hat{\Theta}$. Built upon $\hat{\Theta}$, we propose an
adaptive procedure by spectral clustering a new regularized graph Laplacian
denoted as $L_\tau^{\prime \prime }$, which is defined as 
\begin{equation*}
L_\tau^{\prime \prime } = (D_\tau^{\prime \prime })^{-1/2} A_\tau^{\prime
\prime } (D_\tau^{\prime \prime })^{-1/2},
\end{equation*}
where $A_\tau^{\prime \prime } = A+ \tau n^{-1} \hat{\Theta}\iota_n
\iota_n^T \hat{\Theta}$ and $D_\tau^{\prime \prime } = \text{diag}%
(A_\tau^{\prime \prime } \iota_n)$. The population counterpart of $%
L_\tau^{\prime \prime }$ is denoted as $\mathcal{L}_\tau^{\prime \prime }$
and defined as 
\begin{equation*}
\mathcal{L}_\tau^{\prime \prime } = (\mathcal{D}_\tau^{\prime \prime
})^{-1/2}P_\tau^{\prime \prime } (\mathcal{D}_\tau^{\prime \prime })^{-1/2},
\end{equation*}
where $P_\tau^{\prime \prime } = P + \tau n^{-1} \Theta\iota_n \iota_n^T
\Theta = \Theta Z B_\tau^{\prime \prime }Z^T \Theta$, $B_\tau^{\prime \prime
} = B + \tau n^{-1} \iota_k \iota_k^T$, and $\mathcal{D}_\tau^{\prime \prime
} = \text{diag}(P_\tau^{\prime \prime }\iota_n) = D+\tau \Theta$.

Provided $\hat{\Theta}$ is consistent, we conjecture that one can show the
adaptive procedure is strongly consistent by applying the same proof
strategy as used in the derivation of strong consistency of the spectral
clustering based on $L_{\tau }$ and $L_{\tau }^{\prime }$. We leave this
important extension for future research. In the following, we focus on
establishing the consistency of $\hat{\Theta}$.

Given the estimated group membership $\{\hat{g}_i\}_{i=1}^n$, we follow \cite%
{WSW2016} and estimate $\Theta$ by $\hat{\Theta} = \text{diag}(\hat{\theta}%
_1,\cdots,\hat{\theta}_n)$, where 
\begin{equation}
\hat{\theta}_{i}=\hat{n}_{\hat{g}_{i}}(\sum\nolimits_{j=1}^{n}A_{ij})/(\sum%
\nolimits_{i^{\prime }:\hat{g}_{i^{\prime }}=\hat{g}_{i}}\sum%
\nolimits_{j=1}^{n}A_{i^{\prime }j})  \label{eq:theta}
\end{equation}%
and $\hat{n}_{k}=\#\{i:\hat{g}_{i}=k\}$. Next, we show $\hat{\theta}_{i}
\rightarrow \theta_i$ a.s. uniformly in $i = 1,\cdots,n$.

\begin{ass}
\begin{enumerate}
\item $\limsup_{n}\overline{\theta }<\infty .$

\item $\sup_{1\leq i\leq n}\mathbf{1}\{\hat{g}_{i}\neq g_{i}^{0}\}=0\quad
a.s.$
\end{enumerate}

\label{ass:adaptive_rate}
\end{ass}

Assumption \ref{ass:adaptive_rate}.1 requires that the degree of
heterogeneity is bounded, which is common in practical applications.
Assumption \ref{ass:adaptive_rate}.2 requires the preliminary clustering is
strongly consistent. For instance, this assumption can be verified by
Corollary \ref{cor:DC}. However, we also allow for any other strongly
consistent clustering methods, such as the conditional pseudo likelihood
method proposed by \cite{ACBL13}.

Let $m_{k}=\sum_{j=1}^{n}\theta _{j}B_{kg_{j}^{0}}$ and $\underline{m}%
_{n}=\inf_{k}m_{k}$. Note $m_{k}=\sum_{i^{\prime }\in C_{k}}d_{i^{\prime
}}/n_{k}$ is the average degree of nodes in community $k$ and \underline{$m$}%
$_{n}$ is the minimal average degree.

\begin{thm}
If Assumption \ref{ass:adaptive_rate} holds, then $\sup_{1\leq i\leq n}|\hat{%
\theta}_{i}-\theta _{i}|=O_{a.s.}(\log (n)/\underline{m}_{n})$. \label%
{thm:thetahat}
\end{thm}

In order for $\hat{\Theta}$ to be consistent, we need the average degree for
each community to grow faster than $\log (n)$. In some cases, the average
degree and the minimal degree are of the same order of magnitude. Then we
basically need $\mu _{n}/\log (n)\rightarrow \infty $ for the consistency of 
$\hat{\Theta}$. In our simulation designs, $\mu _{n}/\log (n)\rightarrow 0$,
which is, in some sense, the worst case for the adaptive procedure. However,
even in this case, the performance of the adaptive procedure improves upon
that of the spectral clustering based on $L_{\tau }^{\prime }$.

\section{Numerical Examples on Simulated Networks}

\label{sec:sim} In this section, we consider the finite sample performance
of spectral clustering with two and three communities, i.e., $K=2$ and $K=3$%
. The corresponding numbers of community members have ratio $1:1$ and $1:1:1$
for these two cases, respectively. The number of nodes is given by 50 and
200 for each community, which indicates $n=100$ and $400$ for the case of $%
K=2$ and $150$ and $600$ for the case of $K=3$. We use four variants of
graph Laplacian to conduct the spectral clustering, namely, $L$, $L_{\tau }$%
, $L_{\tau }^{\prime }$, and $L_\tau^{\prime \prime}$ defined in Sections %
\ref{sec:SC} and \ref{sec:ext}.

\begin{enumerate}
\item $L=D^{-1/2}AD^{-1/2}$ where $D=\mathrm{diag}(A\iota _{n})$. It is
possible that for some realizations, the minimum degree is 0, yielding
singular $D$.

\item $L_{\tau }=D_{\tau }^{-1/2}A_{\tau }D_{\tau }^{-1/2}$ where $A_{\tau
}=A+\tau J_{n},$ $D_{\tau }=\mathrm{diag}(A_{\tau }\iota _{n})$, and $%
J_{n}=n^{-1}\iota _{n}\iota _{n}^{T }$.

\item $L_{\tau }^{\prime }=D_{\tau }^{-1/2}AD_{\tau }^{-1/2}$ where $D_{\tau
}=D+\tau I_{n}$ and $I_{n}$ is an $n\times n$ identity matrix.

\item $L_{\tau }^{\prime \prime }=(D_{\tau }^{\prime \prime })^{-1/2}A_{\tau
}^{\prime \prime }(D_{\tau }^{\prime \prime })^{-1/2}$ where $A_{\tau
}^{\prime \prime }=A+\tau n^{-1}\hat{\Theta}\iota _{n}\iota _{n}^{T}\hat{
\Theta}$ and $D_{\tau }^{\prime \prime }=\text{diag}(A_{\tau }^{\prime
\prime }\iota _{n})$.
\end{enumerate}

The theoretical results in Sections \ref{sec:SC} and \ref{sec:ext} suggest
the strong consistency of the spectral clustering with $L_{\tau }$ and $%
L_{\tau }^{\prime }$ for the standard SBM and DC-SBM, respectively under
some conditions. In Sections \ref{sec:simsbm} and \ref{sec:simdcsbm}, we
consider these two cases. In addition, for the DC-SBM, we will also consider
the adaptive procedure introduced in Subsection \ref{sec:adaptive}.
Additional simulation results of spectral clustering with $L$ and $L_{\tau
}^{\prime }$ for the standard SBM and $L$ and $L_{\tau }$ for the DC-SBM can
be found in the supplementary Appendix \ref{sec:addsim}.

For the standard SBM, after obtaining the eigenvectors corresponding to the
largest $K$ eigenvalues of the graph Laplacian ($L$, $L_\tau$ and $%
L_\tau^\prime$), we classify them based on K-means algorithm (Matlab
\textquotedblleft kmedoids\textquotedblright\ function, which is more robust
to noise and outliers than \textquotedblleft kmeans\textquotedblright\
function, with default options). For the DC-SBM, before classification, we
normalize each row of the $n\times K$ eigenvectors so that its $L_{2}$ norm
equals 1. For comparison, we apply the unconditional pseudo-likelihood
method (UPL) and conditional pseudo-likelihood method (CPL) proposed by \cite%
{ACBL13} to detect the communities in the SBM and the DC-SBM, respectively.%
\footnote{\label{foot:scp} As \cite{ACBL13} remark, the UPL and CPL are
correctly fitting the SBM and the DC-SBM, respectively. In both UPL and CPL,
the initial classification is generated by spectral clustering with
perturbations (SCP). The SCP is spectral clustering based on $L_{\tau }$
with $\tau =\bar{d}/4$ and $\bar{d}$ being the average degree.} To evaluate
the classification performance, we consider two criteria: the Correct
Classification Proportion (CCP) and the Normalized Mutual Information (NMI). All the simulation results below are computed using the modified K-means algorithm. The simulation results for the standard K-means algorithm can be found in previous versions of this paper. When the regularizer $\tau$ is small, the modified K-means algorithm can produce slightly more accurate classification while at the optimal $\tau$ selected by our data-driven method explained below, the classification results in terms of CCP and NMI for the two algorithms are basically the same. 

\subsection{The standard SBM\label{sec:simsbm}}

We consider two data generating processes (DGPs).

\noindent \textbf{DGP 1:} Let $K=2$. Each community has $n/2$ nodes. The
matrix $B$ is set as 
\begin{equation*}
B=\frac{2}{n}%
\begin{pmatrix}
\log ^{2}(n) & 0.2\log (n) \\ 
0.2\log (n) & 0.8\log (n)%
\end{pmatrix}%
.
\end{equation*}%
The expected degrees are of order $\log ^{2}(n)$ and $\log (n)$ respectively
for communities 1 and 2.

\noindent \textbf{DGP 2:} Let $K=3$. Each community has $n/3$ nodes. The
matrix $B$ is set as 
\begin{equation*}
B=\frac{3}{n}%
\begin{pmatrix}
n^{1/2} & 0.1\log ^{5/6}(n) & 0.1\log ^{5/6}(n) \\ 
0.1\log ^{5/6}(n) & \log ^{3/2}(n) & 0.1\log ^{5/6}(n) \\ 
0.1\log ^{5/6}(n) & 0.1\log ^{5/6}(n) & 0.8\log ^{5/6}(n)%
\end{pmatrix}%
.
\end{equation*}%
The expected degrees are of order $n^{1/2}$, $\log ^{3/2}(n)$ and $\log
^{5/6}(n)$ respectively for communities 1, 2 and 3. 

We follow \cite{JY16} and select the regularizer $\tau $ that minimizes a
feasible version of 
\begin{equation*}
\Vert L_{\tau }-\mathcal{L}_{\tau }\Vert /|\sigma _{Kn}^{\tau }|.
\end{equation*}%
In particular, for a given $\tau $, we can obtain the community identities $%
\hat{Z}$ based on the spectral clustering of $L_{\tau }$. Given $\hat{Z}$,
we can estimate the block probability matrix $B$ by the fraction of links
between the estimated communities, which is denoted as $\hat{B}$. Let $\hat{P%
}=\hat{Z}\hat{B}\hat{Z}^{T}$, $\hat{P}_{\tau }=\hat{P}+\tau J_{n}$, $\hat{ 
\mathcal{D}}_{\tau }=\mathrm{diag}(\hat{P}_{\tau }\iota _{n})$, $\hat{ 
\mathcal{L}}_{\tau }=\hat{\mathcal{D}}_{\tau }^{-1/2}\hat{P}_{\tau }\hat{ 
\mathcal{D}}_{\tau }^{-1/2}$, and $\hat{\sigma}_{Kn}^{\tau }$ be the $K$-th
largest in absolute value eigenvalue of $\hat{\mathcal{L}}_{\tau }$. Then we
can compute 
\begin{equation*}
Q(\tau )=\Vert L_{\tau }-\hat{\mathcal{L}}_{\tau }\Vert / \vert \hat{\sigma}%
_{Kn}^{\tau } \vert.
\end{equation*}%
We search for some $\tau ^{\mathrm{JY}}$ that minimizes $Q(\tau )$ over a
grid of 20 points, $\tau _{j},$ on the interval $\left[ \tau _{\min },\tau
_{\max }\right] ,$ where $j=1,\ldots ,20,$ $\tau _{\min }=10^{-4}$ and $\tau
_{\max }$ is set to be the expected average degree. We set $\tau _{1}=\tau
_{\min },$ $\tau _{2}=1,$ and $\tau _{j+2}=(\tau _{\max })^{j/18}$ \ for $%
j=1,\ldots ,18.$ \cite{QR13} suggested choosing $\tau $ as the average
degree of nodes, which is approximately equal to the expected average degree.

All results reported here are based on 500 replications. For DGPs 1 and 2,
we report the classification results based on $L_{\tau }=D_{\tau
}^{-1/2}A_{\tau }D_{\tau }^{-1/2}$ in Figures \ref{fig:dgp_1_2} and \ref%
{fig:dgp_2_2}. The results based on $L$ and $L_{\tau }^{\prime }$ are
relegated to the supplementary Appendix \ref{sec:addsim}.


\begin{figure}[]
\centering
\includegraphics[scale = 0.5]{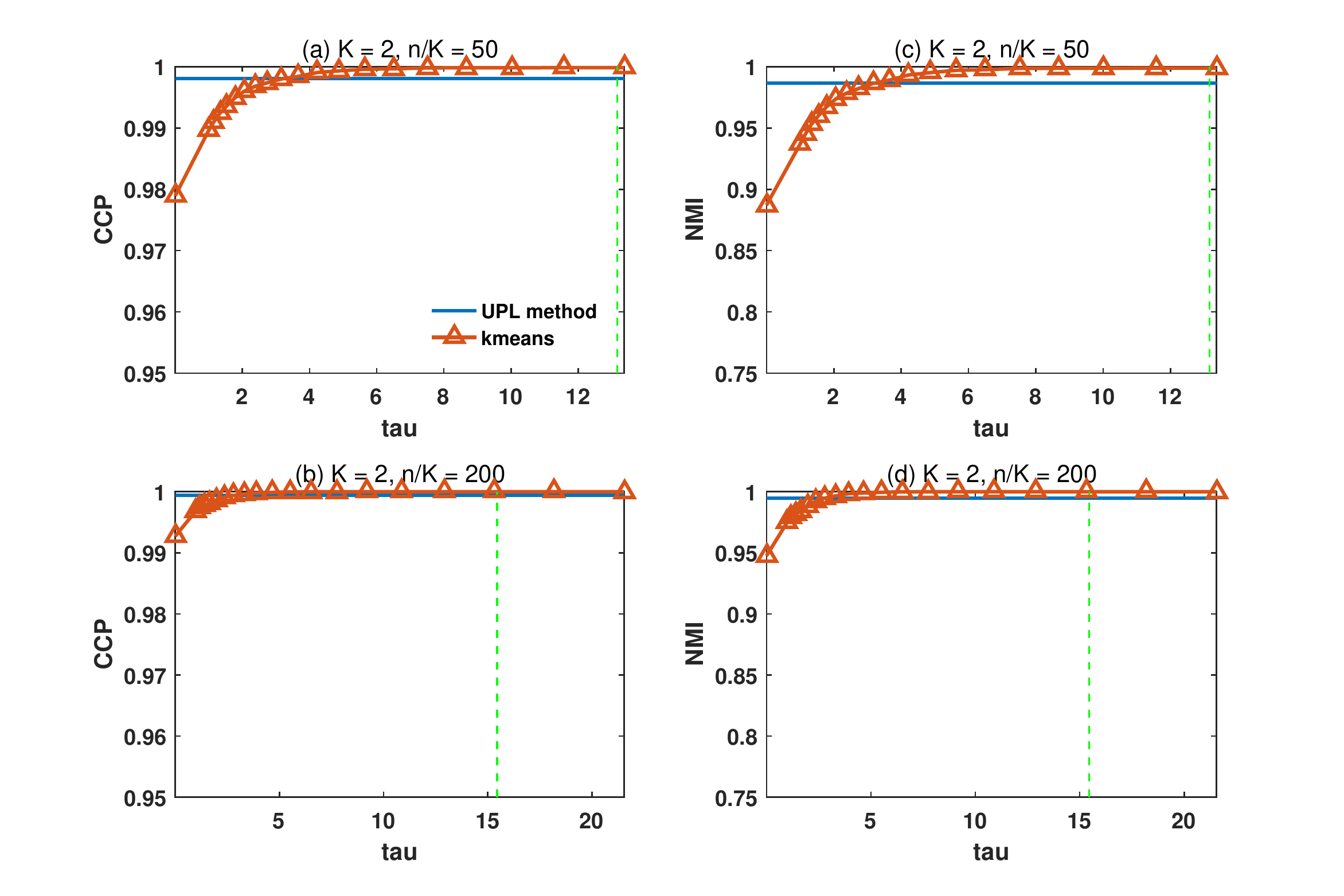}
\caption{Classification results for K-means for DGP 1 ($K=2$) based on $L_{%
\protect\tau }=D_{\protect\tau }^{-1/2}A_{\protect\tau }D_{\protect\tau %
}^{-1/2}$ and for UPL method. The $x $-axis marks $\protect\tau $ values,
and the $y$-axis is either CCP (left column) or NMI (right column). The
green vertical line in each subplot indicates the estimated $\protect\tau $
value by using the method of \protect\cite{JY16}. The first and second rows
correspond to $n/K=50$ and 200, respectively.}
\label{fig:dgp_1_2}
\end{figure}

In Figures \ref{fig:dgp_1_2} and \ref{fig:dgp_2_2}, the first and second
rows correspond to the results with $n=100$ and $n=400$, respectively. For
each replication, we can compute the feasible $\tau ^{\mathrm{JY}}$ as
mentioned above. Their averages across all replications are reported in each
subplot of Figures \ref{fig:dgp_1_2} and \ref{fig:dgp_2_2}. In particular,
the green dashed line represents $\tau ^{\mathrm{JY}}$, which can be easily
compared with the expected average degree, the rightmost vertical border.

We summarize our findings from Figures \ref{fig:dgp_1_2} and \ref%
{fig:dgp_2_2}. First, despite the fact that the minimal degrees for neither
DGP satisfies Assumption \ref{ass:rate} so that the standard spectral
clustering may not be consistent, the regularized spectral clustering
performs quite well in both DGPs. This confirms our theoretical finding that
the regularization can help to relax the requirement on the minimal degree
and to achieve the strong consistency. In addition, when a proper $\tau $ is
used, the spectral clustering based on $L_{\tau }$ outperforms the UPL
method of \cite{ACBL13}. Both results are in line with the theoretical
analysis by \cite{JY16}.

\begin{figure}[]
\centering
\includegraphics[scale = 0.5]{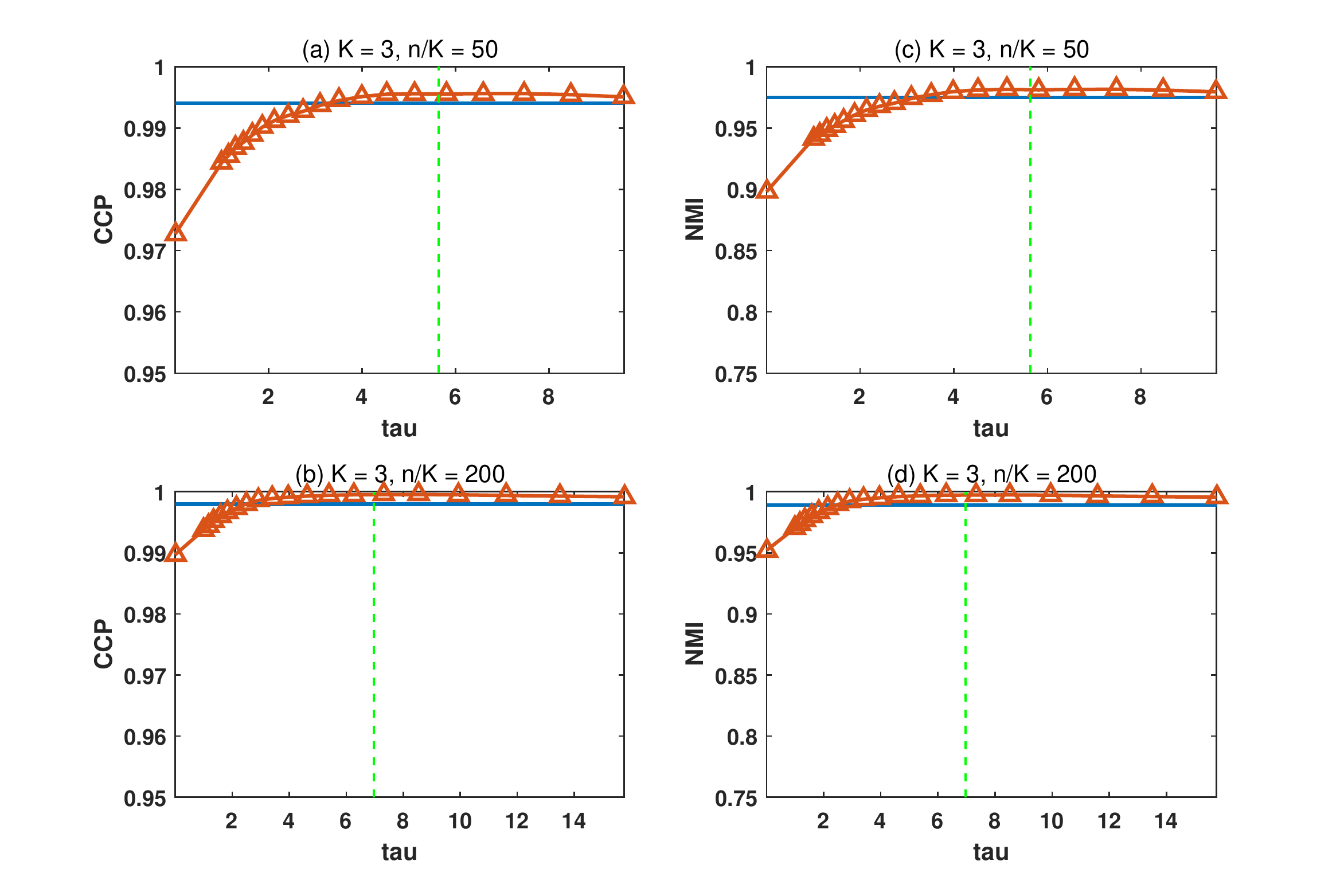}
\caption{Classification results for DGP 2 ($K=3$) based on $L_{\protect\tau %
}=D_{\protect\tau }^{-1/2}A_{\protect\tau }D_{\protect\tau }^{-1/2}$. (See
the explanations in Figure 1.)}
\label{fig:dgp_2_2}
\end{figure}


\subsection{The DC-SBM\label{sec:simdcsbm}}

The next two DGPs consider the degree-corrected SBM.

\noindent \textbf{DGP 3:} This DGP is the same as DGP 1 except that here $%
P=\Theta ZBZ^{T}\Theta ^{T}$, where $\Theta $ is a diagonal matrix with each
diagonal element taking a value from $\{0.5,1.5\}$ with equal probability.

\noindent \textbf{DGP 4:} This one is the same as DGP 2 except that here $%
P=\Theta ZBZ^{T}\Theta ^{T}$ and $\Theta $ is generated as in DGP 3.

To compute the feasible regularizer for the DC-SBM, we modify the previous
procedure to incorporate the degree heterogeneity. In particular, given $%
\tau $, by spectral clustering $L_{\tau }^{\prime }$, we can obtain a
classification $\hat{Z}=(\hat{Z}_{1},\ldots ,\hat{Z}_{n})^{T}$, where $\hat{Z%
}_{i}$ is a $K$ by 1 vector with its $\hat{g}_{i}$th entry being 1 and the
rest being 0 and $\hat{g}_{i}$ is an estimator of node $i$'s community
membership. Let $\hat{n}_{k}=\#\{i:\hat{g}_{i}=k\}$. Then we can estimate
the block probability matrix $B$ and $\Theta $ by $\hat{B}=[\hat{B}%
_{kl}]_{1\leq k,l\leq K}$ and $\hat{\Theta}=\mathrm{diag}(\hat{\theta}%
_{1},\ldots ,\hat{\theta}_{n}),$ where $\hat{\theta}_{i}$ is defined in (\ref%
{eq:theta}) and $\hat{B}_{kl}=(\sum\nolimits_{(i,j):\hat{g}_{i}=k,\hat{g}
_{j}=l}A_{ij})/(\hat{n}_{k}\hat{n}_{l}).$ Let $\hat{P}=\hat{\Theta}\hat{Z}%
\hat{B}\hat{Z}^{T}\hat{\Theta}^{T}$, $\hat{\mathcal{D}}_{\tau }=\mathrm{diag}%
(\hat{P}\iota _{n})+\tau I_{n}$, and $\hat{\mathcal{L}}_{\tau }^{\prime }=%
\hat{\mathcal{D}}_{\tau }^{-1/2}\hat{P}\hat{\mathcal{D}}_{\tau }^{-1/2}$.
Let $\hat{\sigma}_{Kn}^{\prime \tau }$ denote the $K$-th largest eigenvalue
of $\hat{\mathcal{L}}_{\tau }^{\prime }$ (in absolute value). Let 
\begin{equation*}
Q^{\prime }(\tau )=\Vert L_{\tau }^{\prime }-\hat{\mathcal{L}}_{\tau
}^{\prime }\Vert / \vert \hat{\sigma}_{Kn}^{\prime \tau }\vert.
\end{equation*}%
We search for some $\tau ^{\prime \mathrm{JY}}$ that minimizes $Q^{\prime
}(\tau )$ over the same aforementioned grid.

\begin{figure}[]
\centering
\includegraphics[scale = 0.6]{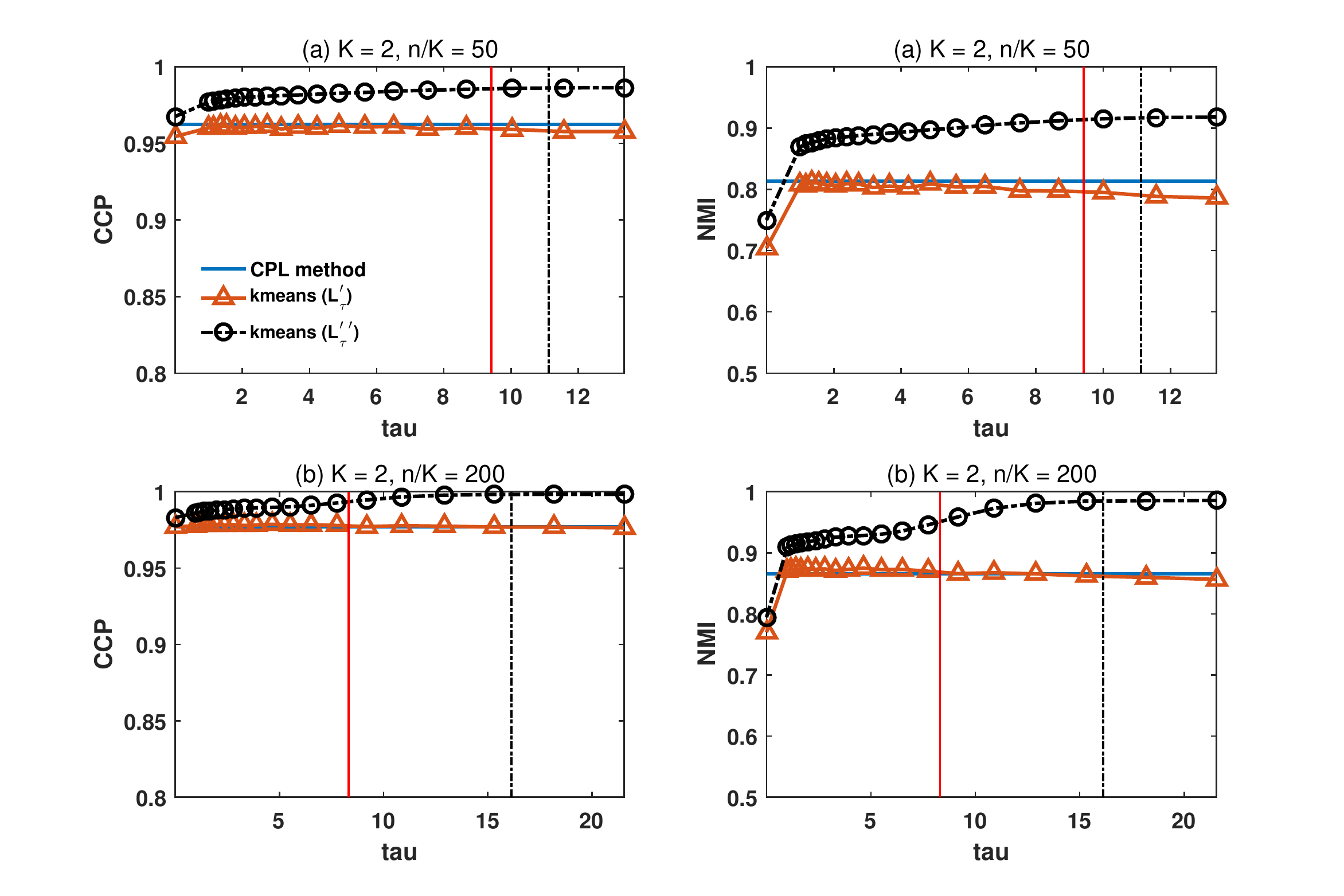}
\caption{Classification results for DGP 3 ($K=2$, degree-corrected) based on 
$L_{\protect\tau }^{\prime }=D_{\protect\tau }^{-1/2}AD_{\protect\tau %
}^{-1/2}$ and $L_{\protect\tau }^{\prime \prime }=D_{\protect\tau }^{-1/2}A_{%
\protect\tau }D_{\protect\tau }^{-1/2}$. The red and black vertical lines
correspond to the optimal regularizers $\protect\tau ^{\prime \mathrm{JY}}$
and $\protect\tau ^{\prime \prime \mathrm{JY}}$, respectively. (See Figure 
\protect\ref{fig:dgp_1_2} for the explanation of other features of the
figure.)}
\label{fig:dgp_3_1}
\end{figure}


For DGPs 3 and 4, we report the classification results based on $L_{\tau
}^{\prime }=D_{\tau }^{-1/2}AD_{\tau }^{-1/2}$ as the orange lines in
Figures \ref{fig:dgp_3_1} and \ref{fig:dgp_4_1}. For each subplot, the
rightmost border line and the red vertical line represent the averages of $%
\bar{d}$ and $\tau ^{\prime \mathrm{JY}}$, respectively. Figures \ref%
{fig:dgp_3_1} and \ref{fig:dgp_4_1} show the regularized spectral clustering
based on $L_\tau^\prime $ is slightly outperformed by CPL in DC-SBMs.
However, $\tau ^{\prime \mathrm{JY}}$ has the close-to-optimal performance
in terms of both CCP and NMI over a range of values for $\tau $.

Table \ref{table:feasible} reports the classification results for the
spectral clustering with $\tau =\tau ^{\mathrm{JY}}$ for DGPs 1--2 (or $\tau
^{\prime \mathrm{JY}}$ for DGPs 3--4) and $\bar{d}$ in comparison with those
for the UPL (or CPL for DGPs 3--4) method over 500 replications. In general,
the spectral clustering with $\tau =\tau ^{\mathrm{JY}}$ outperforms the UPL
method in DGPs 1--2 but slightly underperforms the CPL method for DGPs 3 and
4. In all cases, we observe that the increase of the probability of correct
classification as $n$ increases. This is consistent with the theory because
both the UPL/CPL method and our regularized spectral clustering method are
strongly consistent.

\begin{figure}[]
\centering
\includegraphics[scale = 0.5]{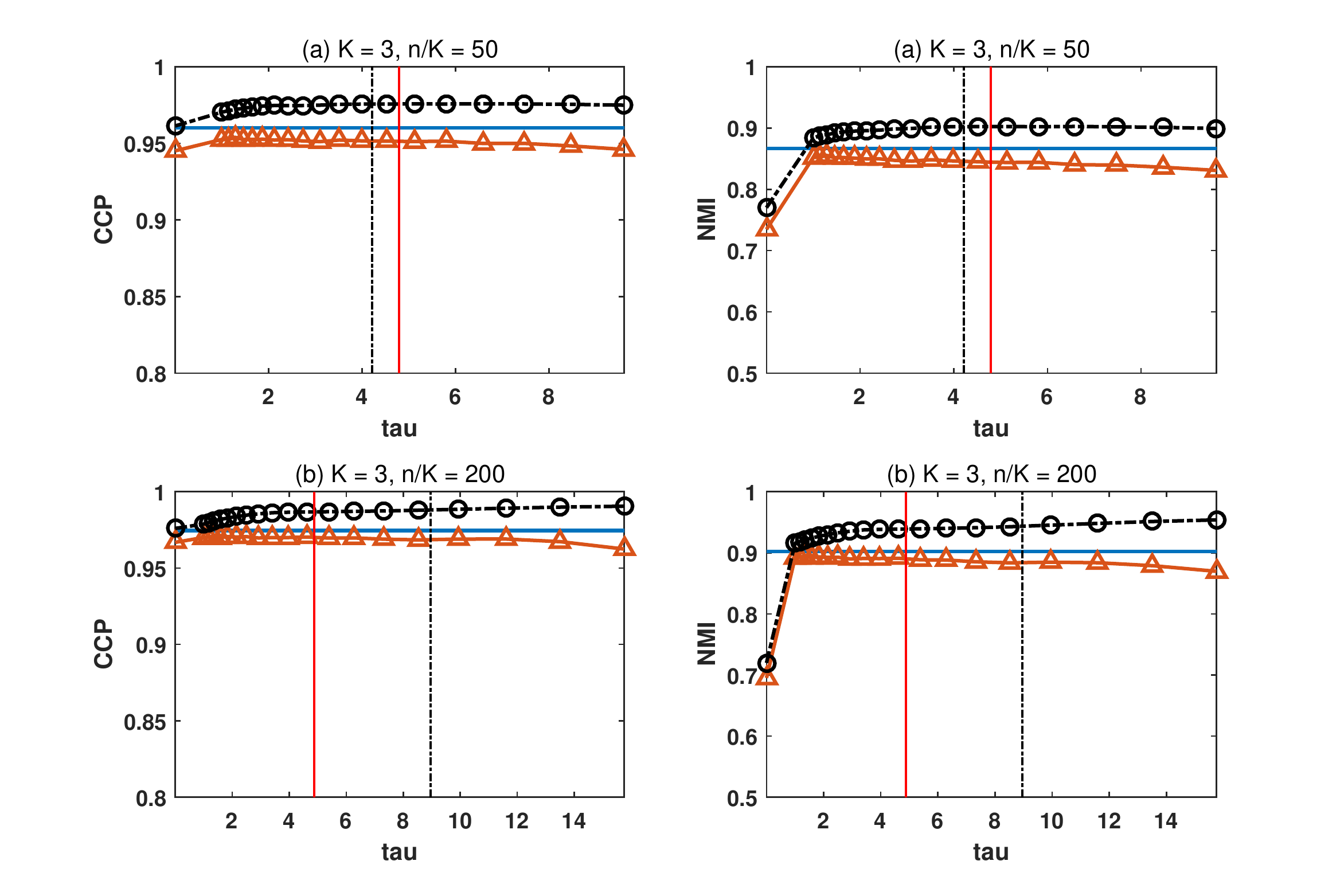}
\caption{Classification results for DGP 4 ($K=3$, degree-corrected) based on 
$L_{\protect\tau }^{\prime }=D_{\protect\tau }^{-1/2}AD_{\protect\tau %
}^{-1/2}$ and $L_{\protect\tau }^{\prime \prime }=D_{\protect\tau }^{-1/2}A_{%
\protect\tau }D_{\protect\tau }^{-1/2}$. The red and black vertical lines
corresponds to the optimal regularizers $\protect\tau ^{\prime \mathrm{JY}}$
and $\protect\tau ^{\prime \prime \mathrm{JY}}$, respectively. (See Figure 
\protect\ref{fig:dgp_1_2} for the explanation of other features of the
figure.)}
\label{fig:dgp_4_1}
\end{figure}


\linespread{1.2}%
\begin{table}[H]%
\caption{Comparison of classification results} \label{table:feasible} %
\centering%
\begin{tabular}{lll|ccc|ccc}
\hline\hline
&  &  & \multicolumn{3}{c}{CCP} & \multicolumn{3}{|c}{NMI} \\ \hline
&  &  & \multicolumn{2}{c}{Spectral clustering} & UPL/CPL & 
\multicolumn{2}{c}{Spectral clustering} & UPL/CPL \\ \hline
DGP & $K$ & $n/K$ & $\bar{d}$ & $\tau ^{\mathrm{JY}}/\tau ^{\prime \mathrm{JY%
}}$ &  & $\bar{d}$ & $\tau ^{\mathrm{JY}}/\tau ^{\prime \mathrm{JY}}$ &  \\ 
\hline
1 & 2 & 50 & 0.9998 & 0.9998 & 0.9980 & 0.9989 & 0.9989 & 0.9865 \\ 
& 2 & 200 & 1.0000 & 1.0000 & 0.9994 & 1.0000 & 1.0000 & 0.9947 \\ 
2 & 3 & 50 & 0.9951 & 0.9956 & 0.9941 & 0.9795 & 0.9812 & 0.9748 \\ 
& 3 & 200 & 0.9992 & 0.9995 & 0.9979 & 0.9954 & 0.9972 & 0.9889 \\ \hline
3 & 2 & 50 & 0.9576 & 0.9596 & 0.9623 & 0.7857 & 0.7964 & 0.8134 \\ 
& 2 & 200 & 0.9764 & 0.9777 & 0.9769 & 0.8564 & 0.8689 & 0.8658 \\ 
4 & 3 & 50 & 0.9460 & 0.9513 & 0.9600 & 0.8308 & 0.8444 & 0.8668 \\ 
& 3 & 200 & 0.9624 & 0.9701 & 0.9745 & 0.8696 & 0.8902 & 0.9022 \\ \hline
\end{tabular}
\end{table}%
\linespread{1.25}%

Figures \ref{fig:dgp_3_1} and \ref{fig:dgp_4_1} also report the
classification results based on $L_{\tau }^{\prime \prime }$, which are
shown as the dark lines. We find the performance of spectral clustering
based on $L_\tau^{\prime \prime}$ is better than those using the CPL method.
In addition, our choice of $\tau ^{\prime \prime \mathrm{JY}}$, marked as
the dark vertical line in each subplot, performs well in both DGPs 3 and 4.

\section{Proof strategy}

\label{sec:strategy} In this section we outline the proof strategies for the
main results in Section \ref{sec:ext2}. First, noting that the regularized
spectral clustering for the DC-SBM nests standard SBM without regularization
by setting $\tau =0$ and $\theta _{i}=1$ $\forall $ $i=1,\cdots ,n$, all the
main results in Section \ref{sec:SC} follow that in Section \ref{sec:ext2}.
Second, based on the results in Section \ref{sec:SC}, the results for the
standard SBM with regularization in Section \ref{sec:ext1} can be derived by
replacing $B_{0}$, $\mu _{n}$, $\rho _{n}$, and $\sigma _{Kn}$ by their
counterparts with regularization, i.e., $B_{0}^{\tau }$, $\mu _{n}^{\tau }$, 
$\rho _{n}^{\tau }$, and $\sigma _{Kn}^{\tau }$, respectively.

Section \ref{sec:ext2} contains Theorems \ref{thm:id3}, \ref{thm:main_DC}
and \ref{thm:thetahat}, Lemma \ref{lem:dk3} and Corollary \ref{cor:DC}.
Since the proofs of Theorems \ref{thm:id3} and \ref{thm:thetahat}, Lemma \ref%
{lem:dk3} and Corollary \ref{cor:DC} are relatively simple, below we focus
on the proof strategy for Theorem \ref{thm:main_DC}.

Theorem \ref{thm:main_DC} aims to establish a uniform upper bound for each
row of the gap between sample and population eigenvectors (up to some
rotation), i.e., $\sup_{i}||\hat{u}_{i}^{T}\hat{O}_{n}-u_{i}^{T}||$, where $%
\hat{u}_{i}^{T}$ and $u_{i}^{T}$ are the $i$-th rows of $\hat{U}_{1n}$ and $%
U_{1n}$, respectively. Let $\hat{\Lambda}=L_{\tau }^{\prime }\hat{U}_{1n}%
\hat{O}_{n}=\hat{U}_{1n}\hat{\Sigma}_{n}\hat{O}_{n}$, $\Lambda =\mathcal{L}%
_{\tau }^{\prime }U_{1n}=U_{1n}\Sigma _{n}$, $\hat{\Lambda}_{i}=\hat{u}%
_{i}^{T}\hat{\Sigma}_{n}\hat{O}_{n}$, and $\Lambda _{i}=u_{i}^{T}\Sigma _{n}$%
. Our proof strategy is to obtain the upper and lower bounds for $%
(n_{g_{i}^{0}}^{\tau })^{1/2}(\theta _{i}^{\tau })^{-1/2}||\hat{ \Lambda}%
_{i}-\Lambda _{i}||$, both of which involve $(n_{g_{i}^{0}}^{\tau
})^{1/2}(\theta _{i}^{\tau })^{-1/2}||\hat{u}_{i}^{T}\hat{O}_{n}-u_{i}^{T}||$%
. The two bounds produce a contraction mapping for $\sup_{i}(n_{g_{i}^{0}}^{%
\tau })^{1/2}(\theta _{i}^{\tau })^{-1/2}||\hat{u}_{i}^{T}\hat{O}%
_{n}-u_{i}^{T}||$. By iterating the contraction mapping sufficiently many
times, we obtain the desired bound.

\textbf{Lower bound.} In order to derive the lower bound for $%
(n_{g_{i}^{0}}^{\tau })^{1/2}||\hat{\Lambda}_{i}-\Lambda _{i}||$, we note
that 
\begin{align}
(n_{g_{i}^{0}}^{\tau })^{1/2}(\theta _{i}^{\tau })^{-1/2}\Vert \hat{\Lambda}%
_{i}-\Lambda _{i}\Vert & =(n_{g_{i}^{0}}^{\tau })^{1/2}(\theta _{i}^{\tau
})^{-1/2}\Vert \hat{u}_{i}^{T}\hat{\Sigma}_{n}\hat{O}_{n}-u_{i}^{T}\Sigma
_{n}\Vert  \notag  \label{eq:lower} \\
& \geq (n_{g_{i}^{0}}^{\tau })^{1/2}(\theta _{i}^{\tau })^{-1/2}\Vert (\hat{u%
}_{i}^{T}\hat{O}_{n}-u_{i}^{T})\hat{\Sigma}_{n}\Vert -(n_{g_{i}^{0}}^{\tau
})^{1/2}(\theta _{i}^{\tau })^{-1/2}\Vert u_{i}^{T}(\hat{\Sigma}_{n}-\Sigma
_{n})\Vert  \notag \\
& \quad \ -(n_{g_{i}^{0}}^{\tau })^{1/2}(\theta _{i}^{\tau })^{-1/2}\Vert 
\hat{u}_{i}^{T}(\hat{\Sigma}_{n}\hat{O}_{n}-\hat{O}_{n}\hat{\Sigma}_{n})\Vert
\notag \\
& \equiv I_{i}-II_{i}-III_{i}.
\end{align}%
Clearly, by the Hoffman-Wielandt inequality, Lemma \ref{lem:dk3}, and Assumption \ref{ass:rate3}.2, 
\begin{align*}
|\hat{	\sigma}_{Kn}| \geq |\sigma_{Kn}| - 7 \left(\frac{\log(n)}{\mu_n^\tau \sigma_{Kn}^2} \right)^{1/2}|\sigma_{Kn}| \geq 0.999|\sigma_{Kn}| \quad a.s.,
\end{align*}
 and thus, 
\begin{equation*}
\sup_{i}I_{i}\geq 0.999|\sigma _{Kn}|\Gamma _{n}\quad a.s.,
\end{equation*}%
where $\Gamma _{n}=\sup_{i}|(n_{g_{i}^{0}}^{\tau })^{1/2}(\theta _{i}^{\tau
})^{-1/2}\Vert \hat{u}_{i}^{T}\hat{O}_{n}-u_{i}^{T}\Vert $. It is the
leading term of the lower bound involving $\Gamma _{n}$. In the online
Appendix \ref{sec:31pf}, we show that $\sup_{i}II_{i}\leq 7(\log (n)/\mu
_{n}^{\tau })^{1/2}\quad a.s.$ and $\sup_{i}III_{i}\leq 34(\log (n)/\mu
_{n}^{\tau })^{1/2}|\sigma _{Kn}|^{-1}(\Gamma _{n}+1)\quad a.s.$ It follows
that 
\begin{align}
\sup_{i}(n_{g_{i}^{0}}^{\tau })^{1/2}(\theta _{i}^{\tau })^{-1/2}\Vert \hat{%
\Lambda}_{i}-\Lambda _{i}\Vert \geq & (0.999|\sigma _{Kn}|-34(\log (n)/\mu
_{n}^{\tau })^{1/2}|\sigma _{Kn}^{-1}|)\Gamma _{n}-41(\log (n)/\mu
_{n}^{\tau })^{1/2}|\sigma _{Kn}^{-1}|  \notag \\
\geq & 0.99|\sigma _{Kn}|\Gamma _{n}-41(\log (n)/\mu _{n}^{\tau
})^{1/2}|\sigma _{Kn}^{-1}|,  \label{LB1}
\end{align}%
where we use the fact that $34(\log (n)/\mu _{n}^{\tau })^{1/2}|\sigma
_{Kn}^{-2}|\leq 0.09.$

\textbf{Upper bound.} To derive the upper bound for $%
\sup_{i}(n_{g_{i}^{0}}^{\tau })^{1/2}(\theta _{i}^{\tau })^{-1/2}\Vert \hat{%
\Lambda}_{i}-\Lambda _{i}\Vert $, we first denote $\tilde{\Lambda}=D_{\tau
}^{-1/2}PD_{\tau }^{-1/2}U_{1n}$ and $\tilde{\Lambda}_{i}=(\hat{d}_{i}^{\tau
})^{-1/2}[P]_{i\cdot }D_{\tau }^{-1/2}U_{1n}$ as the $i$-th row of $\tilde{%
\Lambda}$. Then, we have 
\begin{align}
\sup_{i}(n_{g_{i}^{0}}^{\tau })^{1/2}(\theta _{i}^{\tau })^{-1/2}\Vert \hat{%
\Lambda}_{i}-\Lambda _{i}\Vert & \leq \sup_{i}(n_{g_{i}^{0}}^{\tau
})^{1/2}(\theta _{i}^{\tau })^{-1/2}\Vert \Lambda _{i}-\tilde{\Lambda}%
_{i}\Vert +\sup_{i}(n_{g_{i}^{0}}^{\tau })^{1/2}(\theta _{i}^{\tau
})^{-1/2}\Vert \hat{\Lambda}_{i}-\tilde{\Lambda}_{i}\Vert  \notag \\
& \equiv T_{1}+T_{2}.  \label{eq:12DC3}
\end{align}%
For $T_{2},$ we have 
\begin{align}
T_{2}& =\sup_{i}(n_{g_{i}^{0}}^{\tau })^{1/2}(\theta _{i}^{\tau
})^{-1/2}\Vert (\hat{d}_{i}^{\tau })^{-1/2}[A]_{i\cdot }D_{\tau }^{-1/2}\hat{%
U}_{1n}\hat{O}_{n}-(\hat{d}_{i}^{\tau })^{-1/2}[P]_{i\cdot }D_{\tau
}^{-1/2}U_{1n}\Vert  \notag \\
& \leq \sup_{i}(n_{g_{i}^{0}}^{\tau })^{1/2}(\theta _{i}^{\tau })^{-1/2}(%
\hat{d}_{i}^{\tau })^{-1/2}\Vert \lbrack P]_{i\cdot }D_{\tau }^{-1/2}(\hat{U}%
_{1n}\hat{O}_{n}-U_{1n})\Vert  \notag \\
& +\sup_{i}(n_{g_{i}^{0}}^{\tau })^{1/2}(\theta _{i}^{\tau })^{-1/2}(\hat{d}%
_{i}^{\tau })^{-1/2}\Vert ([A]_{i\cdot }-[P]_{i\cdot })(D_{\tau }^{-1/2}-%
\mathcal{D}_{\tau }^{-1/2})\hat{U}_{1n}\hat{O}_{n}\Vert  \notag \\
& +\sup_{i}(n_{g_{i}^{0}}^{\tau })^{1/2}(\theta _{i}^{\tau })^{-1/2}(\hat{d}%
_{i}^{\tau })^{-1/2}\Vert ([A]_{i\cdot }-[P]_{i\cdot })\mathcal{D}_{\tau
}^{-1/2}\hat{U}_{1n}\hat{O}_{n}\Vert  \notag \\
& \equiv T_{2,1}+T_{2,2}+T_{2,3}.  \label{eq:2DC}
\end{align}%
Lemma \ref{lem:B3} in the online Appendix \ref{sec:lem} provides the upper
bounds for $T_{1}$, $T_{2,1}$, $T_{2,2}$, and $T_{2,3}$. Taking $T_{2,3}$ as
an example, we note that 
\begin{equation*}
T_{2,3}=\sup_{i}\sup_{h=\hat{U}_{1n}\hat{O}_{n}f,f\in
S^{K-1}}(n_{g_{i}^{0}}^{\tau })^{1/2}(\theta _{i}^{\tau
})^{-1/2}\sum_{j=1}^{n}(A_{ij}-P_{ij})(\hat{d}_{i}^{\tau }d_{j}^{\tau
})^{-1/2}h_{j}.
\end{equation*}%
Here, $h_{j}$ denotes the $j$th element of $h.$ Lemma \ref{lem:V1n3} builds
a Bernstein-type concentration inequality to upper bound $T_{2,3}$, which
involves the $l_{\infty }$ and $l_{2}$ norms of $h$, In particular, $%
||h||_{\infty }$ depends on the rough upper bound $\delta _{n}^{(0)}$ for $%
\Gamma _{n}$.\footnote{%
In fact, the upper bound for $||h||_{\infty }$ in the proof, which is
denoted as $\psi _{n}$, is $\delta _{n}^{(0)}+1$.} One of the technical
difficulties is that, due to the correlation between the sample graph
Laplacian and its eigenvectors, the sequence of random variables $%
A_{ij}:j=1,\cdots ,n$ are not independent of $h=\hat{U}_{1n}\hat{O}_{n}f$
for some $f\in S^{K-1}$. To deal with it, we rely on the \textquotedblleft
leave-one-out" technique used in \cite{abbe2017}, \cite{B13}, \cite{JM15},
and \cite{Z18}. The idea is to approximate the eigenvector by a vector which
is independent of one particular row of the sample graph Laplacian. This
helps to restore the independence. Then, the approximation errors are
bounded in Lemma \ref{lem:looDC}, which further calls upon Lemmas \ref%
{lem:Li-LDC} and \ref{lem:DiDC}.

At the end, Lemma \ref{lem:B3} establishes that 
\begin{align}
& \sup_{i}(n_{g_{i}^{0}}^{\tau })^{1/2}(\theta _{i}^{\tau })^{-1/2}\Vert 
\hat{\Lambda}_{i}-\Lambda _{i}\Vert  \notag \\
\leq & 3450C_{1}c_{1}^{-1/2}\rho _{n}\log ^{1/2}(n)(\mu _{n}^{\tau
})^{-1/2}|\sigma _{Kn}^{-1}|\biggl[\delta _{n}^{(0)}+1+\rho _{n}+\frac{%
\left( \frac{1}{K}+\frac{\log (5)}{\log (n)}\right) ^{1/2}\rho _{n}^{1/2}%
\overline{\theta }^{1/4}}{\underline{\theta }^{1/4}}\biggr],\quad a.s.,
\label{UB1}
\end{align}%
where we can choose $\delta _{n}^{(0)}=n^{1/2}\underline{\theta }^{-1/2}$.
Combining the lower and upper bounds in (\ref{LB1}) and (\ref{UB1}) for $%
\sup_{i}(n_{g_{i}^{0}}^{\tau })^{1/2}(\theta _{i}^{\tau })^{-1/2}\Vert \hat{%
\Lambda}_{i}-\Lambda _{i}\Vert $ and applying Assumption \ref{ass:rate3}, we
have 
\begin{equation}
0.001\delta _{n}^{(0)}+3527C_{1}c_{1}^{-1/2}\eta _{n}\geq \Gamma _{n},
\label{UB2}
\end{equation}%
where $\eta _{n}$ is defined in Theorem \ref{thm:main_DC}.

\textbf{Iteration.} (\ref{UB2}) suggests that the initial rough upper bound $%
\delta _{n}^{(0)}$ for $\Gamma _{n}$ can be refined to $\delta
_{n}^{(1)}\equiv 0.001\delta _{n}^{(0)}+3527C_{1}c_{1}^{-1/2}\eta _{n}$.
Then we can take this new upper bound into the previous calculations to
obtain 
\begin{equation*}
0.001\delta _{n}^{(1)}+3527C_{1}c_{1}^{-1/2}\eta _{n}\geq \Gamma _{n}.
\end{equation*}%
Therefore, we have constructed a contraction mapping, through which we can
refine our upper bound for $\Gamma _{n}$ via iterations. We iterate the
above calculation $t$ times for some arbitrary integer $t$, and obtain that 
\begin{equation*}
\Gamma _{n}\leq \delta _{n}^{(t)},\quad \delta _{n}^{(t)}=0.001\delta
_{n}^{(t-1)}+3527C_{1}c_{1}^{-1/2}\eta _{n}.
\end{equation*}%
This implies 
\begin{equation*}
\delta _{n}^{(t)}=\left( 0.001\right) ^{t}\biggl[\delta
_{n}^{(0)}-3527C_{1}c_{1}^{-1/2}\eta _{n}\biggr]+3527C_{1}c_{1}^{-1/2}\eta
_{n}.
\end{equation*}%
Letting $t=n$, we have 
\begin{equation*}
\Gamma _{n}\leq \delta _{n}^{(n)}\leq 1000^{-n}n^{1/2}\underline{\theta }%
^{-1/2}+3527C_{1}c_{1}^{-1/2}\eta _{n}\leq 3528C_{1}c_{1}^{-1/2}\eta _{n},
\end{equation*}%
where we denote $C^{\ast }$ in Theorem \ref{thm:main_DC} as $%
3528C_{1}c_{1}^{-1/2}$ and we use the fact that it is possible to choose $%
\delta _{n}^{(0)}=n^{1/2}\underline{\theta }^{-1/2}$ as the initial rough
bound for $\Gamma _{n}$.

\section{Conclusion}

\label{sec:concl} In this paper, we show that under suitable conditions, the
K-means algorithm applied to the eigenvectors of the graph Laplacian
associated with its first few largest eigenvalues can classify all
individuals into the true community uniformly correctly almost surely in
large samples. In the special case where the number of communities is fixed
and the probability block matrix has minimal eigenvalue bounded away from
zero, the strong consistency essentially requires that the minimal degree
diverges to infinity at least as fast as $\log(n)$, which is the minimal
rate requirement for the strong consistency discussed in \cite{A18}. Similar
results are also established for the regularized DC-SBMs. The simulations
confirm our theoretical findings and indicate that an adaptive procedure can
improve the finite sample performance of the regularized spectral clustering
for DC-SBMs.



\newpage
\appendix

\begin{center}
\huge{
Online Supplement to ``Strong Consistency of Spectral Clustering for
Stochastic Block Models"}
\end{center}

\begin{abstract}
	This supplement is composed of four parts. Sections \ref{sec:2pf} and \ref%
	{sec:31pf} provide the proofs of the main results in Sections \ref{sec:SC}
	and \ref{sec:ext}, respectively. Section \ref{sec:lem} contains some lemmas
	that are used in the proofs of the main results. Section \ref{sec:addsim}
	presents some additional simulation results. \medskip
	
	\noindent \textbf{Key words and phrases: }Community detection,
	degree-corrected stochastic block model, K-means, regularization, strong
	consistency.\vspace{2mm}
\end{abstract}
\setcounter{page}{1} \renewcommand\thesection{\Alph{section}} %

\section{Proofs of the results in Section \protect\ref{sec:SC} \label%
	{sec:2pf}}

In this section, we prove the main results in Section \ref{sec:SC}, viz.,
Theorems \textbf{\ref{thm:id}--\ref{thm:strong}}, Lemmas \textbf{\ref{lem:dk}%
	--\ref{lem:kmeans1}}, and Corollary \textbf{\ref{cor:sc}.} In particular, we
note that the standard SBM is a special case of regularized DC-SBM with
regularizer $\tau =0$ and degree-corrected parameter $\theta _{i}=1$.
Therefore, Lemma \ref{lem:dk} and Theorem \ref{thm:main} follow Lemma \ref%
{lem:dk3} and Theorem \ref{thm:main_DC}, respectively.\bigskip

\begin{proof}[\textbf{Proof of Theorem \protect\ref{thm:id}}]
	By the proof of \citet[][Lemma 3.1]{RCY11}, we have $\mathcal{L}%
	=n^{-1}ZB_{0}Z^{T}$. Therefore, $\mathcal{L}%
	^{2}=n^{-1}ZB_{0}(Z^{T}Z/n)B_{0}Z^{T}.$ Let $\Pi _{n}=Z^{T}Z/n=\text{diag}%
	(\pi _{1n},\ldots ,\pi _{Kn})$. By the spectral decomposition in Assumption %
	\ref{ass:id}, we have 
	\begin{equation}
	\Pi _{n}^{1/2}B_{0}\Pi _{n}B_{0}\Pi _{n}^{1/2}=S_{n}\Omega _{n}^{2}S_{n}^{T},
	\label{eq:sigma}
	\end{equation}%
	where $\Omega _{n}=\text{diag}(\omega _{1n},\ldots ,\omega _{Kn})$ such that 
	$|\omega _{1n}|\geq |\omega _{2n}|\geq \cdots \geq |\omega _{Kn}|>0$ and $%
	S_{n}$ is a $K\times K$ matrix such that $S_{n}^{T}S_{n}=I_{K}.$ Let $%
	U_{1n}^{\ast }=Z(Z^{T}Z)^{-1/2}S_{n}$. Then, we have 
	\begin{equation}
	U_{1n}^{\ast }\Omega _{n}^{2}U_{1n}^{\ast T}=\mathcal{L}^{2}=U_{1n}\Sigma
	_{1n}^{2}U_{1n}^{T}.  \label{eq:sigma2}
	\end{equation}%
	In addition, $U_{1n}^{\ast T}U_{1n}^{\ast }=S_{n}^{T}S_{n}=I_{K}$. Therefore
	the columns of $U_{1n}^{\ast }$ are the eigenvectors of $\mathcal{L}$
	associated with eigenvalues $\sigma _{1n},\ldots ,\sigma _{Kn}$, up to sign
	normalization. Without loss of generality (W.l.o.g.), we can take $%
	U_{1n}=U_{1n}^{\ast }$ and $\Omega _{n}=\Sigma _{1n}$.
	
	Furthermore, if node $i$ is in cluster $C_{k_{1}}$, then $%
	z_{i}^{T}(Z^{T}Z)^{-1/2}S_{n}=n_{k_{1}}^{-1/2}[S_{n}]_{k_{1}\cdot }$, where $%
	[S_{n}]_{k\cdot }$ denotes the $k$-th row of $S_{n}$. Therefore, by
	Assumption \ref{ass:nk} and the fact that $\Vert \lbrack S_{n}]_{k_{1}\cdot
	}\Vert =1$, 
	\begin{equation*}
	(n/K)^{1/2}\Vert z_{i}^{T}(Z^{T}Z)^{-1/2}S_{n}\Vert \leq c_1^{-1/2}\Vert
	\lbrack S_{n}]_{k_{1}\cdot }\Vert =c_1^{-1/2}.
	\end{equation*}%
	Taking $\sup_{i}$ on both sides establishes the first desired result.
	
	Similarly, by Assumption \ref{ass:nk}, we can also establish the lower
	bound: for node $j$ in cluster $C_{k_{2}}$ with $k_{1}\neq k_{2}$ 
	\begin{equation*}
	(n/K)^{1/2}\Vert (z_{i}-z_{j})^{T}(Z^{T}Z)^{-1/2}S_{n}\Vert
	=||n_{k_{1}}^{-1/2}[S_{n}]_{k_{1}\cdot }-n_{k_{2}}^{-1/2}[S_{n}]_{k_{2}\cdot
	}||\geq C_{1}^{-1/2}\sqrt{2}=\underline{c}>0.
	\end{equation*}%
	This concludes the proof. \bigskip
\end{proof}

\begin{proof}[\textbf{Proof of Lemma \protect\ref{lem:dk}}]
	Lemma \ref{lem:dk} is a special case of Lemma \ref{lem:dk3} with $\theta
	_{i}=1$ for $i=1,\cdots ,n$ and $\tau =0$. We prove the general result in
	Lemma \ref{lem:dk3} later.\bigskip
\end{proof}

\begin{proof}[\textbf{Proof of Theorem \protect\ref{thm:main}}]
	Theorem \ref{thm:main} is a special case of Theorem \ref{thm:main_DC} when $%
	\theta _{i}=1$ for $i=1,\cdots ,n$ and $\tau =0$. We prove Theorem \ref%
	{thm:main_DC} with $C^{\ast }=3528C_1c_1^{-1/2}$ later. \bigskip
\end{proof}

\begin{proof}[\textbf{Proof of Lemma \protect\ref{lem:kmeans1}}]
	Let $Q_{n}(\mathcal{A})=\sum_{k=1}^{K}\min_{1\leq l\leq K}\Vert \beta
	_{kn}-\alpha _{l}\Vert ^{2}\pi _{kn}.$ We first derive the convergence rate
	of $\widehat{Q}_{n}(\mathcal{A})-Q_{n}(\mathcal{A})$ uniformly over $%
	\mathcal{A}\in \mathcal{M}=\{(\alpha _{1},\ldots ,\alpha _{K}):\sup_{1\leq
		k\leq K}\Vert \alpha _{k}\Vert \leq 2M\}$ for some constant $M$ independent
	of $n$. Let $R_{n}=\sup_{i}\Vert \hat{\beta}_{in}-\beta _{g_{i}^{0}n}\Vert $
	. Then, by Assumption \ref{ass:theta}.3, 
	\begin{equation}
	R_{n} \leq c_{2n} \leq M\quad a.s.  \label{eq:Rn}
	\end{equation}%
	In addition, 
	\begin{align*}
	\Vert \hat{\beta}_{in}-\alpha _{l}\Vert ^{2}& \geq \Vert \beta
	_{g_{i}^{0}n}-\alpha _{l}\Vert ^{2} - 2|(\beta _{g_{i}^{0}n}-\hat{\beta}%
	_{in})^{T}(\beta _{g_{i}^{0}n}-\alpha _{l})| - \Vert \beta _{g_{i}^{0}n}-%
	\hat{\beta}_{in}\Vert ^{2} \\
	& \geq \Vert \beta _{g_{i}^{0}n}-\alpha _{l}\Vert ^{2} - 2\Vert \beta
	_{g_{i}^{0}n}-\hat{\beta}_{in}\Vert _{1}\Vert \beta _{g_{i}^{0}n}-\alpha
	_{l}\Vert _{\infty } -R_{n}^{2} \\
	& \geq \Vert \beta _{g_{i}^{0}n}-\alpha _{l}\Vert ^{2} - 2\sqrt{K}R_{n}\Vert
	\beta _{g_{i}^{0}n}-\alpha _{l}\Vert - R_{n}^{2} \\
	& \geq \Vert \beta _{g_{i}^{0}n}-\alpha _{l}\Vert ^{2} - 2\sqrt{K}%
	R_{n}(\Vert \beta _{g_{i}^{0}n}\Vert +\Vert \alpha _{l}\Vert )- R_{n}^{2},
	\end{align*}%
	where the third inequality follows the Cauchy–Schwarz inequality with the
	fact that both $\beta_{g_i^0n}$ and $\hat{\beta}_{in}$ are $K \times 1$
	vectors. Taking $\min_{1\leq l\leq K}$ on both sides and averaging over $i$,
	we have 
	\begin{equation*}
	\widehat{Q}_{n}(\mathcal{A})\geq Q_{n}(\mathcal{A})-(6\sqrt{K}+1)M c_{2n}.
	\end{equation*}%
	Similarly, we have $\widehat{Q}_{n}(\mathcal{A})\leq Q_{n}(\mathcal{A})+(6%
	\sqrt{K}+1)M c_{2n}.$ By \eqref{eq:Rn}, 
	\begin{equation*}
	\breve{R}_{n}\equiv \sup_{\mathcal{A}\in \mathcal{M}}|\widehat{Q}_{n}(%
	\mathcal{A})-Q_{n}(\mathcal{A})| \leq (6\sqrt{K}+1)M c_{2n}\quad a.s.
	\end{equation*}
	
	Next, we show $\widehat{\mathcal{A}}_{n}\in \mathcal{M}$. Denote $\widehat{%
		\mathcal{A}}_{n}=\{\widehat{\alpha }_{1},\ldots ,\widehat{\alpha }_{K}\}$.
	By Assumption \ref{ass:theta}.1, 
	\begin{equation*}
	\sup_{i}\Vert \hat{\beta}_{in}\Vert \leq R_{n}+\sup_{1\leq k\leq K}\Vert
	\beta _{kn}\Vert \leq 2M.
	\end{equation*}%
	Denote $I_{n}(k)=\{i:k=\argmin_{1\leq l\leq K}\Vert \hat{\beta}_{in}-%
	\widehat{\alpha }_{l}\Vert \}$ for some $k\leq K$. If $\Vert \widehat{\alpha 
	}_{k}\Vert >2M$ and $I_{n}(k)=\emptyset $, then we can choose 
	\begin{equation*}
	\widehat{\mathcal{A}}_{n}^{\prime }=\{\widehat{\alpha }_{1},\ldots ,\widehat{%
		\alpha }_{k-1},\widehat{\alpha }_{k}^{\prime },\widehat{\alpha }%
	_{k+1},\ldots ,\widehat{\alpha }_{K}\},
	\end{equation*}%
	where $\widehat{\alpha }_{k}^{\prime }=\hat{\beta}_{in}$ for some arbitrary $%
	i\leq n$. Therefore, we have $\Vert \widehat{\alpha }_{k}^{\prime }\Vert
	\leq 2M<\Vert \widehat{\alpha }_{k}\Vert $ and $\widehat{Q}_{n}(\widehat{%
		\mathcal{A}}_{n}^{\prime })<\widehat{Q}_{n}(\widehat{\mathcal{A}}_{n}),$
	which is a contradiction. On the other hand, if $\Vert \widehat{\alpha }%
	_{k}\Vert >2M$ and $I_{n}(k)\neq \emptyset $, then we can choose 
	\begin{equation*}
	\widehat{\mathcal{A}}_{n}^{\prime }=\{\widehat{\alpha }_{1},\ldots ,\widehat{%
		\alpha }_{k-1},\widehat{\alpha }_{k}^{\prime },\widehat{\alpha }%
	_{k+1},\ldots ,\widehat{\alpha }_{K}\},
	\end{equation*}%
	where $\widehat{\alpha }_{k}^{\prime }=\frac{1}{|I_{n}(k)|}\sum_{i\in I_{n}}%
	\hat{\beta}_{in}$ and $|I_{n}(k)|$ is the cardinality of $I_{n}(k)$. This
	means $\Vert \widehat{\alpha }_{k}^{\prime }\Vert \leq 2M<\Vert \widehat{%
		\alpha }_{k}\Vert $ and $\widehat{Q}_{n}(\widehat{\mathcal{A}}_{n}^{\prime
	})<\widehat{Q}_{n}(\widehat{\mathcal{A}}_{n})$, which is a contradiction
	too. Therefore, $\Vert \widehat{\alpha }_{k}\Vert \leq 2M$. Since $k$ is
	arbitrary, $\widehat{\mathcal{A}}_{n}\in \mathcal{M}$.
	
	Third, we show for any $\eta>0$, 
	\begin{equation}  \label{eq:QA}
	\inf_{\mathcal{A}:H(\mathcal{A},\mathcal{B}_{n})>\eta}Q_{n}(\mathcal{A})\geq 
	\frac{c_1}{K}\min (\eta^2,c_{1n}^{2}/2),
	\end{equation}%
	where $\mathcal{B}_{n}=\{\beta _{1n},\ldots ,\beta _{Kn}\}$ and $c_1$ is the
	constant defined in Assumption \ref{ass:theta}.2. If there exist some $%
	l_{0}\in \{1,\ldots ,K\}$ and two indexes $k_{1}$ and $k_{2}$ such that 
	\begin{equation*}
	l_{0}=\argmin_{1\leq l\leq K}\Vert \beta _{k_{1}n}-\alpha _{l}\Vert =\argmin%
	_{1\leq l\leq K}\Vert \beta _{k_{2}n}-\alpha _{l}\Vert ,
	\end{equation*}%
	then by Assumption \ref{ass:theta}.2 
	\begin{align*}
	Q_{n}(\mathcal{A})& \geq \pi _{k_{1}n}\Vert \beta _{k_{1}n}-\alpha
	_{l_{0}}\Vert ^{2}+\pi _{k_{2}n}\Vert \beta _{k_{2}n}-\alpha _{l_{0}}\Vert
	^{2} \\
	& \geq \frac{c_1}{2K}(\Vert \beta _{k_{1}n}-\alpha _{l_{0}}\Vert +\Vert
	\beta _{k_{2}n}-\alpha _{l_{0}}\Vert )^{2}\geq \frac{c_1}{2K}\Vert \beta
	_{k_{1}n}-\beta _{k_{2},n}\Vert ^{2}\geq \frac{c_1c_{1n}^{2}}{2K}.
	\end{align*}%
	On the other hand, if there does not exist such an $l_{0}$, then there is a
	one-to-one mapping $h:\{1,\ldots ,K\}\mapsto \{1,\ldots ,K\}$ such that 
	\begin{equation*}
	h(k)=\argmin_{1\leq l\leq K}\Vert \beta _{kn}-\alpha _{l}\Vert .
	\end{equation*}%
	Therefore, 
	\begin{equation*}
	Q_{n}(\mathcal{A})=\sum_{k=1}^{K}\pi _{kn}\Vert \beta _{kn}-\alpha
	_{h(k)}\Vert ^{2}\geq (\inf_{k}\pi _{kn})H^{2}(\mathcal{A},\mathcal{B}%
	_{n})\geq c_1\eta^2/K.
	\end{equation*}
	
	Last, we show $H(\widehat{\mathcal{A}}_{n},\mathcal{B}_{n}) \leq
	(15M/c_1)^{1/2}c_{2n}^{1/2}K^{3/4}$. For any $\varepsilon >0$ and
	sufficiently large $C_{2}$, 
	\begin{align*}
	\hspace{2em}& \hspace{-2em}P(H(\widehat{\mathcal{A}}_{n},\mathcal{B}%
	_{n})\geq (15M/c_1)^{1/2}c_{2n}^{1/2}K^{3/4}\quad i.o.) \\
	& =P(H(\widehat{\mathcal{A}}_{n},\mathcal{B}_{n})\geq
	(15M/c_1)^{1/2}c_{2n}^{1/2}K^{3/4},Q_{n}(\widehat{\mathcal{A}}_{n})\geq
	Q_{n}(\mathcal{B}_{n}) \\
	& +\min (15Mc_{2n}K^{1/2},c_1c_{1n}^{2}/(2K))\quad i.o.) \\
	& \leq P(\widehat{Q}_{n}(\widehat{\mathcal{A}}_{n})+\breve{R}_{n}\geq 
	\widehat{Q}_{n}(\mathcal{B}_{n})-\breve{R}_{n}+ \min
	(15Mc_{2n}K^{1/2},c_1c_{1n}^{2}/(2K))\quad i.o.) \\
	& =P(2\breve{R}_{n}\geq \widehat{Q}_{n}(\mathcal{B}_{n})-\widehat{Q}_{n}(%
	\widehat{\mathcal{A}}_{n})+\min (15Mc_{2n}K^{1/2},c_1c_{1n}^{2}/(2K))\quad
	i.o.) \\
	& \leq P(2\breve{R}_{n}\geq \min (15Mc_{2n}K^{1/2},c_1c_{1n}^{2}/(2K))\quad
	i.o.) = 0,
	\end{align*}%
	where the first equality holds due to \eqref{eq:QA} and the fact that $Q_n(%
	\mathcal{B}_n) = 0$, the last inequality holds because $\widehat{Q}_{n}(\mathcal{B}%
	_{n})-\widehat{ Q}_{n}(\widehat{\mathcal{A}}_{n})\geq 0$, and the last
	equality holds because, by Assumption \ref{ass:theta}.3, 
	\begin{align*}
	2\breve{R}_{n} \leq 2(6\sqrt{K}+1)M c_{2n} < 15 \sqrt{K} M c_{2n} \leq
	c_1c_{1n}^2/(2K).
	\end{align*}
	
	This concludes the proof. \bigskip
\end{proof}

\begin{proof}[\textbf{Proof of Theorem \protect\ref{thm:strong}}]
	By Lemma \ref{lem:kmeans1} and Assumption \ref{ass:theta}.2 and (iii), for
	each $n$, there is a one-to-one mapping $F_{n}:\{1,\ldots ,K\}\mapsto
	\{1,\ldots ,K\}$, such that 
	\begin{equation*}
	\sup_{k}\Vert \widehat{\alpha }_{kn}-\beta _{F_{n}(k)n}\Vert \leq
	(15M/c_1)^{1/2}c_{2n}^{1/2}K^{3/4} \quad a.s.
	\end{equation*}%
	W.l.o.g., we can assume $F_{n}(k)=k$ such that 
	\begin{equation}
	\tilde{R}_{n}\equiv \sup_{k}\Vert \widehat{\alpha }_{kn}-\beta _{kn}\Vert
	\leq (15M/c_1)^{1/2}c_{2n}^{1/2}K^{3/4} \quad a.s.  \label{eq:alphahat}
	\end{equation}%
	If $\hat{g}_{i}\neq g_{i}^{0}$, then $\Vert \hat{\beta}_{in}-\widehat{\alpha 
	}_{\hat{g}_{i}n}\Vert \leq \Vert \hat{\beta}_{in}-\widehat{\alpha }%
	_{g_{i}^{0}n}\Vert .$ This, in conjunction with the triangle inequality,
	implies that 
	\begin{equation*}
	\Vert \widehat{\alpha }_{\hat{g}_{i}n}-\widehat{\alpha }_{g_{i}^{0}n}\Vert
	-\Vert \hat{\beta}_{in}-\widehat{\alpha }_{g_{i}^{0}n}\Vert \leq \Vert \hat{%
		\beta}_{in}-\widehat{\alpha }_{\hat{g}_{i}n}\Vert \leq \Vert \hat{\beta}%
	_{in}-\widehat{\alpha }_{g_{i}^{0}n}\Vert .
	\end{equation*}%
	It follows that $\Vert \hat{\beta}_{in}-\widehat{\alpha }_{g_{i}^{0}n}\Vert
	\geq \frac{1}{2}\Vert \widehat{\alpha }_{\hat{g}_{i}n}-\widehat{\alpha }%
	_{g_{i}^{0}n}\Vert .$ By \eqref{eq:Rn}, \eqref{eq:alphahat}, and the
	repeated use of the triangle inequality, we have 
	\begin{align*}
	c_{2n}+\tilde{R}_n& \geq \Vert \hat{\beta}_{in}-\beta _{g_{i}^{0}n}\Vert
	+\Vert \beta _{g_{i}^{0}n}-\widehat{\alpha }_{g_{i}^{0}n}\Vert \text{ } \\
	& \geq \Vert \hat{\beta}_{in}-\widehat{\alpha }_{g_{i}^{0}n}\Vert \geq \frac{%
		1}{2}\Vert \widehat{\alpha }_{\hat{g}_{i}n}-\widehat{\alpha }%
	_{g_{i}^{0}n}\Vert \text{ } \\
	& =\frac{1}{2}\Vert (\beta _{\hat{g}_{i}n}-\beta _{g_{i}^{0}n})+(\widehat{%
		\alpha }_{\hat{g}_{i}n}-\beta _{\hat{g}_{i}n})+(\beta _{g_{i}^{0}n}-\widehat{%
		\alpha }_{g_{i}^{0}n})\Vert \\
	& \geq \frac{1}{2}\Vert \beta _{\hat{g}_{i}n}-\beta _{g_{i}^{0}n}\Vert -%
	\tilde{R}_{n}\geq c_{1n}/2-\tilde{R}_{n}.
	\end{align*}
	
	This implies $1\{\hat{g}_{i}\neq g_{i}^{0}\}\leq 1\{R_{n}+2\tilde{R}_{n}\geq
	c_{1n}/2\}.$ Noting that the RHS of the above display is independent of $i$,
	we have 
	\begin{align*}
	P(\sup_{i}1\{\hat{g}_{i}\neq g_{i}^{0}\}>0\quad i.o.)& \leq P(c_{2n}+2\tilde{%
		R}_{n}\geq c_{1n}/2\quad i.o.) \\
	& =P(c_{2n}+2(15M/c_{1})^{1/2}c_{2n}^{1/2}K^{3/4}\geq c_{1n}/2\quad i.o.) \\
	& =0\text{ under Assumption \ref{ass:theta}.3.}
	\end{align*}%
	This concludes the proof.\bigskip
\end{proof}

\begin{proof}[\textbf{Proof of Corollary \protect\ref{cor:sc}}]
	We note that $\beta _{kn}=(K\pi _{kn})^{-1/2}[S_{n}\hat{O}_{n}^{T}]_{k\cdot
	} $, $M=||\beta _{kn}||\leq c_{1}^{-1/2}$, $\hat{\beta}_{in}=(n/K)^{1/2}\hat{%
		u}_{1i}^{T}$, $c_{1n}=C_{1}^{-1/2}\sqrt{2}>0$, and 
	\begin{equation*}
	c_{2n}=C^{\ast }\frac{\rho _{n}\log ^{1/2}(n)}{\mu _{n}^{1/2}\sigma _{Kn}^{2}%
	}\left( 1+\rho _{n}+\left( \frac{1}{K}+\frac{\log (5)}{\log (n)}\right)
	^{1/2}\rho _{n}^{1/2}\right) .
	\end{equation*}%
	Then, by Theorem \ref{thm:strong} and Assumption \ref{ass:rate}, we have 
	\begin{equation*}
	(2c_{2n}c_{1}^{1/2}+16K^{3/4}M^{1/2}c_{2n}^{1/2})^{2}\leq
	16.02^{2}K^{3/2}Mc_{2n}\leq 257K^{3/2}c_{1}^{-1/2}c_{2n}\leq
	2c_{1}C_{1}^{-1},
	\end{equation*}%
	where the first inequality holds because by Assumption \ref{ass:rate} and
	the facts that $C^{\ast }=3528C_1c_1^{-1/2}$ and 
	\begin{equation*}
	2c_{2n}c_{1}^{1/2}\leq 2 (10^{-8}C_1^{-1}c_1^{1/2}C^{\ast
	})^{1/2}c_{2n}^{1/2}\leq 0.02c_{2n}^{1/2}.
	\end{equation*}%
	This verifies Assumption \ref{ass:theta}.3.$\ $\bigskip
\end{proof}

\begin{proof}[\textbf{Proof of Lemma \protect\ref{lem:kmed}}]
	Following the first step in the proof of Lemma \ref{lem:kmeans1}, we can
	show that 
	\begin{align}  \label{eq:rbreve}
	\breve{R}_n \equiv \sup_{\mathcal{A}}|\widetilde{Q}_n(\mathcal{A}) - Q_n(%
	\mathcal{A})| \leq c_{2n}.
	\end{align}
	Suppose $H(\mathcal{A},\mathcal{B}_{n}) \geq \eta$ for any $\eta>0$. Then,
	by Step 3 in the proof of Lemma \ref{lem:kmeans1}, if there exist some $%
	l_{0}\in \{1,\ldots ,K\}$ and two indexes $k_{1}$ and $k_{2}$ such that 
	\begin{equation*}
	l_{0}=\argmin_{1\leq l\leq K}\Vert \beta _{k_{1}n}-\alpha _{l}\Vert =\argmin%
	_{1\leq l\leq K}\Vert \beta _{k_{2}n}-\alpha _{l}\Vert ,
	\end{equation*}%
	then by Assumption \ref{ass:theta}.2 
	\begin{align*}
	Q_{n}(\mathcal{A})& \geq \pi _{k_{1}n}\Vert \beta _{k_{1}n}-\alpha
	_{l_{0}}\Vert +\pi _{k_{2}n}\Vert \beta _{k_{2}n}-\alpha _{l_{0}}\Vert \\
	& \geq \frac{c_1}{K}\Vert \beta _{k_{1}n}-\beta _{k_{2},n}\Vert \geq
	c_1c_{1n}/K.
	\end{align*}%
	On the other hand, if there does not exist such an $l_{0}$, then there is a
	one-to-one mapping $h:\{1,\ldots ,K\}\mapsto \{1,\ldots ,K\}$ such that 
	\begin{equation*}
	h(k)=\argmin_{1\leq l\leq K}\Vert \beta _{kn}-\alpha _{l}\Vert .
	\end{equation*}%
	Therefore, 
	\begin{equation*}
	Q_{n}(\mathcal{A})=\sum_{k=1}^{K}\pi _{kn}\Vert \beta _{kn}-\alpha
	_{h(k)}\Vert\geq (\inf_{k}\pi _{kn})H(\mathcal{A},\mathcal{B}_{n})\geq
	c_1\eta/K
	\end{equation*}
	and 
	\begin{align}  \label{eq:kmed}
	\inf_{\mathcal{A}:H(\mathcal{A},\mathcal{B}_{n})>\eta}Q_{n}(\mathcal{A})
	\geq \frac{c_1 (c_{1n} \wedge \eta)}{K}.
	\end{align}
	
	By Step 4 of the proof of Lemma \ref{lem:kmeans1} and letting $\eta =\frac{%
		3Kc_{2n}}{c_{1}}$, we have 
	\begin{align*}
	& P(H(\widetilde{\mathcal{A}}_{n},\mathcal{B}_{n})\geq \frac{3Kc_{2n}}{c_{1}}%
	\quad i.o.) \\
	=& P(H(\widetilde{\mathcal{A}}_{n},\mathcal{B}_{n})\geq \frac{3Kc_{2n}}{c_{1}%
	},Q_{n}(\widetilde{\mathcal{A}}_{n})\geq Q_{n}(\mathcal{B}_{n})+\min \frac{%
		c_{1}(c_{1n}\wedge \frac{3Kc_{2n}}{c_{1}})}{K}\quad i.o.) \\
	\leq & P(\widetilde{Q}_{n}(\widetilde{\mathcal{A}}_{n})+\breve{R}_{n}\geq 
	\widetilde{Q}_{n}(\mathcal{B}_{n})-\breve{R}_{n}+3c_{2n}\quad i.o.) \\
	\leq & P(2\breve{R}_{n}\geq 3c_{2n}\quad i.o.) \\
	\leq & P(2c_{2n}\geq 3c_{2n}\quad i.o.)=0,
	\end{align*}%
	where the first equality is due to \eqref{eq:kmed}, the first inequality is
	due to Assumption \ref{ass:theta2}.2, the second inequality is because $%
	\widetilde{Q}_{n}(\widetilde{\mathcal{A}}_{n})\leq \widetilde{Q}_{n}(%
	\mathcal{\beta }_{n})$, and the third inequality is due to \eqref{eq:rbreve}%
	.\bigskip
\end{proof}

\begin{proof}[\textbf{Proof of Theorem \protect\ref{thm:strong2}}]
	By Lemma \ref{lem:kmed} and Assumption \ref{ass:theta}.2, for each $n$,
	there is a one-to-one mapping $F_{n}:\{1,\ldots ,K\}\mapsto \{1,\ldots ,K\}$%
	, such that 
	\begin{equation*}
	\sup_{k}\Vert \widetilde{\alpha }_{kn}-\beta _{F_{n}(k)n}\Vert \leq
	3Kc_1^{-1}c_{2n} \quad a.s.
	\end{equation*}%
	W.l.o.g., we can assume $F_{n}(k)=k$ such that 
	\begin{equation}
	\tilde{R}_{n}\equiv \sup_{k}\Vert \widetilde{\alpha }_{kn}-\beta _{kn}\Vert
	\leq 3Kc_1^{-1}c_{2n} \quad a.s.  \label{eq:alphahat1}
	\end{equation}%
	If $\tilde{g}_{i}\neq g_{i}^{0}$, then $\Vert \hat{\beta}_{in}-\widetilde{%
		\alpha }_{\tilde{g}_{i}n}\Vert \leq \Vert \hat{\beta}_{in}-\widetilde{\alpha 
	}_{g_{i}^{0}n}\Vert .$ This, in conjunction with the triangle inequality,
	implies that 
	\begin{equation*}
	\Vert \widetilde{\alpha }_{\tilde{g}_{i}n}-\widetilde{\alpha }%
	_{g_{i}^{0}n}\Vert -\Vert \hat{\beta}_{in}-\widetilde{\alpha }%
	_{g_{i}^{0}n}\Vert \leq \Vert \hat{ \beta}_{in}-\widetilde{\alpha }_{\tilde{g%
		}_{i}n}\Vert \leq \Vert \hat{\beta}_{in}-\widetilde{\alpha }%
	_{g_{i}^{0}n}\Vert .
	\end{equation*}%
	It follows that $\Vert \hat{\beta}_{in}-\widetilde{\alpha }%
	_{g_{i}^{0}n}\Vert \geq \frac{1}{2}\Vert \widetilde{\alpha }_{\tilde{g}%
		_{i}n}-\widetilde{\alpha }_{g_{i}^{0}n}\Vert .$ By \eqref{eq:rbreve}, %
	\eqref{eq:alphahat1}, and the repeated use of the triangle inequality, we
	have 
	\begin{align*}
	c_{2n}+\tilde{R}_n& \geq \Vert \hat{\beta}_{in}-\beta _{g_{i}^{0}n}\Vert
	+\Vert \beta _{g_{i}^{0}n}-\widetilde{\alpha }_{g_{i}^{0}n}\Vert \text{ } \\
	& \geq \Vert \hat{\beta}_{in}-\widetilde{\alpha }_{g_{i}^{0}n}\Vert \geq 
	\frac{ 1}{2}\Vert \widetilde{\alpha }_{\tilde{g}_{i}n}-\widetilde{\alpha }%
	_{g_{i}^{0}n}\Vert \text{ } \\
	& =\frac{1}{2}\Vert (\beta _{\tilde{g}_{i}n}-\beta _{g_{i}^{0}n})+(%
	\widetilde{ \alpha }_{\tilde{g}_{i}n}-\beta _{\tilde{g}_{i}n})+(\beta
	_{g_{i}^{0}n}-\widetilde{ \alpha }_{g_{i}^{0}n})\Vert \\
	& \geq \frac{1}{2}\Vert \beta _{\tilde{g}_{i}n}-\beta _{g_{i}^{0}n}\Vert -%
	\tilde{R}_{n}\geq c_{1n}/2-\tilde{R}_{n}.
	\end{align*}
	
	This implies $1\{\tilde{g}_{i}\neq g_{i}^{0}\}\leq 1\{c_{2n}+2\tilde{R}%
	_{n}\geq c_{1n}/2\}.$ Noting that the RHS of the above display is
	independent of $i$, we have 
	\begin{align*}
	P(\sup_{i}1\{\tilde{g}_{i}\neq g_{i}^{0}\}>0\quad i.o.)& \leq P(c_{2n}+2%
	\tilde{R}_{n}\geq c_{1n}/2\quad i.o.) \\
	& =P(c_{2n}+6Kc_{1}^{-1}c_{2n}\geq c_{1n}/2\quad i.o.) \\
	& =0\text{ under Assumption \ref{ass:theta2}.2.}
	\end{align*}%
	This concludes the proof.\bigskip
\end{proof}

\begin{proof}[Proof of Corollary \protect\ref{cor:sc2}]
	By Theorems \ref{thm:id} and \ref{thm:main}, we have $c_{1n} = C_1^{-1/2}%
	\sqrt{2}$ and 
	\begin{equation*}
	c_{2n} = C^*\frac{\rho _{n}\log ^{1/2}(n)}{\mu _{n}^{1/2}\sigma _{Kn}^{2}}%
	\left(1+\rho_n + \left(\frac{1}{K} + \frac{\log(5)}{\log(n)}%
	\right)^{1/2}\rho_n^{1/2}\right).
	\end{equation*}
	Then, the result directly follows Theorem \ref{thm:strong2}.
\end{proof}

\section{Proofs of the results in Section \protect\ref{sec:ext}\label%
	{sec:31pf}}

In this appendix, we prove the main results in Section \ref{sec:ext}, viz.,
Theorems \ref{thm:id2}-\ref{thm:thetahat}, Lemma \ref{lem:dk3}, and
Corollary \textbf{\ref{cor:DC}. }The proof of Lemma \ref{lem:dk3} calls upon
Lemma\textbf{\ }\ref{lem:matrixBernstein} and that of Theorem \ref%
{thm:main_DC} calls upon Lemmas \ref{lem:Pij3}, \ref{lem:V1n3} and \ref%
{lem:B3} in Appendix \ref{sec:lem}. Theorems \ref{thm:id2} and \ref%
{thm:regularization} can be proved in the same manner as Theorems \ref%
{thm:id} and \ref{thm:main}, respectively, while Theorem \ref{thm:main} is a
special case of Theorem \ref{thm:main_DC}. Therefore, the key part of this
section is to prove Lemma \ref{lem:dk3} and Theorem \ref{thm:main_DC}.
\bigskip

\begin{proof}[\textbf{Proof of Theorem \protect\ref{thm:id2}}]
	Since $\mathcal{L}_{\tau }=n^{-1}ZB_{0}^{\tau }Z$, the proof follows that of
	Theorem \ref{thm:id} with $A,$ $B_{0},$ and $S_{n}$ replaced by $A_{\tau },$ 
	$B_{0}^{\tau },$ and $S_{n}^{\tau },$ respectively.\bigskip
\end{proof}

\begin{proof}[\textbf{Proof of Theorem \protect\ref{thm:regularization}}]
	The proof of part (i) is analogous to that of Theorem \ref{thm:main}. The
	main difference is that we need to use Theorem \ref{thm:id2} in place of
	Theorem \ref{thm:id}.
	
	Theorem \ref{thm:id2} and the first part of Theorem \ref{thm:regularization}
	verify Assumptions \ref{ass:theta}.1 and (ii) and Assumption \ref{ass:theta}
	(iii), respectively, with $\beta _{kn}=(K\pi _{kn})^{-1/2}[S_{n}^{\tau }(%
	\hat{O}_{n}^{\tau })^{T}]_{k\cdot }$ and $\hat{\beta}_{in}=(n/K)^{1/2}(\hat{u%
	}_{1i}^{\tau })^{T}.$ Assumption \ref{ass:nk} is maintained. Then part (ii)
	follows from Theorem \ref{thm:strong}.\bigskip
\end{proof}

To prove the results in Section \ref{sec:ext2}, we follow the notation
there. In particular, we consider the spectral decomposition of $\mathcal{L}%
_{\tau }^{\prime }:$ 
\begin{equation*}
\mathcal{L}_{\tau }^{\prime }=U_{1n}\Sigma _{n}U_{1n}^{T},
\end{equation*}%
where $\Sigma _{n}=\text{diag}(\sigma _{1n},\ldots ,\sigma _{Kn})$ is a $%
K\times K$ matrix that contains the eigenvalues of $\mathcal{L}_{\tau
}^{\prime }$ such that $|\sigma _{1n}|\geq |\sigma _{2n}|\geq \cdots \geq
|\sigma _{Kn}|>0 $ and $U_{1n}^{T}U_{1n}=I_{K}$. The sample normalized graph
Laplacian is denoted as $L_{\tau }^{\prime }$. We consider the spectral
decomposition 
\begin{equation*}
L_{\tau }^{\prime }=\hat{U}_{n}\hat{\Sigma}_{n}\hat{U}_{n}^{T}=\hat{U}_{1n}%
\hat{\Sigma}_{1n}\hat{U}_{1n}^{T}+\hat{U}_{2n}\hat{\Sigma}_{2n}\hat{U}%
_{2n}^{T},
\end{equation*}%
where $\hat{\Sigma}_{n}=\text{diag}(\hat{\sigma}_{1n},\ldots ,\hat{\sigma}%
_{nn})=\text{diag}(\hat{\Sigma}_{1n},\hat{\Sigma}_{2n})$ with $|\hat{\sigma}%
_{1n}|\geq |\hat{\sigma}_{2n}|\geq \cdots \geq |\hat{\sigma}_{nn}|\geq 0,$ $%
\hat{\Sigma}_{1n}=\text{diag}(\hat{\sigma}_{1n},\ldots ,\hat{\sigma}_{Kn})$, 
$\hat{\Sigma}_{2n}=\text{diag}(\hat{\sigma}_{K+1,n},\ldots ,\hat{\sigma}%
_{nn})$, and $\hat{U}_{n}=(\hat{U}_{1n},\hat{U}_{2n})$ is the corresponding
eigenvectors such that $\hat{U}_{1n}^{T}\hat{U}_{1n}=I_{K}$ and $\hat{U}%
_{2n}^{T}\hat{U}_{1n}=0.\medskip $

\begin{proof}[\textbf{Proof of Theorem \protect\ref{thm:id3}}]
	Let $g_{i}^{0}\in \{1,\ldots ,K\}$ denote node $i$'s membership. Similar to %
	\citet[][Lemma 3.2]{QR13}, we have by (\ref{eq:thetanormalization}) 
	\begin{equation}
	d_{i}=\sum_{j=1}^{n}P_{ij}=\theta _{i}\sum_{j=1}^{n}\theta
	_{j}B_{g_{i}^{0}g_{j}^{0}}=\theta _{i}\sum_{k=1}^{K}\sum_{j\in C_{k}}\theta
	_{j}B_{g_{i}^{0}k}=n\theta _{i}\sum_{k=1}^{K}\pi _{kn}B_{g_{i}^{0}k}=n\theta
	_{i}W_{g_{i}^{0}}.  \label{eq:d_i}
	\end{equation}%
	Therefore, 
	\begin{align*}
	\lbrack \mathcal{L}_{\tau }^{\prime }]_{ij}=P_{ij}((d_{i}+\tau )(d_{j}+\tau
	))^{-1/2}& =B_{g_{i}^{0}g_{j}^{0}}(\theta _{i}\theta _{j})((d_{i}+\tau
	)(d_{j}+\tau ))^{-1/2} \\
	& =B_{g_{i}^{0}g_{j}^{0}}(\theta _{i}^{\tau }\theta _{j}^{\tau
	})^{1/2}(\theta _{i}\theta _{j})^{1/2}(d_{i}d_{j})^{-1/2} \\
	& =n^{-1}B_{g_{i}^{0}g_{j}^{0}}(\theta _{i}^{\tau }\theta _{j}^{\tau
	})^{1/2}(W_{g_{i}^{0}}W_{g_{j}^{0}})^{-1/2} \\
	& =n^{-1}[\Theta _{\tau }^{1/2}Z\mathcal{D}_{B}^{-1/2}B\mathcal{D}%
	_{B}^{-1/2}Z^{T}\Theta _{\tau }^{1/2}]_{ij} \\
	& =n^{-1}[\Theta _{\tau }^{1/2}ZB_{0}Z^{T}\Theta _{\tau }^{1/2}]_{ij}.
	\end{align*}%
	That is, $\mathcal{L}_{\tau }^{\prime }=n^{-1}\Theta _{\tau
	}^{1/2}ZB_{0}Z^{T}\Theta _{\tau }^{1/2}.$ Then 
	\begin{equation*}
	(\mathcal{L}_{\tau }^{\prime })^{2}=n^{-1}\Theta _{\tau
	}^{1/2}ZB_{0}(Z^{T}\Theta _{\tau }Z/n)B_{0}Z^{T}\Theta _{\tau
	}^{1/2}=n^{-1}\Theta _{\tau }^{1/2}ZB_{0}\Pi _{n}^{\tau }B_{0}Z^{T}\Theta
	_{\tau }^{1/2},
	\end{equation*}%
	where $\Pi _{n}^{\tau }=Z^{T}\Theta _{\tau }Z/n=\text{diag}(\pi _{1n}^{\tau
	},\ldots ,\pi _{Kn}^{\tau }),$ and $\pi _{kn}^{\tau }=n_{k}^{\tau
	}/n=\sum_{i\in C_{k}}\theta _{i}^{\tau }/n$. By the spectral decomposition,
	we have 
	\begin{equation}
	(\Pi _{n}^{\tau })^{1/2}B_{0}\Pi _{n}^{\tau }B_{0}(\Pi _{n}^{\tau
	})^{1/2}=S_{n}^{\tau }\Omega _{n}(S_{n}^{\tau })^{T},  \label{eq:sigma3}
	\end{equation}%
	where $\Omega _{n}=\text{diag}(\omega _{n},\ldots ,\omega _{Kn})$ such that $%
	\omega _{n}\geq \omega _{2n}\geq \cdots \geq \omega _{Kn}>0$ and $S^\tau_{n}$
	is a $K\times K$ matrix such that $(S_{n}^{\tau })^{T}S_{n}^{\tau }=I_{K}.$
	Let $U^*_{1n}=\Theta _{\tau }^{1/2}Z(Z^{T}\Theta _{\tau
	}Z)^{-1/2}S_{n}^{\tau }$. Then, we have 
	\begin{equation*}
	U^*_{1n}\Omega _{n}U_{1n}^{\ast T}=(\mathcal{L}_{\tau }^{\prime
	})^{2}=U_{1n}\Sigma _{n}^{2}U_{1n}^{T}.
	\end{equation*}%
	In addition, $U_{1n}^{\ast T}U^\ast_{1n}=(S_{n}^{\tau })^{T}S_{n}^{\tau
	}=I_{K}$. Therefore the columns of $U^\ast_{1n}$ are the eigenvectors of $%
	\mathcal{L}_{\tau }^{\prime }$ associated with eigenvalues $\sigma
	_{n},\ldots ,\sigma _{Kn}$, up to sign normalization. W.l.o.g., we can take $%
	U_{1n}=U^\ast_{1n}$ to obtain the first result.
	
	Now we turn to the second result. If node $i$ is in cluster $C_{k_{1}}$,
	then 
	\begin{equation*}
	u_{i}^{T}=(\theta _{i}^{\tau })^{1/2}z_{i}^{T}(Z^{T}\Theta _{\tau
	}Z)^{-1/2}S_{n}^{\tau }=(\theta _{i}^{\tau })^{1/2}(n_{k_{1}}^{\tau
	})^{-1/2}[S_{n}^{\tau }]_{k_{1}\cdot },
	\end{equation*}%
	where $[S_{n}^{\tau }]_{k\cdot }$ denotes the $k$-th row of $S_{n}^{\tau }$.
	Therefore, 
	\begin{equation*}
	(n_{k_{1}}^{\tau })^{1/2}(\theta _{i}^{\tau })^{-1/2}\Vert u_{i}^{T}\Vert
	=\Vert \lbrack S_{n}^{\tau }]_{k_{1}\cdot }\Vert =1.
	\end{equation*}
	
	Last, we note that $\frac{u_{i}^{T}}{\Vert u_{i}^{T}\Vert }=[S_{n}^{\tau
	}]_{g_{i}^{0}\cdot }$. Therefore, if $z_{i}\neq z_{j}$, then $g_{i}^{0}\neq
	g_{j}^{0}$ and 
	\begin{equation*}
	\biggl\Vert\frac{u_{i}^{T}}{\Vert u_{i}^{T}\Vert }-\frac{u_{1j}^{T}}{\Vert
		u_{1j}^{T}\Vert }\biggr\Vert=\Vert \lbrack S_{n}^{\tau }]_{g_{i}^{0}\cdot
	}-[S_{n}^{\tau }]_{g_{j}^{0}\cdot }\Vert =\sqrt{2}.
	\end{equation*}%
	Similarly, if $z_{i}=z_{j}$, then $g_{i}^{0}=g_{j}^{0}$ and $\frac{u_{i}^{T}%
	}{\Vert u_{i}^{T}\Vert }=\frac{u_{1j}^{T}}{\Vert u_{1j}^{T}\Vert }.\medskip $
\end{proof}

Lemma \ref{lem:dk3} derives an upper bound for spectral norm of the gap
between the first $K$ columns of sample and population eigenvectors. By
Lemma \ref{lem:matrixBernstein}, we first derive the upper bound for
spectral norm of the gap between sample and population graph Laplacians.
Then, we use the Davis-Kahan theorem (Lemma \ref{lem:davis}) to establish
the bound for the eigenvectors.

\bigskip

\begin{proof}[\textbf{Proof of Lemma \protect\ref{lem:dk3}}]
	The proof is similar to that in \cite{JY16} and \cite{QR13}. Let $\tilde{L}%
	_{\tau }=\mathcal{D}_{\tau }^{-1/2}A\mathcal{D}_{\tau }^{-1/2}$. Then 
	\begin{equation*}
	\Vert \mathcal{L}_{\tau }^{\prime }-L_{\tau }^{\prime }\Vert \leq \Vert 
	\mathcal{L}_{\tau }^{\prime }-\tilde{L}_{\tau }\Vert +\Vert L_{\tau
	}^{\prime }-\tilde{L}_{\tau }\Vert \equiv I+II.
	\end{equation*}%
	Let $d_{i}^{\tau }=d_{i}+\tau $, $Y_{ij}=(d_{i}^{\tau }d_{j}^{\tau
	})^{-1/2}(A_{ij}-P_{ij})(e_{i}e_{j}^{T}+e_{j}e_{i}^{T})$ for $1\leq i<j\leq
	n $, and $Y_{ii}=-(d_{i}^{\tau })^{-1}P_{ii}e_{i}e_{i}^{T}$, where $e_{i}$
	is the $n\times 1$ vector with its $i$-th coordinate being 1 and the rest
	being 0. Then $\{Y_{ij}\}_{1\leq i<j\leq n}$ is a sequence of independent
	symmetric random matrices such that $\mathbb{E}Y_{ij}=0$, 
	\begin{equation*}
	\tilde{L}_{\tau }-\mathcal{L}_{\tau }^{\prime }+\text{diag}(\mathcal{L}%
	_{\tau }^{\prime })=\sum_{1\leq i<j\leq n}Y_{ij},\text{ and diag}(\mathcal{L}%
	_{\tau }^{\prime })=\sum_{i=1}^{n}(d_i^\tau)^{-1}P_{ii}e_ie_i^T.
	\end{equation*}%
	In addition, we note that $\sup_{1\leq i<j\leq n}\Vert Y_{ij}\Vert \leq 
	\sqrt{2}/\mu _{n}^{\tau }$ and 
	\begin{align*}
	\sigma ^{2}=\Vert \sum_{1\leq i<j\leq n}\mathbb{E}Y_{ij}^{2}\Vert & =\Vert 
	\text{diag}(\sum_{j\neq 1}p_{1j}(1-p_{1j})/(d_{1}^{\tau }d_{j}^{\tau
	}),\ldots ,\sum_{j\neq n}p_{nj}(1-p_{nj})/(d_{n}^{\tau }d_{j}^{\tau }))\Vert
	\\
	& \leq (\mu _{n}^{\tau })^{-1}\max_{1\leq i\leq
		n}\sum_{j=1}^{n}p_{ij}(1-p_{ij})/d_{i}^{\tau }\leq (\mu _{n}^{\tau })^{-1}.
	\end{align*}%
	By Lemma \ref{lem:matrixBernstein}, for $n$ sufficiently large and $C = 2.6$%
	, we have 
	\begin{align}
	& P(\Vert \tilde{L}_{\tau }-\mathcal{L}_{\tau }^{\prime }+\text{diag}(%
	\mathcal{\ L}_{\tau }^{\prime })\Vert \geq C(\log (n)/\mu _{n}^{\tau
	})^{1/2})  \notag \\
	& =P(\Vert \sum_{1\leq i<j\leq n}Y_{ij}\Vert \geq C(\log (n)/\mu _{n}^{\tau
	})^{1/2})  \notag \\
	& \leq 2n\exp \biggl(\frac{-C^2\log (n)/\mu _{n}^{\tau }}{3(\mu _{n}^{\tau
		})^{-1}+2C\sqrt{2}(\log (n)/\mu _{n}^{\tau })^{1/2}(\mu _{n}^{\tau })^{-1}}%
	\biggr)  \notag \\
	& \leq 2n^{-1.1},  \label{eq:WnDC}
	\end{align}
	where for the last inequality, we use the fact that $(\log(n)/\mu_n^%
	\tau)^{1/2} \leq 0.01$ and $2.6^2 > 2.1\times (3+2.6\sqrt{2}/50)$. This
	implies 
	\begin{equation*}
	\sum_{n=1}^{\infty }P(\Vert \tilde{L}_{\tau }-\mathcal{L}_{\tau }^{\prime }+%
	\text{diag}(\mathcal{L}_{\tau }^{\prime })\Vert \geq 2.6(\log (n)/\mu
	_{n}^{\tau })^{1/2})<\infty ,
	\end{equation*}%
	and thus, $\Vert \tilde{L}_{\tau }-\mathcal{L}_{\tau }^{\prime }+\text{diag}(%
	\mathcal{L}_{\tau }^{\prime })\Vert \leq 2.6(\log (n)/\mu _{n}^{\tau
	})^{1/2}\quad a.s.$ In addition, for $n$ sufficiently large, 
	\begin{equation*}
	\Vert \text{diag}(\mathcal{L}_{\tau }^{\prime })\Vert \leq (\mu _{n}^{\tau
	})^{-1}\leq 0.01(\log (n)/\mu _{n}^{\tau })^{1/2}.
	\end{equation*}%
	Therefore, 
	\begin{equation}  \label{eq:I}
	I\leq \Vert \tilde{L}_{\tau }-\mathcal{L}_{\tau }^{\prime }+\text{diag}(%
	\mathcal{L}_{\tau }^{\prime })\Vert +\Vert \text{diag}(\mathcal{L}_{\tau
	}^{\prime })\Vert \leq 2.61(\log n/\mu _{n}^{\tau })^{1/2}\quad a.s.
	\end{equation}
	
	Now we turn to $II$. Let $\hat{d}_{i}^{\tau }=\hat{d}_{i}+\tau $. By
	Bernstein inequality, for some $C=2.09$, we have, 
	\begin{align}  \label{eq:d}
	P(\sup_{i}|\hat{d}_{i}^{\tau }-d_{i}^{\tau }|/d_{i}^{\tau }\geq C(\log
	(n)/\mu _{n}^{\tau })^{1/2})& \leq 2\sum_{i=1}^{n}\exp \biggl(\frac{%
		-C^{2}(d_{i}^{\tau })^{2}\log (n)/\mu _{n}^{\tau }}{2d_{i}^{\tau }+2C(\log
		n/\mu _{n}^{\tau })^{1/2}d_{i}^{\tau }/3}\biggr)  \notag \\
	& \leq 2n^{-1.1},
	\end{align}%
	where the last inequality holds because $(\log(n)/\mu_n^\tau)^{1/2} \leq
	0.01 $ and $2.09^2 > 2.1\times (2+2\times 2.09/300)$. Therefore, $\sup_{i}|%
	\hat{d}_{i}^{\tau }-d_{i}^{\tau }|/d_{i}^{\tau }\leq 2.09(\log (n)/\mu
	_{n}^{\tau })^{1/2}\quad a.s.$, and thus, 
	\begin{equation*}
	\Vert \mathcal{D}_{\tau }^{-1/2}D_{\tau }^{1/2}-I\Vert =\max_{i}|(\hat{d}%
	_{i}^{\tau }/d_{i}^{\tau })^{1/2}-1|\leq \max_{i}|(\hat{d}_{i}^{\tau
	}/d_{i}^{\tau })-1|\leq 2.09(\log (n)/\mu _{n}^{\tau })^{1/2}\quad a.s.
	\end{equation*}%
	In addition, by \citet[][Lemma 1.7]{C97}$,\Vert L_{\tau }^{\prime }\Vert
	\leq \Vert L\Vert \leq 1$. Therefore, 
	\begin{align}  \label{eq:II}
	\Vert \tilde{L}_{\tau }-L_{\tau }^{\prime }\Vert & =\Vert L_{\tau }^{\prime
	}-\mathcal{D}_{\tau }^{-1/2}D_{\tau }^{1/2}L_{\tau }^{\prime }D_{\tau }^{1/2}%
	\mathcal{D}_{\tau }^{-1/2}\Vert  \notag \\
	& \leq \Vert \mathcal{D}_{\tau }^{-1/2}D_{\tau }^{1/2}L_{\tau }^{\prime }-%
	\mathcal{D}_{\tau }^{-1/2}D_{\tau }^{1/2}L_{\tau }^{\prime }D_{\tau }^{1/2}%
	\mathcal{D}_{\tau }^{-1/2}\Vert +\Vert L_{\tau }^{\prime }-\mathcal{D}_{\tau
	}^{-1/2}D_{\tau }^{1/2}L_{\tau }^{\prime }\Vert  \notag \\
	& \leq \Vert \mathcal{D}_{\tau }^{-1/2}D_{\tau }^{1/2}-I\Vert \Vert \mathcal{%
		\ D}_{\tau }^{-1/2}D_{\tau }^{1/2}\Vert +\Vert \mathcal{D}_{\tau
	}^{-1/2}D_{\tau }^{1/2}-I\Vert  \notag \\
	\leq & 2.09(\log (n)/\mu _{n}^{\tau })^{1/2}(1+2.09(\log (n)/\mu _{n}^{\tau
	})^{1/2}) + 2.09(\log (n)/\mu _{n}^{\tau })^{1/2}  \notag \\
	\leq & 4.39(\log (n)/\mu _{n}^{\tau })^{1/2} \quad a.s.
	\end{align}%
	Combining \eqref{eq:I} and \eqref{eq:II}, we can conclude the first part of
	the proof. Then, by Lemma \ref{lem:davis} and fact that $7(\log (n)/\mu
	_{n}^{\tau })^{1/2} \leq \frac{|\sigma_{Kn}|}{100}$, we have 
	\begin{equation*}
	\Vert \hat{U}_{1n}\hat{O}_{n}-U_{1n}\Vert \leq \frac{\sqrt{2}\Vert L_{\tau
		}^{\prime }-\mathcal{L}_{\tau }^{\prime }\Vert }{0.99|\sigma _{Kn}|}\leq
	10(\log (n)/\mu _{n}^{\tau })^{1/2}|\sigma _{Kn}|^{-1}\quad a.s.
	\end{equation*}
\end{proof}

\bigskip

\begin{proof}[\textbf{Proof of Theorem \protect\ref{thm:main_DC}}]
	We aim to show the result with $C^{\ast }=3528C_1c_1^{-1/2}$. First, by
	the Hoffman-Wielandt inequality and Lemma \ref{lem:dk3} 
	\begin{equation}
	\Vert \hat{\Sigma}_{n}-\Sigma _{n}\Vert \leq \Vert L_{\tau }^{\prime }-%
	\mathcal{L}_{\tau }^{\prime }\Vert \leq 7(\log (n)/\mu _{n}^{\tau
	})^{1/2}\quad a.s.  \label{eq:LL3}
	\end{equation}%
	Then, by Lemmas \ref{lem:B3} and \ref{lem:dk3}, 
	\begin{align*}
	17(\log (n)/\mu _{n}^{\tau })^{1/2}|\sigma _{Kn}|^{-1}& \geq \Vert \hat{%
		\Lambda}-\Lambda \Vert \\
	& =\Vert \hat{U}_{1n}\hat{\Sigma}_{n}\hat{O}_{n}-U_{1n}\Sigma _{n}\Vert \\
	& \geq \Vert \hat{U}_{1n}(\hat{O}_{n}\hat{\Sigma}_{n}-\hat{\Sigma}_{n}\hat{O}%
	_{n})\Vert -\Vert (\hat{U}_{1n}\hat{O}_{n}-U_{1n})\hat{\Sigma}_{n}\Vert
	-\Vert U_{1n}(\hat{\Sigma}_{n}-\Sigma _{n})\Vert \\
	& =\Vert \hat{O}_{n}\hat{\Sigma}_{n}-\hat{\Sigma}_{n}\hat{O}_{n}\Vert
	-17(\log (n)/\mu _{n}^{\tau })^{1/2}|\sigma _{Kn}|^{-1}\quad a.s.
	\end{align*}%
	Therefore, 
	\begin{equation}
	\Vert \hat{O}_{n}\hat{\Sigma}_{n}-\hat{\Sigma}_{n}\hat{O}_{n}\Vert \leq
	34(\log (n)/\mu _{n}^{\tau })^{1/2}|\sigma _{Kn}|^{-1}\quad a.s.
	\label{eq:Osigmadc}
	\end{equation}%
	In addition, 
	\begin{align*}
	(n_{g_{i}^{0}}^{\tau })^{1/2}(\theta _{i}^{\tau })^{-1/2}\Vert \hat{\Lambda}%
	_{i}-\Lambda _{i}\Vert & =(n_{g_{i}^{0}}^{\tau })^{1/2}(\theta _{i}^{\tau
	})^{-1/2}\Vert \hat{u}_{i}^{T}\hat{\Sigma}_{n}\hat{O}_{n}-u_{i}^{T}\Sigma
	_{n}\Vert \\
	& \geq (n_{g_{i}^{0}}^{\tau })^{1/2}(\theta _{i}^{\tau })^{-1/2}\Vert (\hat{u%
	}_{i}^{T}\hat{O}_{n}-u_{i}^{T})\hat{\Sigma}_{n}\Vert -(n_{g_{i}^{0}}^{\tau
	})^{1/2}(\theta _{i}^{\tau })^{-1/2}\Vert u_{i}^{T}(\hat{\Sigma}_{n}-\Sigma
	_{n})\Vert \\
	& \quad \ -(n_{g_{i}^{0}}^{\tau })^{1/2}(\theta _{i}^{\tau })^{-1/2}\Vert 
	\hat{u}_{i}^{T}(\hat{\Sigma}_{n}\hat{O}_{n}-\hat{O}_{n}\hat{\Sigma}_{n})\Vert
	\\
	& \equiv I_{i}-II_{i}-III_{i}.
	\end{align*}%
	Next, we bound the three terms on the RHS of the above display. By
	Assumption \ref{ass:rate3} and Lemma \ref{lem:dk3}, $|\hat{ \sigma}_{Kn}|
	\geq 0.999 |\sigma_{Kn}|$ $a.s.$, and thus, 
	\begin{equation*}
	I_{i}\geq 0.999|\sigma _{Kn}|(n_{g_{i}^{0}}^{\tau })^{1/2}(\theta _{i}^{\tau
	})^{-1/2}\Vert \hat{u}_{i}^{T}\hat{O}_{n}-u_{i}^{T}\Vert \quad a.s.
	\end{equation*}%
	By Theorem \ref{thm:id3} and (\ref{eq:LL3}), 
	\begin{equation*}
	\sup_{i}II_{i}\leq \sup_{i}(n_{g_{i}^{0}}^{\tau })^{1/2}(\theta _{i}^{\tau
	})^{-1/2}\Vert u_{i}^{T}\Vert \Vert \hat{\Sigma}_{n}-\Sigma _{n}\Vert \leq
	7(\log (n)/\mu _{n}^{\tau })^{1/2}\quad a.s.
	\end{equation*}%
	Denote $\Gamma _{n}=\sup_{i}(n_{g_{i}^{0}}^{\tau })^{1/2}(\theta _{i}^{\tau
	})^{-1/2}\Vert \hat{u}_{i}^{T}\hat{O}_{n}-u_{i}^{T}\Vert $. By %
	\eqref{eq:Osigmadc} and Theorem \ref{thm:id3}.2, 
	\begin{equation*}
	\sup_{i}III_{i}\leq 34(\log (n)/\mu _{n}^{\tau })^{1/2}|\sigma
	_{Kn}|^{-1}(\Gamma _{n}+1)\quad a.s.
	\end{equation*}%
	Therefore, we have 
	\begin{align}
	\sup_{i}(n_{g_{i}^{0}}^{\tau })^{1/2}(\theta _{i}^{\tau })^{-1/2}\Vert \hat{%
		\Lambda}_{i}-\Lambda _{i}\Vert \geq & (0.999|\sigma _{Kn}|-34(\log (n)/\mu
	_{n}^{\tau })^{1/2}|\sigma _{Kn}^{-1}|)\Gamma _{n}-41(\log (n)/\mu
	_{n}^{\tau })^{1/2}|\sigma _{Kn}^{-1}|  \notag \\
	\geq & 0.99|\sigma _{Kn}|\Gamma _{n}-41(\log (n)/\mu _{n}^{\tau
	})^{1/2}|\sigma _{Kn}^{-1}|,
	\end{align}%
	where we use the fact that $34(\log (n)/\mu _{n}^{\tau })^{1/2}|\sigma
	_{Kn}^{-2}|\leq 0.09$ under Assumption \ref{ass:rate3}.2.
	
	On the other hand, if $\Gamma _{n}\leq \delta _{n}^{(0)}$ a.s. for some
	deterministic sequence $\{\delta _{n}^{(0)}\}_{n\geq 1}$, then by Theorem %
	\ref{thm:id3}.2, 
	\begin{equation*}
	\sup_{i}(n_{g_{i}^{0}}^{\tau })^{1/2}(\theta _{i}^{\tau })^{-1/2}\Vert \hat{u%
	}_{i}\Vert \leq \delta _{n}^{(0)}+1\quad a.s.
	\end{equation*}%
	Applying Lemma \ref{lem:B3} with $\psi _{n}=\delta _{n}^{(0)}+1$, we have 
	\begin{align*}
	\hspace{2em}& \hspace{-2em}3450C_{1}c_{1}^{-1/2}\rho _{n}\log ^{1/2}(n)(\mu
	_{n}^{\tau })^{-1/2}|\sigma _{Kn}^{-1}|\biggl[\delta _{n}^{(0)}+1+\rho _{n}+%
	\frac{\left( \frac{1}{K}+\frac{\log (5)}{\log (n)}\right) ^{1/2}\rho
		_{n}^{1/2}\overline{\theta }^{1/4}}{\underline{\theta }^{1/4}}\biggr] \\
	& \geq 0.99|\sigma _{Kn}|\Gamma _{n}-41(\log (n)/\mu _{n}^{\tau
	})^{1/2}|\sigma _{Kn}^{-1}|.
	\end{align*}%
	By combining and rearranging terms and the fact that $\rho _{n}\geq 1$, we
	have, 
	\begin{equation}
	\left[ 3485C_{1}c_{1}^{-1/2}\log ^{1/2}(n)(\mu _{n}^{\tau })^{-1/2}\sigma
	_{Kn}^{-2}\rho _{n}\right] \delta _{n}^{(0)}+3527C_{1}c_{1}^{-1/2}\eta
	_{n}\geq \Gamma _{n},  \label{UB}
	\end{equation}%
	where 
	\begin{equation*}
	\eta _{n}=\biggl(\frac{\rho _{n}\log ^{1/2}(n)}{\left( \mu _{n}^{\tau
		}\right) ^{1/2}\sigma _{Kn}^{2}}\biggr)\biggl(\frac{\left( \frac{1}{K}+\frac{%
			\log (5)}{\log (n)}\right) ^{1/2}\rho _{n}^{1/2}\overline{\theta }^{1/4}}{%
		\underline{\theta }^{1/4}}+\rho _{n}+1\biggr).
	\end{equation*}
	
	In addition, for $n$ sufficiently large, Assumption \ref{ass:rate3}.2
	ensures that 
	\begin{equation*}
	3485C_{1}c_{1}^{-1/2}\log ^{1/2}(n)(\mu _{n}^{\tau })^{-1/2}\sigma
	_{Kn}^{-2}\rho _{n}\leq 0.001.
	\end{equation*}%
	This, in conjunction with (\ref{UB}), implies that 
	\begin{equation*}
	\Gamma_{n} \leq \delta _{n}^{(1)}\equiv 0.001\delta
	_{n}^{(0)}+3527C_{1}c_{1}^{-1/2}\eta _{n}.
	\end{equation*}%
	We iterate the above calculation $t$ times for some arbitrary integer $t$,
	and obtain that for $n\geq n_{1}$, 
	\begin{equation*}
	\Gamma _{n}\leq \delta _{n}^{(t)},\quad \delta _{n}^{(t)}=0.001\delta
	_{n}^{(t-1)}+3527C_{1}c_{1}^{-1/2}\eta _{n}.
	\end{equation*}%
	This implies 
	\begin{equation*}
	\delta _{n}^{(t)}=\left( 0.001\right) ^{t}\biggl[\delta
	_{n}^{(0)}-C_{1}c_{1}^{-1/2}\eta _{n}\biggr]+3527C_{1}c_{1}^{-1/2}\eta _{n}.
	\end{equation*}%
	In addition, because $\sup_{i}n_{g_{i}^{0}}^{\tau }(\theta _{i}^{\tau
	})^{-1}\Vert \hat{u}_{i}\Vert ^{2}\leq n\underline{\theta }^{-1}\Vert \hat{U}%
	_{1n}\Vert _{F}^{2}/K=n\underline{\theta }^{-1}$, we have 
	\begin{equation*}
	\sup_{i}(n_{g_{i}^{0}}^{\tau })^{1/2}(\theta _{i}^{\tau })^{-1/2}\Vert \hat{u%
	}_{i}\Vert \leq n^{1/2}\underline{\theta }^{-1/2}.
	\end{equation*}%
	Therefore, we can set $\delta _{n}^{(0)}=n^{1/2}\underline{\theta }^{-1/2}$
	and choose $n_{2}>n_{1}$ sufficiently large and $t=n$ such that for $n\geq
	n_{2}$, 
	\begin{equation*}
	\Gamma _{n}\leq \delta _{n}^{(n)}\leq 1000^{-n}n^{1/2}\underline{\theta }%
	^{-1/2}+3527C_{1}c_{1}^{-1/2}\eta _{n}\leq 3528C_{1}c_{1}^{-1/2}\eta _{n},
	\end{equation*}%
	where the last inequality holds because $\eta _{n}$ is either bounded away
	from zero or at most decays polynomially. This concludes the proof.\bigskip
\end{proof}

\begin{proof}[\textbf{Proof of Corollary \protect\ref{cor:DC}}]
	By the triangle inequality and Theorem \ref{thm:main_DC}, 
	\begin{align*}
	\sup_{i}\biggl\Vert\frac{\hat{u}_{i}^{T}}{\Vert \hat{u}_{i}^{T}\Vert }-\frac{%
		u_{i}^{T}\hat{O}_{n}^{T}}{\Vert u_{i}^{T}\Vert }\biggr\Vert& =\sup_{i}%
	\biggl\Vert\frac{\hat{u}_{i}^{T}\hat{O}_{n}}{\Vert \hat{u}_{i}^{T}\hat{O}%
		_{n}\Vert }-\frac{u_{i}^{T}}{\Vert u_{i}^{T}\Vert }\biggr\Vert \\
	& \leq \sup_{i}\biggl\Vert\frac{\hat{u}_{i}^{T}\hat{O}_{n}}{\Vert \hat{u}%
		_{i}^{T}\hat{O}_{n}\Vert }-\frac{\hat{u}_{i}^{T}\hat{O}_{n}}{\Vert
		u_{i}^{T}\Vert }\biggr\Vert+\sup_{i}n_{g_{i}^{0}}^{1/2}(\theta _{i}^{\tau
	})^{-1/2}\Vert \hat{u}_{i}^{T}\hat{O}_{n}-u_{i}^{T}\Vert \\
	& \leq 2\sup_{i}n_{g_{i}^{0}}^{1/2}(\theta _{i}^{\tau })^{-1/2}\Vert \hat{u}%
	_{i}^{T}\hat{O}_{n}-u_{i}^{T}\Vert \\
	& \leq 2C^*\biggl(\frac{\rho _{n}\log ^{1/2}(n)}{\left( \mu _{n}^{\tau
		}\right) ^{1/2}\sigma _{Kn}^{2}}\biggr)\biggl(\frac{\left(\frac{1}{K} + 
		\frac{\log(5)}{\log(n)}\right)^{1/2}\rho _{n}^{1/2}\overline{\theta }^{1/4}}{%
		\underline{\theta }^{1/4}}+\rho _{n}+1\biggr)\quad a.s.,
	\end{align*}%
	where $C^* = 3528C_1c_1^{-1/2}$.
	
	\medskip The second result follows Theorem \ref{thm:strong} with $\hat{\beta}%
	_{in}=\frac{\hat{u}_{i}^{T}}{\Vert \hat{u}_{i}^{T}\Vert }$, $\beta
	_{kn}=[S_{n}^{\tau }]_{k\cdot }$, $M=1$, $c_{1n}=\sqrt{2}$, and 
	\begin{equation*}
	c_{2n}=2C^{\ast }\biggl(\frac{\rho _{n}\log ^{1/2}(n)}{\left( \mu _{n}^{\tau
		}\right) ^{1/2}\sigma _{Kn}^{2}}\biggr)\biggl(\frac{\left( \frac{1}{K}+\frac{%
			\log (5)}{\log (n)}\right) ^{1/2}\rho _{n}^{1/2}\overline{\theta }^{1/4}}{%
		\underline{\theta }^{1/4}}+\rho _{n}+1\biggr),
	\end{equation*}%
	where $S_{n}^{\tau }$ is defined in Theorem \ref{thm:id3}. In addition,
	Assumption \ref{ass:theta}.3 holds because 
	\begin{align*}
	& (2c_{2n}c_{1}^{1/2}+16K^{3/4}M^{1/2}c_{2n}^{1/2})^{2} \\
	\leq & 16.02^{2}K^{3/2}c_{2n} \\
	\leq & 514K^{3/2}C^{\ast }\biggl(\frac{\rho _{n}\log ^{1/2}(n)}{\left( \mu
		_{n}^{\tau }\right) ^{1/2}\sigma _{Kn}^{2}}\biggr)\biggl(\frac{\left( \frac{1%
		}{K}+\frac{\log (5)}{\log (n)}\right) ^{1/2}\rho _{n}^{1/2}\overline{\theta }%
		^{1/4}}{\underline{\theta }^{1/4}}+\rho _{n}+1\biggr) \\
	\leq & 2c_{1}=c_{1}c_{1n}^{2},
	\end{align*}%
	where the first inequality holds because $c_{2n}c_{1}^{1/2}\leq
	0.01c_{2n}^{1/2}$ by Assumption \ref{ass:rate3} and the last inequality
	holds by Assumption \ref{ass:ratestrong4}.1.
	
	Similarly, the third result follows Theorem \ref{thm:strong2} with the same $%
	c_{1n}$ and $c_{2n}$ as above. The Assumption \ref{ass:ratestrong4}.2
	verifies Assumption \ref{ass:theta2}.2.
\end{proof}

\begin{proof}[\textbf{Proof of Theorem \protect\ref{thm:thetahat}}]
	Let $\varepsilon _{n}=C\log (n)/\underline{m}_{n}$, for some positive
	constant $C$ which is sufficiently large. 
	\begin{align*}
	& P\biggl(\sup_{1\leq i\leq n}|\hat{\theta}_{i}-\theta _{i}|\geq \varepsilon
	_{n}\quad i.o.\biggr) \\
	\leq & P\biggl(\sup_{1\leq i\leq n}|\hat{\theta}_{i}-\theta _{i}|\geq
	\varepsilon _{n}\quad i.o.,\text{ }\sup_{1\leq i\leq n}1\{\hat{g}_{i}\neq
	g_{i}^{0}\}=0\biggr)+P\biggl(\sup_{1\leq i\leq n}1\{\hat{g}_{i}\neq
	g_{i}^{0}\}>0\quad i.o.\biggr) \\
	\leq & P\biggl(\sup_{1\leq i\leq
		n}|n_{g_{i}^{0}}(\sum\nolimits_{j=1}^{n}A_{ij})/(\sum\nolimits_{i^{\prime
		}:g_{i^{\prime }}^{0}=g_{i}^{0}}\sum\nolimits_{j=1}^{n}A_{i^{\prime
		}j})-\theta _{i}|\geq \varepsilon _{n}\quad i.o.\biggr).
	\end{align*}%
	where the last inequality holds by Assumption \ref{ass:adaptive_rate}.2. In
	order to show the RHS of the above equation is zero, it suffices to show 
	\begin{equation}
	\sum_{n=1}^{\infty }\sum_{i=1}^{n}P\biggl(|n_{g_{i}^{0}}(\sum%
	\nolimits_{j=1}^{n}A_{ij})/(\sum\nolimits_{i^{\prime }:g_{i^{\prime
		}}^{0}=g_{i}^{0}}\sum\nolimits_{j=1}^{n}A_{i^{\prime }j})-\theta _{i}|\geq
	\varepsilon _{n}\biggr)<\infty .  \label{eq:p1}
	\end{equation}%
	For the simplicity of notation, from now on, we assume $g_{i}^{0}=k$. Then,
	we have 
	\begin{equation*}
	|n_{g_{i}^{0}}(\sum\nolimits_{j=1}^{n}A_{ij})/(\sum\nolimits_{i^{\prime
		}:g_{i^{\prime }}^{0}=g_{i}^{0}}\sum\nolimits_{j=1}^{n}A_{i^{\prime
		}j})-\theta _{i}|=\frac{\sum_{j=1}^{n}(A_{ij}n_{k}-\sum_{i^{\prime }\in
			C_{k}}A_{i^{\prime }j}\theta _{i})}{\sum_{j=1}^{n}\sum_{i^{\prime }\in
			C_{k}}A_{i^{\prime }j}}.
	\end{equation*}%
	For the denominator, note that $\mathbb{E}\sum_{j=1}^{n}\sum_{i^{\prime }\in
		C_{k}}A_{i^{\prime }j}=m_{k}n_{k}$, $\mathbb{E}\sum_{j=1}^{n}\sum_{i^{\prime
		}\in C_{k}}A_{i^{\prime }j}^{2}\leq m_{k}n_{k}$. Then, by Bernstein
	inequality, for any $\lambda >0$, 
	\begin{equation*}
	P\biggl(|\frac{\sum_{j=1}^{n}\sum_{i^{\prime }\in C_{k}}A_{i^{\prime }j}}{%
		m_{k}n_{k}}-1|\geq \lambda \biggr)\leq 2\exp \biggl(-\frac{\frac{1}{2}%
		\lambda ^{2}m_{k}^{2}n_{k}^{2}}{m_{k}n_{k}+\frac{1}{3}\lambda m_{k}n_{k}}%
	\biggr)=2\exp (-C_{\lambda }m_{k}n_{k}),
	\end{equation*}%
	where $C_{\lambda }=\frac{3\lambda ^{2}}{6+2\lambda }$. Similarly, for the
	numerator, we note that $|A_{ij}n_{k}-\sum_{i^{\prime }\in
		C_{k}}A_{i^{\prime }j}\theta _{i}|\leq n_{k}(\theta _{i}+1)$ and $%
	\sum_{j=1}^{n}\mathbb{E}(A_{ij}n_{k}-\sum_{i^{\prime }\in C_{k}}A_{i^{\prime
		}j}\theta _{i})^{2}\leq n_{k}^{2}-\theta _{i}^{2}m_{k}n_{k}$. Then, by
	Assumption \ref{ass:adaptive_rate}.1 and Bernstein inequality, 
	\begin{align*}
	P(|\frac{\sum_{j=1}^{n}(A_{ij}n_{k}-\sum_{i^{\prime }\in C_{k}}A_{i^{\prime
			}j}\theta _{i})}{m_{k}n_{k}}|\geq \varepsilon _{n})\leq & 2\exp \biggl(-%
	\frac{\frac{1}{2}\varepsilon _{n}^{2}m_{k}^{2}n_{k}^{2}}{n_{k}^{2}-\theta
		_{i}^{2}m_{k}n_{k}+\frac{1}{3}\varepsilon _{n}m_{k}n_{k}^{2}(\theta _{i}+1)}%
	\biggr) \\
	\leq & C\exp (-C^{\prime }\varepsilon _{n}m_{k}).
	\end{align*}%
	Therefore, 
	\begin{align*}
	& P\biggl(|\frac{\sum_{j=1}^{n}(A_{ij}n_{k}-\sum_{i^{\prime }\in
			C_{k}}A_{i^{\prime }j}\theta _{i})}{\sum_{j=1}^{n}\sum_{i^{\prime }\in
			C_{k}}A_{i^{\prime }j}}|\geq \varepsilon _{n}\biggr) \\
	\leq & P\biggl(|\frac{\sum_{j=1}^{n}(A_{ij}n_{k}-\sum_{i^{\prime }\in
			C_{k}}A_{i^{\prime }j}\theta _{i})}{m_{k}n_{k}}|\geq \varepsilon
	_{n}(1-\lambda )\biggr)+2\exp (-C_{\lambda }m_{k}n_{k}) \\
	\leq & C\exp (-C^{\prime }\varepsilon _{n}(1-\lambda )\underline{m}%
	_{n})+2\exp (-C_{\lambda }\underline{m}_{n}n_{k}).
	\end{align*}%
	By construction, $\varepsilon _{n}\underline{m}_{n}=C\log (n)$ for $C$
	sufficiently large. Therefore, \eqref{eq:p1} holds, which concludes the
	proof.
\end{proof}

\section{Some technical lemmas}

\label{sec:lem} In this appendix we collect some technical lemmas that are
used in the proofs of the main results in the paper.

We first state a version of Davis-Kahan $\sin \Theta $ theorem that is
closely related to the results in \cite{DK70}, \cite{YWS15} and \cite%
{abbe2017}.

\begin{lem}
	\label{lem:davis} Let $A$ and $A^{\ast }$ be two $n\times n$ matrices with
	spectral decompositions given by 
	\begin{equation*}
	A=V\Sigma V^{T}\quad \text{and}\quad A^{\ast }=V^{\ast }\Sigma ^{\ast
	}(V^{\ast })^{T},
	\end{equation*}%
	where $\Sigma =\text{diag}(\sigma _{1},\sigma _{2},\cdots ,\sigma _{n})$, $%
	\Sigma ^{\ast }=\text{diag}(\sigma _{1}^{\ast },\sigma _{2}^{\ast },\cdots
	,\sigma _{n}^{\ast })$, $|\sigma _{1}|\geq \cdots \geq |\sigma _{n}|\geq 0$, 
	$|\sigma _{1}^{\ast }|\geq \cdots \geq |\sigma _{n}^{\ast }|\geq 0,$ and $V$
	and $V^{\ast }$ are the associated eigenvectors. Suppose that $A^{\ast }$
	has rank $K.$ Let $V_{1}$ and $V_{1}^{\ast }$ be the first $K$ columns of $V$
	and $V^{\ast }$, respectively. Suppose there exists some rate $\gamma
	_{n}\downarrow 0$ such that $|\sigma _{K}^{\ast }|-\gamma _{n}>0$ and $\Vert
	A-A^{\ast }\Vert \leq \gamma _{n}\mbox{ }a.s.$ Let $\Omega =\text{diag}(\cos
	(\theta _{1}),\cdots ,\cos (\theta _{K})),$ where $\theta _{k}\in (0,\pi /2)$%
	, $k=1,\cdots ,K$, denote the principal angles between the column spaces of $%
	V$ and $V^{\ast }$ such that $\theta _{1}\leq \cdots \leq \theta _{K}$.
	
	Then 
	\begin{equation*}
	\Vert V_{1}O-V_{1}^{\ast }\Vert \leq \frac{\sqrt{2}\Vert (A-A^{\ast
		})V_{1}^{\ast }\Vert }{|\sigma _{K}^{\ast }|-\gamma _{n}}\mbox{ }a.s.,
	\end{equation*}%
	where $O=O_{1}O_{2}^{T}$ and $V_{1}^{T}V_{1}^{\ast }$ has the singular value
	decomposition $O_{1}\Omega O_{2}^{T}$ so that $O_{1}$ and $O_{2}$ are $%
	K\times K$ orthogonal matrices such that $O_{1}^{T}V_{1}^{T}V_{1}^{\ast
	}O_{2}=\Omega .$
\end{lem}

\begin{proof}[\textbf{Proof of Lemma \protect\ref{lem:davis}}]
	By the proof of \citet[Theorem 2]{YWS15}, 
	\begin{eqnarray*}
		\Vert V_{1}O-V_{1}^{\ast }\Vert ^{2} &=&\left\Vert (V_{1}O-V_{1}^{\ast
		})^{T}(V_{1}O-V_{1}^{\ast })\right\Vert =2\left\Vert I_{K}-O_{2}\Omega
		O_{2}^{T}\right\Vert \\
		&\leq &2[1-\cos (\theta _{K})]\leq 2(1-\left[ \cos (\theta _{K})\right]
		^{2})\leq 2\left[ \sin (\theta _{K})\right] ^{2}.
	\end{eqnarray*}%
	In addition, by the $\sin \Theta $ theorem in \cite{DK70} (see also Appendix
	A.1 in Abbe et al. (2017)), 
	\begin{equation*}
	\sin (\theta _{K})\leq \frac{\Vert (A-A^{\ast })V_{1}^{\ast }\Vert }{\Delta
		_{K}},
	\end{equation*}%
	where $\Delta _{K}=(\left\vert \sigma _{K}^{\ast }\right\vert -\left\vert
	\sigma _{K+1}\right\vert )\vee 0$. In addition, by Weyl's inequality, 
	\begin{equation*}
	|\sigma _{K+1}|=\left\vert \sigma _{K+1}-\sigma _{K+1}^{\ast }\right\vert
	\leq \gamma _{n}\mbox{ }a.s.
	\end{equation*}%
	Therefore, 
	\begin{equation*}
	\Vert V_{1}O-V_{1}^{\ast }\Vert \leq \frac{\sqrt{2}\Vert (A-A^{\ast
		})V_{1}^{\ast }\Vert }{|\sigma _{K}^{\ast }|-\gamma _{n}}\mbox{ }a.s.
	\end{equation*}
\end{proof}

The following lemma states a version of Bernstein inequality for random
matrix that is used in the proof of Lemma \ref{lem:dk3}.

\begin{lem}
	\label{lem:matrixBernstein} Consider an independent sequence $(Y_{k})_{k\geq
		1}$ of real symmetric $d\times d$ random matrices that satisfy $\mathbb{E}%
	Y_{k}=0$ and $\Vert Y_{k}\Vert \leq R$ for each index $k$. Then for all $%
	t\geq 0$ and $\sigma ^{2}=\Vert \sum_{k\geq 1}\mathbb{E}Y_{k}^{2}\Vert $, 
	\begin{equation*}
	P(\Vert \sum_{k\geq 1}Y_{k}\Vert \geq t)\leq d\exp \biggl(\frac{-t^{2}}{%
		3\sigma ^{2}+2Rt}\biggr).
	\end{equation*}
\end{lem}

\begin{proof}[\textbf{Proof of Lemma \protect\ref{lem:matrixBernstein}}]
	See Corollary 5.2 in \cite{mackey14}.
\end{proof}

To prove Theorem \ref{thm:main_DC} in Section \ref{sec:ext2}, we need the
following three lemmas.

\begin{lem}
	\label{lem:Pij3}If Assumption \ref{ass:id3} holds, then $P_{ij}\leq \rho
	_{n}n^{-1}(\theta _{i}\theta _{j})^{1/2}(d_{i}d_{j})^{1/2}.$
\end{lem}

\begin{proof}
	Consider the case in which nodes $i$ and $j$ are in $C_{k_{1}}$ and $%
	C_{k_{2}}$, respectively. Then by the definition of $B_{0}$ and %
	\eqref{eq:d_i} 
	\begin{align*}
	P_{ij}=\theta _{i}\theta _{j}B_{k_{1}k_{2}}&= n^{-1}\theta _{i}\theta
	_{j}(nW_{k_{1}})^{1/2}[B_{0}]_{k_{1}k_{2}}(nW_{k_{2}})^{1/2} \\
	&= n^{-1}(\theta _{i}\theta
	_{j})^{1/2}[B_{0}]_{k_{1}k_{2}}(d_{i}d_{j})^{1/2}\leq \rho_n n^{-1}(\theta
	_{i}\theta _{j})^{1/2}(d_{i}d_{j})^{1/2}.
	\end{align*}
\end{proof}

\begin{lem}
	\label{lem:V1n3}Let $V_{n}^{(i)}$, $i=1,\cdots ,n,$ be $n\times K$ random
	matrices. Suppose $V_{n}^{(i)}$ and $[A]_{i\cdot }-[P]_{i\cdot }$ are
	independent for $i=1,\cdots ,n$ and there exist two deterministic sequences $%
	\{\phi _{1n}\}_{n\geq 1}$ and $\{\phi _{2n}\}_{n\geq 1}$ such that $%
	\sup_{i}\Vert V_{n}^{(i)}\Vert \leq \phi _{1n}\quad $and$\quad \sup_{i}\Vert
	V_{n}^{(i)}\Vert _{2\rightarrow \infty }\leq \phi _{2n}\mbox{ }a.s.$ Suppose
	that Assumptions \ref{ass:id3}--\ref{ass:rate3} hold. Then 
	\begin{equation*}
	\hspace{2em}\hspace{-2em}\sup_{i}\biggl((n_{g_{i}^{0}}^{\tau })^{1/2}(\theta
	_{i}^{\tau })^{-1/2}(d_{i}^{\tau })^{-1/2}\Vert ([A]_{i\cdot }-[P]_{i\cdot })%
	\mathcal{D}_{\tau }^{-1/2}V_{n}^{(i)}\Vert \biggr)\leq
	6C_{1}^{1/2}r_{n}\quad a.s.,
	\end{equation*}%
	where $r_{n}=\biggl[\frac{\phi _{2n}(\log (n)+\log (5)K)n^{1/2}}{\mu
		_{n}^{\tau }(K\underline{\theta })^{1/2}}\vee \biggl(\frac{(\log (n)+\log
		(5)K)\rho _{n}\phi _{1n}^{2}\overline{\theta }^{1/2}}{\mu _{n}^{\tau }K%
		\underline{\theta }^{1/2}}\biggr)^{1/2}\biggr]$.
\end{lem}

\begin{proof}
	Let $C=3C_{1}^{1/2}$. Define 
	\begin{eqnarray*}
		\mathcal{E}_{1n} &=&\biggl\{\sup_{i}\biggl((n_{g_{i}^{0}}^{\tau
		})^{1/2}(\theta _{i}^{\tau })^{-1/2}(d_{i}^{\tau })^{-1/2}\Vert ([A]_{i\cdot
		}-[P]_{i\cdot })\mathcal{D}_{\tau }^{-1/2}V_{n}^{(i)}\Vert \biggr)\geq
		2Cr_{n}\biggr\},\text{ and} \\
		\mathcal{E}_{2n} &=&\{\sup_{i}\Vert V_{n}^{(i)}\Vert \leq \phi _{1n}\quad 
		\text{and}\quad \sup_{i}\Vert V_{n}^{(i)}\Vert _{2\rightarrow \infty }\leq
		\phi _{2n}\}.
	\end{eqnarray*}%
	It suffices to show that $P(\mathcal{E}_{1n}\mbox{ }i.o.)=0.$ By the
	assumptions in Lemma \ref{lem:V1n3}, we have 
	\begin{equation}
	P(\mathcal{E}_{2n}^{c}\mbox{ }i.o.)=0.  \label{eq:E2n}
	\end{equation}%
	It follows that 
	\begin{align*}
	P(\mathcal{E}_{1n},i.o.)=& P(\cap _{k=1}^{\infty }\cup _{n\geq k}\mathcal{E}%
	_{1n}) \\
	=& P(\cap _{k=1}^{\infty }\cup _{n\geq k}(\mathcal{E}_{1n}\cap \mathcal{E}%
	_{2n}))+P(\cap _{k=1}^{\infty }\cup _{n\geq k}(\mathcal{E}_{1n}\cap \mathcal{%
		E}_{2n}^{c})) \\
	\leq & P(\cap _{k=1}^{\infty }\cup _{n\geq k}(\mathcal{E}_{1n}\cap \mathcal{E%
	}_{2n}))+P(\cap _{k=1}^{\infty }\cup _{n\geq k}\mathcal{E}_{2n}^{c}) \\
	=& P(\cap _{k=1}^{\infty }\cup _{n\geq k}(\mathcal{E}_{1n}\cap \mathcal{E}%
	_{2n})),
	\end{align*}%
	where the last step is due to \eqref{eq:E2n}. Therefore, we only need to
	show that 
	\begin{align*}
	& P(\mathcal{E}_{1n}\cap \mathcal{E}_{2n}\mbox{ }i.o.) \\
	=& P\biggl(\sup_{i}\biggl((n_{g_{i}^{0}}^{\tau })^{1/2}(\theta _{i}^{\tau
	})^{-1/2}(d_{i}^{\tau })^{-1/2}\Vert ([A]_{i\cdot }-[P]_{i\cdot })\mathcal{D}%
	_{\tau }^{-1/2}V_{n}^{(i)}\Vert \biggr)\geq 2Cr_{n}\cap \mathcal{E}_{2n}%
	\mbox{ }i.o.\biggr)=0.
	\end{align*}%
	By the Borel-Cantelli lemma and union bound, it suffices to show that 
	\begin{equation*}
	\sum_{n\geq 1}\sum_{i=1}^{n}P\biggl((n_{g_{i}^{0}}^{\tau })^{1/2}(\theta
	_{i}^{\tau })^{-1/2}(d_{i}^{\tau })^{-1/2}\Vert ([A]_{i\cdot }-[P]_{i\cdot })%
	\mathcal{D}_{\tau }^{-1/2}V_{n}^{(i)}\Vert \geq 2Cr_{n}\cap \mathcal{E}_{2n}%
	\biggr)<\infty .
	\end{equation*}
	
	Now, let $S^{K-1}=\{g\in \Re ^{K}:\Vert g\Vert =1\}$ and $\mathcal{F}$ be a $%
	1/2$-net of $S^{K-1}$. By \citet[Lemma 4.4.1]{Vy18}, $|\mathcal{F}|\leq
	5^{K} $. Then, 
	\begin{align}
	& \sum_{i=1}^{n}P\biggl((n_{g_{i}^{0}}^{\tau })^{1/2}(\theta _{i}^{\tau
	})^{-1/2}(d_{i}^{\tau })^{-1/2}\Vert ([A]_{i\cdot }-[P]_{i\cdot })\mathcal{D}%
	_{\tau }^{-1/2}V_{n}^{(i)}\Vert \geq 2Cr_{n},\mathcal{E}_{2n}\biggr)  \notag
	\\
	\leq & \sum_{i=1}^{n}5^{K}\sup_{f\in \mathcal{F}}P\biggl((n_{g_{i}^{0}}^{%
		\tau })^{1/2}(\theta _{i}^{\tau })^{-1/2}(d_{i}^{\tau })^{-1/2}|([A]_{i\cdot
	}-[P]_{i\cdot })\mathcal{D}_{\tau }^{-1/2}V_{n}^{(i)}f)|\geq Cr_{n},\mathcal{%
		E}_{2n}\biggr)  \notag \\
	\equiv & I_{n},  \label{eq:I+IIDC}
	\end{align}%
	where the first inequality holds due the union bound and 
	\citet[][Corollary
	4.2.13, Lemma 4.4.1]{Vy18}. Let 
	\begin{equation*}
	\mathcal{H}=\{h\in \Re ^{n}:\Vert h\Vert \leq \phi _{1n}\text{ and }%
	\sup_{j}|h_{j}|\leq \phi _{2n}\},
	\end{equation*}%
	where $h_{j}$ is the $j$-th element of $h$. Note that for any $f\in S^{K-1}$
	, $\Vert V_{n}^{(i)}f\Vert =\Vert V_{n}^{(i)}\Vert \leq \phi _{1n}$ and $%
	|[V_{n}^{(i)}f]_{j\cdot }|\leq \Vert \lbrack V_{n}^{(i)}]_{j\cdot }\Vert
	\leq \phi _{2n}$ a.s. Thus, under $\mathcal{E}_{2}$, $\{V_{n}^{(i)}f:f\in
	S^{K-1}\}\subset \mathcal{H}$. For any $h\in \mathcal{H}$, 
	\begin{equation*}
	(n_{g_{i}^{0}}^{\tau })^{1/2}(\theta _{i}^{\tau
	})^{-1/2}|(A_{ij}-P_{ij})(d_{i}^{\tau }d_{j}^{\tau })^{-1/2}h_{j}|\leq
	C_{1}^{1/2}\phi _{2n}n^{1/2}(\underline{\theta }K)^{-1/2}(\mu _{n}^{\tau
	})^{-1}.
	\end{equation*}%
	In addition, by Lemma \ref{lem:Pij3}, 
	\begin{equation*}
	\sum_{j\neq i}n_{g_{i}^{0}}^{\tau }(\theta _{i}^{\tau })^{-1}\mathbb{E}%
	(A_{ij}-P_{ij})^{2}(d_{i}^{\tau }d_{j}^{\tau })^{-1}h_{j}^{2}\leq
	\sum_{j=1}^{n}n_{g_{i}^{0}}^{\tau }(\theta _{i}^{\tau
	})^{-1}P_{ij}(d_{i}^{\tau }d_{j}^{\tau })^{-1}h_{j}^{2}\leq C_{1}\rho _{n}%
	\overline{\theta }^{1/2}\underline{\theta }^{-1/2}K^{-1}(\mu _{n}^{\tau
	})^{-1}\phi _{1n}^{2}
	\end{equation*}%
	and for $n$ sufficiently large, 
	\begin{equation}
	(n_{g_{i}^{0}}^{\tau })^{1/2}(\theta _{i}^{\tau
	})^{-1/2}|A_{ii}-P_{ii}|(d_{i}^{\tau })^{-1}|h_{i}|\leq C_{1}^{1/2}\phi
	_{2n}n^{1/2}(K\underline{\theta })^{-1/2}(\mu _{n}^{\tau })^{-1}\leq
	Cr_{n}/100.  \label{eq:ii}
	\end{equation}%
	Then, by the Bernstein inequality in Lemma \ref{lem:matrixBernstein}, 
	\begin{align}
	I_{n}& \leq n5^{K}\sup_{i=1,\cdots ,n,f\in \mathcal{S}^{K-1}}P\biggl(%
	(n_{g_{i}^{0}}^{\tau })^{1/2}(\theta _{i}^{\tau })^{-1/2}(d_{i}^{\tau
	})^{-1/2}|([A]_{i\cdot }-[P]_{i\cdot })\mathcal{D}_{\tau
	}^{-1/2}V_{n}^{(i)}f|\geq Cr_{n},\mathcal{E}_{2n}\biggr)  \notag \\
	& \leq n5^{K}\sup_{i=1,\cdots ,n,h\in \mathcal{H}}P\biggl(%
	(n_{g_{i}^{0}}^{\tau })^{1/2}(\theta _{i}^{\tau })^{-1/2}|\sum_{j\neq
		i}(A_{ij}-P_{ij})(d_{i}d_{j})^{-1/2}h_{j}|\geq 0.99Cr_{n}|V_{n}^{(i)}f=h%
	\biggr)  \notag \\
	& \leq n5^{K}\sup_{i=1,\cdots ,n,h\in \mathcal{H}}P\biggl(%
	(n_{g_{i}^{0}}^{\tau })^{1/2}(\theta _{i}^{\tau })^{-1/2}|\sum_{j\neq
		i}(A_{ij}-P_{ij})(d_{i}d_{j})^{-1/2}h_{j}|\geq 0.99Cr_{n}\biggr)  \notag \\
	& \leq 2n5^{K}\exp \biggl(\frac{-(0.99C)^{2}r_{n}^{2}}{\frac{%
			1.98CC_{1}^{1/2}r_{n}\phi _{2n}n^{1/2}}{3\mu _{n}^{\tau }(K\underline{\theta 
			})^{1/2}}+\frac{2C_{1}\rho _{n}\phi _{1n}^{2}\overline{\theta }^{1/2}}{K\mu
			_{n}^{\tau }\underline{\theta }^{1/2}}}\biggr)  \notag \\
	& \leq 2n5^{K}\exp \biggl(\frac{-(0.99C)^{2}r_{n}^{2}(\log (n)+\log (5)K)}{%
		(1.98CC_{1}^{1/2}/3+2C_{1})r_{n}^{2}}\biggr)\leq 2(5n)^{-1.1},
	\label{eq:IIDC}
	\end{align}%
	where the second inequality holds by the Bayes rule and \eqref{eq:ii}, the
	third inequality holds because we assume that $([A]_{i\cdot }-[P]_{i\cdot })$
	and $V_{n}^{(i)}$ are independent, the fourth inequality holds by the
	Bernstein inequality, and the fifth inequality holds because of the
	definition of $r_{n}$, and the sixth inequality holds because $(0.99\times
	3)^{2}>2.1\times (1.98+2)$ and we have set $C=3C_{1}^{1/2}.$
	
	Combining \eqref{eq:I+IIDC} and \eqref{eq:IIDC}, we have 
	\begin{equation*}
	\sum_{n=1}^{\infty }\sum_{i=1}^{n}P\biggl((n_{g_{i}^{0}}^{\tau
	})^{1/2}(\theta _{i}^{\tau })^{-1/2}\biggl((d_{i}^{\tau })^{-1/2}\Vert
	([A]_{i\cdot }-[P]_{i\cdot })\mathcal{D}_{\tau }^{-1/2}V_{n}^{(i)}\Vert %
	\biggr)\geq 2Cr_{n}\biggr)<\infty .
	\end{equation*}%
	This leads to the desired result.\medskip
\end{proof}

Recall $\hat{\Lambda}=L_{\tau }^{\prime }\hat{U}_{1n}\hat{O}_{n}=\hat{U}_{1n}%
\hat{\Sigma}_{n}\hat{O}_{n}$, $\Lambda =\mathcal{L}_{\tau }^{\prime
}U_{1n}=U_{1n}\Sigma _{n}$, $\hat{\Lambda}_{i}=\hat{u}_{i}^{T}\hat{\Sigma}%
_{n}\hat{O}_{n}$, and $\Lambda _{i}=u_{i}^{T}\Sigma _{n}$, where $\hat{u}%
_{i}^{T}$ and $u_{i}^{T}$ are the $i$-th rows of $\hat{U}_{1n}$ and $U_{1n}$%
, respectively. In order to state and prove the next lemma, we need to
introduce some extra notation. Let $A^{(i)}$ be the matrix obtained by
replacing all the elements in the $i$-th row and column of $A$ by their
expectations, except $A_{ii}^{(i)}$ which is set as zero. Following the
notation in \cite{abbe2017}, we denote 
\begin{equation*}
\hat{H}_{n}=\hat{U}_{1n}^{T}U_{1n}.
\end{equation*}%
Then 
\begin{equation*}
\hat{O}_{n}=\bar{U}\bar{V}^{T},
\end{equation*}%
where $\bar{U}\bar{\Sigma}\bar{V}^{T}$ is the singular value decomposition
of $\hat{H}_{n}$. Similarly, let 
\begin{equation*}
L_{\tau }^{(i)}=\mathcal{D}_{\tau }^{-1/2}A^{(i)}\mathcal{D}_{\tau }^{-1/2}=%
\hat{U}_{n}^{(i)}\hat{\Sigma}_{n}^{(i)}\hat{U}_{n}^{(i)},
\end{equation*}%
where $\hat{\Sigma}_{n}^{(i)}=\text{diag}(\sigma _{1n}^{(i)},\cdots ,\sigma
_{nn}^{(i)})$ and $|\sigma _{1n}^{(i)}|\geq \cdots \geq |\sigma _{nn}^{(i)}|$%
. Further denote $\hat{U}_{1n}^{(i)}$ and $\hat{\Sigma}_{1n}^{(i)}$ as the
first $K$ eigenvectors of $L^{(i)}$ and the corresponding eigenvalues $\text{%
	diag}(\sigma _{1n}^{(i)},\cdots ,\sigma _{Kn}^{(i)})$, respectively. We
denote 
\begin{equation*}
\hat{H}_{n}^{(i)}=(\hat{U}_{1n}^{(i)})^{T}U_{1n}
\end{equation*}%
and 
\begin{equation*}
\hat{O}_{n}^{(i)}=\bar{U}^{\left( i\right) }(\bar{V}^{\left( i\right) })^{T},
\end{equation*}%
where $\bar{U}^{\left( i\right) }\bar{\Sigma}^{\left( i\right) }(\bar{V}%
^{\left( i\right) })^{T}$ is the singular value decomposition of $\hat{H}%
_{n}^{(i)}$.

\begin{lem}
	\label{lem:B3} Suppose that Assumptions \ref{ass:id3}--\ref{ass:rate3} hold.
	Then there exists a sufficiently large positive constant $C$ such that 
	\begin{equation*}
	\Vert \hat{\Lambda}-\Lambda \Vert \leq 17(\log (n)/\mu _{n}^{\tau
	})^{1/2}|\sigma _{Kn}|^{-1}\quad a.s.
	\end{equation*}%
	If, in addition, there exists a deterministic sequence $\{\psi _{n}\}_{n\geq
		1}$ such that $\sup_{j}(n_{g_{j}^{0}}^{\tau })^{1/2}(\theta _{j}^{\tau
	})^{-1/2}\Vert \hat{u}_{j}\Vert \leq \psi _{n}$ almost surely, then 
	\begin{align*}
	& \sup_{i}(n_{g_{i}^{0}}^{\tau })^{1/2}(\theta _{i}^{\tau })^{-1/2}\Vert 
	\hat{\Lambda}_{i}-\Lambda _{i}\Vert \\
	\leq & 3450C_1c_1^{-1/2}\rho _{n}(\log(n)/\mu _{n}^{\tau })^{1/2}|\sigma
	_{Kn}^{-1}|\biggl[\psi _{n}+\rho _{n}+\frac{(\frac{1}{K}+\frac{\log(5)}{%
			\log(n)})^{1/2}\rho _{n}^{1/2}\overline{\theta }^{1/4}}{\underline{\theta }%
		^{1/4}}\biggr]\quad a.s.
	\end{align*}
\end{lem}

\begin{proof}
	By \citet[][Lemma 1.7]{C97}, $\Vert L_{\tau }^{\prime }\Vert \leq \Vert
	L\Vert \leq 1$. Then, by Lemma \ref{lem:dk3} 
	\begin{align*}
	\Vert \hat{\Lambda}-\Lambda \Vert & =\Vert L_{\tau }^{\prime }\hat{U}_{1n}%
	\hat{O}_{n}-\mathcal{L}_{\tau }^{\prime }U_{1n}\Vert \\
	& \leq \Vert L_{\tau }^{\prime }(\hat{U}_{1n}\hat{O}_{n}-U_{1n})\Vert +\Vert
	(L_{\tau }^{\prime }-\mathcal{L}_{\tau }^{\prime })U_{1n}\Vert \\
	& \leq \Vert \hat{U}_{1n}\hat{O}_{n}-U_{1n}\Vert +\Vert L_{\tau }^{\prime }-%
	\mathcal{L}_{\tau }^{\prime }\Vert \\
	& \leq 17(\log (n)/\mu _{n}^{\tau })^{1/2}|\sigma _{Kn}|^{-1}\quad a.s.
	\end{align*}%
	This proves the first result.
	
	For the second result, denote $\tilde{\Lambda}=D_{\tau }^{-1/2}PD_{\tau
	}^{-1/2}U_{1n}$ and $\tilde{\Lambda}_{i}=(\hat{d}_{i}^{\tau
	})^{-1/2}[P]_{i\cdot }D_{\tau }^{-1/2}U_{1n}$ as the $i$-th row of $\tilde{%
		\Lambda}$. Then we have 
	\begin{align}
	\sup_{i}(n_{g_{i}^{0}}^{\tau })^{1/2}(\theta _{i}^{\tau })^{-1/2}\Vert \hat{%
		\Lambda}_{i}-\Lambda _{i}\Vert & \leq \sup_{i}(n_{g_{i}^{0}}^{\tau
	})^{1/2}(\theta _{i}^{\tau })^{-1/2}\Vert \Lambda _{i}-\tilde{\Lambda}%
	_{i}\Vert +\sup_{i}(n_{g_{i}^{0}}^{\tau })^{1/2}(\theta _{i}^{\tau
	})^{-1/2}\Vert \hat{\Lambda}_{i}-\tilde{\Lambda}_{i}\Vert  \notag \\
	& \equiv T_1+T_2.  \label{eq:12DC3}
	\end{align}%
	We can further decompose $T_2$ as follows: 
	\begin{align}
	\hspace{2em}& \hspace{-2em}\sup_{i}(n_{g_{i}^{0}}^{\tau })^{1/2}(\theta
	_{i}^{\tau })^{-1/2}\Vert \hat{\Lambda}_{i}-\tilde{\Lambda}_{i}\Vert  \notag
	\\
	& =\sup_{i}(n_{g_{i}^{0}}^{\tau })^{1/2}(\theta _{i}^{\tau })^{-1/2}\Vert (%
	\hat{d}_{i}^{\tau })^{-1/2}[A]_{i\cdot }D_{\tau }^{-1/2}\hat{U}_{1n}\hat{O}%
	_{n}-(\hat{d}_{i}^{\tau })^{-1/2}[P]_{i\cdot }D_{\tau }^{-1/2}U_{1n}\Vert 
	\notag \\
	& \leq \sup_{i}(n_{g_{i}^{0}}^{\tau })^{1/2}(\theta _{i}^{\tau })^{-1/2}(%
	\hat{d}_{i}^{\tau })^{-1/2}\Vert \lbrack P]_{i\cdot }D_{\tau }^{-1/2}(\hat{U}%
	_{1n}\hat{O}_{n}-U_{1n})\Vert  \notag \\
	& + \sup_{i}(n_{g_{i}^{0}}^{\tau })^{1/2}(\theta _{i}^{\tau })^{-1/2}(\hat{d}%
	_{i}^{\tau })^{-1/2}\Vert ([A]_{i\cdot }-[P]_{i\cdot })(D_{\tau }^{-1/2}-%
	\mathcal{D}_{\tau }^{-1/2})\hat{U}_{1n}\hat{O}_{n}\Vert  \notag \\
	& +\sup_{i}(n_{g_{i}^{0}}^{\tau })^{1/2}(\theta _{i}^{\tau })^{-1/2}(\hat{d}%
	_{i}^{\tau })^{-1/2}\Vert ([A]_{i\cdot }-[P]_{i\cdot })\mathcal{D}_{\tau
	}^{-1/2}\hat{U}_{1n}\hat{O}_{n}\Vert  \notag \\
	& \equiv T_{2,1}+T_{2,2}+T_{2,3}.  \label{eq:2DC}
	\end{align}%
	In the following, we bound $T_1$, $T_{2,1}$, $T_{2,2}$, and $T_{2,3}$ in
	four steps.
	
	\textbf{Step 1: Bound for $T_1$} \newline
	For $T_1$, we have 
	\begin{align*}
	T_1& \leq \sup_{i}\Vert (n_{g_{i}^{0}}^{\tau })^{1/2}(\theta _{i}^{\tau
	})^{-1/2}((\hat{d}_{i}^{\tau })^{-1/2}[P]_{i\cdot }D_{\tau
	}^{-1/2}-(d_{i}^{\tau })^{-1/2}[P]_{i\cdot }\mathcal{D}_{\tau
	}^{-1/2})U_{1n}\Vert \\
	& \leq \sup_{i}(n_{g_{i}^{0}}^{\tau })^{1/2}(\theta _{i}^{\tau
	})^{-1/2}\Vert (\hat{d}_{i}^{\tau })^{-1/2}[P]_{i\cdot }D_{\tau
	}^{-1/2}-(d_{i}^{\tau })^{-1/2}[P]_{i\cdot }\mathcal{D}_{\tau }^{-1/2}\Vert
	\\
	& \leq \sup_{i}(n_{g_{i}^{0}}^{\tau })^{1/2}(\theta _{i}^{\tau })^{-1/2}(%
	\hat{d}_{i}^{\tau })^{-1/2}\Vert \lbrack P]_{i\cdot }\mathcal{D}_{\tau
	}^{-1/2}\Vert \Vert \mathcal{D}^{1/2}_\tau D_{\tau }^{-1/2}-I\Vert \\
	& \quad \ +\sup_{i}(n_{g_{i}^{0}}^{\tau })^{1/2}(\theta _{i}^{\tau
	})^{-1/2}|(\hat{d}_{i}^{\tau })^{-1/2}-(d_{i}^{\tau })^{-1/2}|\Vert \lbrack
	P]_{i\cdot }\mathcal{D}_{\tau }^{-1/2}\Vert \\
	& \equiv T_{1,1}+T_{1,2}.
	\end{align*}%
	By Assumption \ref{ass:nk3}, Lemma \ref{lem:Pij3}, and the fact that $%
	\sum_{j=1}^n \theta_j = n$, 
	\begin{equation}
	(n_{g_{i}^{0}}^{\tau })^{1/2}(\theta _{i}^{\tau })^{-1/2}\Vert \lbrack
	P]_{i\cdot }\mathcal{D}_{\tau }^{-1/2}\Vert =(n_{g_{i}^{0}}^{\tau
	})^{1/2}(\theta _{i}^{\tau })^{-1/2}(\sum_{j=1}^{n}P_{ij}^{2}(d_{j}^{\tau
	})^{-1})^{1/2}\leq C_1^{1/2}\rho _{n}(d_{i}/K)^{1/2},  \label{eq:pidotdc3}
	\end{equation}%
	where the constant $C_1$ is defined in Assumption \ref{ass:nk3}. In
	addition, by \eqref{eq:d} in the proof of Lemma \ref{lem:dk3}, for all $i =
	1,\cdots,n$ 
	\begin{align*}
	1-0.0209 \leq (\hat{d}_{i}^{\tau })^{1/2}(d_{i}^{\tau })^{-1/2} \leq
	1+0.0209,
	\end{align*}
	which implies that
	
	\begin{equation}
	\sup_{i}|(\hat{d}_{i}^{\tau })^{-1/2}(d_{i}^{\tau })^{1/2}-1|\leq
	0.0214\quad a.s.  \label{eq:di3}
	\end{equation}%
	and 
	\begin{equation}
	\Vert \mathcal{D}_{\tau }^{1/2}D_{\tau }^{-1/2}-I\Vert \leq 0.0214\quad a.s.
	\label{eq:DD3}
	\end{equation}%
	Therefore, 
	\begin{equation*}
	T_{1,1}\leq 1.022 C_1^{1/2}\log ^{1/2}(n)(K\mu _{n}^{\tau })^{-1/2}\rho
	_{n}\quad a.s.,
	\end{equation*}%
	\begin{equation*}
	T_{1,2}\leq 1.022 C_1^{1/2}\log ^{1/2}(n)(K\mu _{n}^{\tau })^{-1/2}\rho
	_{n}\quad a.s.,
	\end{equation*}%
	and 
	\begin{equation}
	T_1\leq 2.044C_1^{1/2}\log ^{1/2}(n)(K\mu _{n}^{\tau })^{-1/2}\rho _{n}\quad
	a.s.  \label{eq:IDCfinal}
	\end{equation}
	
	\textbf{Step 2: Bound for $T_{2,1}$} \newline
	By Lemma \ref{lem:dk3} and \eqref{eq:pidotdc3}--\eqref{eq:DD3}, 
	\begin{align}  \label{eq:II1DCfinal}
	T_{2,1}\leq & C_1^{1/2}(1.022)^2\times 10 \times \sup_{i}(d_{i}^{\tau
	})^{-1/2}(d_{i}/K)^{1/2}\rho _{n}(\log (n)/\mu _{n}^{\tau })^{1/2}|\sigma
	_{Kn}|^{-1}  \notag \\
	\leq & 10.45C_1^{1/2}\rho _{n}\log ^{1/2}(n)(K\mu _{n}^{\tau
	})^{-1/2}|\sigma _{Kn}|^{-1}\quad a.s.
	\end{align}
	
	\textbf{Step 3: Bound for $T_{2,2}$}\newline
	For $T_{2,2}$, we have 
	\begin{align}  \label{eq:II20}
	& \sup_{i}(n_{g_{i}^{0}}^{\tau })^{1/2}(\theta _{i}^{\tau })^{-1/2}(\hat{d}%
	_{i}^{\tau })^{-1/2}\Vert ([A]_{i\cdot }-[P]_{i\cdot })(D_{\tau }^{-1/2}-%
	\mathcal{D}_{\tau }^{-1/2})\hat{U}_{1n}\hat{O}_{n}\Vert  \notag \\
	\leq & 1.022C_1^{1/2}\sup_{i}(nK/\theta_i)^{1/2}(d_{i}^{\tau
	})^{-1/2}\sup_{g\in S^{K-1}}\left\vert \sum_{j=1}^{n}\frac{%
		(A_{ij}-P_{ij})(d_{j}^{\tau }-\hat{d}_{j}^{\tau })(\hat{u}_{j}^{T}g)}{\sqrt{%
			\hat{d}_{j}^{\tau }d_{j}^{\tau }}(\sqrt{d_{j}^{\tau }}+\sqrt{\hat{d}%
			_{j}^{\tau }})}\right\vert  \notag \\
	\leq & 1.022\times 2.09C_1^{1/2}\sup_{i}(nK/\theta_i)^{1/2}(d_{i}^{\tau
	})^{-1/2}\sum_{j=1}^{n}\frac{|A_{ij}-P_{ij}|\log^{1/2}(n)(n_{g_{j}^{0}}^{%
			\tau })^{-1/2}(\theta _{j}^{\tau })^{1/2}\psi _{n}}{\sqrt{\hat{d}_{j}^{\tau
			}d_{j}^{\tau }}}  \notag \\
	\leq & 2.24C_1^{1/2}c_1^{-1/2}\sup_{i}\sum_{j=1}^{n}\frac{%
		|A_{ij}-P_{ij}|\log ^{1/2}(n)\theta _{j}^{1/2}\psi _{n}}{\theta
		_{i}^{1/2}(d_{i}^{\tau })^{1/2}d_{j}^{\tau }}  \notag \\
	\leq & 2.24C_1^{1/2}c_1^{-1/2}\sup_{i}\left|\sum_{j\neq i}\frac{%
		(A_{ij}-P_{ij})\log ^{1/2}(n)\theta _{j}^{1/2}\psi _{n}}{\theta
		_{i}^{1/2}(d_{i}^{\tau })^{1/2}d_{j}^{\tau }}\right|+
	4.48C_1^{1/2}c_1^{-1/2}\sup_{i}\sum_{j=1}^{n}\frac{P_{ij}\log
		^{1/2}(n)\theta _{j}^{1/2}\psi _{n}}{\theta _{i}^{1/2}(d_{i}^{\tau
		})^{1/2}d_{j}^{\tau }},
	\end{align}%
	where the first inequality holds by the definition of spectral norm; the
	second inequality holds by the facts that $\sup_j
	(n_{g_j^0}^\tau)^{1/2}(\theta_j^\tau)^{-1/2}||\hat{u}_j||\leq \psi_n$ and
	that, by Bernstein inequality, 
	\begin{equation}  \label{eq:dhat-d0}
	\sup_i |(\hat{d}_i^\tau)^{1/2} - (d_i^\tau)^{1/2}| \leq 2.09 \log^{1/2}(n)
	\quad a.s.,
	\end{equation}
	the third inequality holds by Assumption \ref{ass:nk3} and the fact that 
	\begin{align*}
	(\hat{d}_i^\tau)^{1/2} \geq 0.9791 (d_i^\tau)^{1/2} \quad a.s.,
	\end{align*}
	and the last inequality holds because $\left\vert A_{ij}-P_{ij}\right\vert
	\leq (A_{ij}-P_{ij})+2P_{ij}.$ In addition, by Lemma \ref{lem:Pij3} and
	Assumptions \ref{ass:id3} and \ref{ass:nk3}, 
	\begin{align}  \label{eq:II2a}
	\sup_{i}\sum_{j=1}^{n}\frac{P_{ij}\log ^{1/2}(n)\theta _{j}^{1/2}\psi _{n}}{%
		\theta _{i}^{1/2}(d_{i}^{\tau })^{1/2}d_{j}^{\tau }}\leq &
	\sup_{i}\sum_{j=1}^{n}\frac{ \log ^{1/2}(n)\rho
		_{n}(d_{i}d_{j})^{1/2}(\theta _{i}\theta _{j})^{1/2}\theta_j^{1/2}\psi_n}{%
		n\theta_i^{1/2}(d_{i}^{\tau })^{1/2}d_{j}^{\tau }}  \notag \\
	\leq & \frac{\log ^{1/2}(n)\rho _{n}\psi_n}{(\mu _{n}^{\tau })^{1/2}}%
	(\sum_{j}\theta _{j})/n=\frac{\log ^{1/2}(n)\rho _{n}\psi_n}{(\mu _{n}^{\tau
		})^{1/2}}\mbox{ }a.s.\text{ }
	\end{align}%
	By the Bernstein inequality and the facts that 
	\begin{equation*}
	\sup_{j}\left|\frac{(A_{ij}-P_{ij})\theta _{j}^{1/2}}{d_{j}^{\tau }}%
	\right|\leq \frac{\bar{\theta}^{1/2}}{\mu _{n}^{\tau }}
	\end{equation*}
	and 
	\begin{equation*}
	\sum_{j}\frac{\mathbb{E}(A_{ij}-P_{ij})^{2}\theta _{j}}{(d_{j}^{\tau })^{2}}%
	\leq \sum_{j}\frac{ P_{ij}\theta _{j}}{(d_{j}^{\tau })^{2}}\leq \frac{d_{i}%
		\bar{\theta}}{(\mu _{n}^{\tau })^{2}},
	\end{equation*}
	we have 
	\begin{align}  \label{eq:II2b}
	\sup_{i}\left|\sum_{j\neq i}\frac{(A_{ij}-P_{ij})\log ^{1/2}(n)\theta
		_{j}^{1/2}\psi _{n}}{\theta _{i}^{1/2}(d_{i}^{\tau })^{1/2}d_{j}^{\tau }}%
	\right| \leq 3.9\sup_i \frac{\log(n)\psi_n \bar{\theta}^{1/2}}{%
		\theta_i^{1/2}(d_i^\tau)^{1/2}\mu_n^\tau}(\log^{1/2}(n) + d_i^{1/2}) \leq 
	\frac{4\bar{\theta}^{1/2}\log (n)\psi _{n}}{\mu _{n}^{\tau }\underline{
			\theta }^{1/2}}\mbox{
	}a.s.
	\end{align}
	
	Combining \eqref{eq:II20}--\eqref{eq:II2b} with the fact that $\frac{%
		\log^{1/2} (n)\bar{\theta}^{1/2}}{(\mu _{n}^{\tau })^{1/2}\underline{\theta }
		^{1/2}\rho _{n}} \leq 0.01$, we have 
	\begin{equation}
	\sup_{i}(n_{g_{i}^{0}}^{\tau })^{1/2}(\theta _{i}^{\tau })^{-1/2}(\hat{d}%
	_{i}^{\tau })^{-1/2}\Vert ([A]_{i\cdot }-[P]_{i\cdot })(D_{\tau }^{-1/2}-%
	\mathcal{D}_{\tau }^{-1/2})\hat{U}_{1n}\hat{O}_{n}\Vert \leq
	4.57C_1^{1/2}c_1^{-1/2}\frac{\rho _{n}\log ^{1/2}(n)\psi _{n}}{(\mu
		_{n}^{\tau })^{1/2}}.  \label{eq:II2c}
	\end{equation}
	
	\textbf{Step 4: Bound for $T_{2,3}$}\newline
	By the triangle inequality, 
	\begin{align}
	& \sup_{i}(n_{g_{i}^{0}}^{\tau })^{1/2}(\theta _{i}^{\tau })^{-1/2}(\hat{d}%
	_{i}^{\tau })^{-1/2}\Vert ([A]_{i\cdot }-[P]_{i\cdot })\mathcal{D}_{\tau
	}^{-1/2}\hat{U}_{1n}\hat{O}_{n}\Vert  \notag  \label{eq:II30DC} \\
	\leq & \sup_{i}(n_{g_{i}^{0}}^{\tau })^{1/2}(\theta _{i}^{\tau })^{-1/2}(%
	\hat{d}_{i}^{\tau })^{-1/2}\Vert ([A]_{i\cdot }-[P]_{i\cdot })\mathcal{D}%
	_{\tau }^{-1/2}\hat{U}_{1n}^{(i)}\hat{O}_{n}^{(i)}\Vert  \notag \\
	& +\sup_{i}(n_{g_{i}^{0}}^{\tau })^{1/2}(\theta _{i}^{\tau })^{-1/2}(\hat{d}%
	_{i}^{\tau })^{-1/2}\Vert ([A]_{i\cdot }-[P]_{i\cdot })\mathcal{D}_{\tau
	}^{-1/2}(\hat{U}_{1n}\hat{O}_{n}-\hat{U}_{1n}^{(i)}\hat{O}_{n}^{(i)})\Vert 
	\notag \\
	=& T_{2,3,1}+T_{2,3,2}.
	\end{align}%
	Let $\mathcal{N}_{n} = \text{diag}((n_{g_{1}^{0}}^{\tau })^{1/2}(\theta
	_{1}^{\tau })^{-1/2},\cdots,(n_{g_{n}^{0}}^{\tau })^{1/2}(\theta _{n}^{\tau
	})^{-1/2})$. Note that 
	\begin{align}
	\Vert \hat{U}_{1n}^{(i)}\hat{O}_{n}^{(i)}\Vert _{2\rightarrow \infty }\leq &
	\sup_{j}\left[(n_{g_{j}^{0}}^{\tau })^{-1/2}(\theta _{j}^{\tau })^{1/2}%
	\right]\sup_{i,j}(n_{g_{j}^{0}}^{\tau })^{1/2}(\theta _{j}^{\tau
	})^{-1/2}\Vert \lbrack \hat{U}_{1n}^{(i)}\hat{O}_{n}^{(i)}]_{j\cdot }\Vert 
	\notag  \label{eq:II31DC} \\
	\leq & (\overline{\theta }K/(nc_1))^{1/2}\sup_{i}\Vert \mathcal{N}_{n}\hat{U}%
	_{1n}^{(i)}\hat{O}_{n}^{(i)}\Vert _{2\rightarrow \infty }  \notag \\
	\leq & (\overline{\theta }K/(nc_1))^{1/2}\biggl(\sup_{i}\Vert \mathcal{N}%
	_{n}(\hat{U}_{1n}^{(i)}\hat{O}_{n}^{(i)}-\hat{U}_{1n}\hat{O}_{n})\Vert
	_{2\rightarrow \infty }+\Vert \mathcal{N}_{n}\hat{U}_{1n}\hat{O}_{n}\Vert
	_{2\rightarrow \infty }\biggr)  \notag \\
	\leq & c_1^{-1/2}(\overline{\theta }K/n)^{1/2}\biggl[\frac{%
		1676C_1^{1/2}\log^{1/2}(n)}{(\mu_n^\tau)^{1/2}|\sigma _{Kn}|}\left(\psi
	_{n}+\rho _{n}+\frac{(\frac{1}{K} + \frac{\log(5)}{\log(n)})^{1/2} \rho
		_{n}^{1/2} \overline{\theta }^{1/4}}{\underline{\theta }^{1/4}}\right)+\psi
	_{n}\biggr]  \notag \\
	\leq & 1.01c_1^{-1/2}(\overline{\theta }K/n)^{1/2}(\psi _{n}+1),
	\end{align}%
	where the second inequality holds by Assumption \ref{ass:nk3}, the third
	inequality holds by triangle inequality, the fourth inequality holds by
	Lemma \ref{lem:looDC}, and the last inequality holds because under
	Assumption \ref{ass:rate3} 
	\begin{equation*}
	1676C_1^{1/2}\log ^{1/2}(n)(\mu _{n}^{\tau })^{-1/2}|\sigma _{Kn}^{-1}|%
	\biggl(\rho _{n}+\frac{(\frac{1}{K} + \frac{\log(5)}{\log(n)})^{1/2} \rho
		_{n}^{1/2} \overline{\theta }^{1/4}}{\underline{\theta }^{1/4}} \biggr) \leq
	0.01.
	\end{equation*}
	
	Then, by Lemma \ref{lem:V1n3}, \eqref{eq:di3}, and the facts that $\hat{U}%
	_{1n}^{(i)}\hat{O}_n^{(i)}$ is independent of $[A]_{i\cdot} - [P]_{i\cdot}$, 
	$||\hat{U}_{1n}^{(i)}\hat{O}_n^{(i)}||\leq 1$, we have 
	\begin{align}  \label{eq:II31}
	T_{2,3,1}=& \sup_{i}(n_{g_{i}^{0}}^{\tau })^{1/2}(\theta _{i}^{\tau
	})^{-1/2}(\hat{d}_{i}^{\tau })^{-1/2}\Vert ([A]_{i\cdot }-[P]_{i\cdot })%
	\mathcal{D}_{\tau }^{-1/2}\hat{U}_{1n}^{(i)}\hat{O}_{n}^{(i)}\Vert  \notag \\
	\leq & 1.022\sup_{i}(n_{g_{i}^{0}}^{\tau })^{1/2}(\theta _{i}^{\tau
	})^{-1/2}(d_{i}^{\tau })^{-1/2}\Vert ([A]_{i\cdot }-[P]_{i\cdot })\mathcal{D}%
	_{\tau }^{-1/2}\hat{U}_{1n}^{(i)}\hat{O}_{n}^{(i)}\Vert  \notag \\
	\leq & 6.14C_1^{1/2}\biggl(\frac{1.01c_1^{-1/2}(\log (n)+\log(5)K)\bar{\theta%
		}^{1/2}}{\mu _{n}^{\tau }\underline{\theta }^{1/2}}(\psi _{n}+1) \vee \frac{%
		(\log(n)+\log(5)K)^{1/2}\rho _{n}^{1/2}\bar{\theta}^{1/4}}{(\mu _{n}^{\tau
		})^{1/2}\underline{\theta }^{1/4}K^{1/2}}\biggr)  \notag \\
	\leq & 6.21C_1^{1/2}c_1^{-1/2}\biggl(\frac{(\log (n)+\log(5)K)\bar{\theta}%
		^{1/2}}{\mu _{n}^{\tau }\underline{ \theta }^{1/2}}\psi _{n} + \frac{%
		(\log(n)+\log(5)K)^{1/2}\rho _{n}^{1/2}\bar{\theta} ^{1/4}}{(\mu _{n}^{\tau
		})^{1/2}\underline{\theta }^{1/4}K^{1/2}}\biggr),
	\end{align}
	where the last inequality holds because 
	\begin{align*}
	\frac{1.01(\log (n)+\log(5)K)\bar{\theta}^{1/2}}{\mu _{n}^{\tau }\underline{
			\theta }^{1/2}} \leq \frac{0.01(\log(n)+\log(5)K)^{1/2}\rho _{n}^{1/2}\bar{%
			\theta} ^{1/4}}{(\mu _{n}^{\tau })^{1/2}\underline{\theta }^{1/4}}.
	\end{align*}
	
	In addition, from the derivation of \eqref{eq:II31DC}, we have 
	\begin{align}
	& \Vert [ \hat{U}_{1n}\hat{O}_{n}-\hat{U}_{1n}^{(i)}\hat{O}%
	_{n}^{(i)}]_{j\cdot }\Vert  \notag \\
	\leq & 1676C_1^{1/2}c_1^{-1/2}(\theta _{j}K/n)^{1/2}\log ^{1/2}(n)(\mu
	_{n}^{\tau })^{-1/2}|\sigma _{Kn}^{-1}|\biggl(\psi _{n}+\rho _{n}+\frac{(%
		\frac{1}{K} + \frac{\log(5)}{\log(n)})^{1/2} \rho _{n}^{1/2} \overline{%
			\theta }^{1/4}}{\underline{\theta }^{1/4}}\biggr)  \notag \\
	= & 1676C_1^{1/2} c_1^{-1/2}(\theta _{j}K/n)^{1/2}\tilde{\gamma}_n,
	\label{eq:II3c}
	\end{align}
	where $\tilde{\gamma}_n = \log ^{1/2}(n)(\mu _{n}^{\tau })^{-1/2}|\sigma
	_{Kn}^{-1}|\biggl(\psi _{n}+\rho _{n}+\frac{(\frac{1}{K} + \frac{\log(5)}{%
			\log(n)})^{1/2} \rho _{n}^{1/2} \overline{\theta }^{1/4}}{\underline{\theta }%
		^{1/4}}\biggr)$.
	
	Let $V(g)=(\hat{U}_{1n}\hat{O}_{n}-\hat{U}_{1n}^{(i)}\hat{O}_{n}^{(i)})g$
	for some $g\in S^{K-1}$ and $V_{j}(g)$ be the $j$-th element of $V(g)$.
	Then, 
	\begin{align}
	T_{2,3,2}\leq & C_1^{1/2}\sup_{i}(n/K)^{1/2}\sup_{g\in S^{K-1}}\left[%
	\left|\sum_{j\neq i}\frac{(A_{ij}-P_{ij})}{\theta _{i}^{1/2}(\hat{d}%
		_{i}^{\tau })^{1/2}(d_{j}^{\tau })^{1/2}}V_{j}(g)\right|+\left|\frac{-P_{ii}%
	}{\theta _{i}^{1/2}(\hat{d}_{i}^{\tau })^{1/2}(d_{i}^{\tau })^{1/2}}%
	V_{i}(g)\right|\right]  \notag  \label{eq:II32DC} \\
	\leq & 1.0209 C_1^{1/2}\sup_{i}(n/K)^{1/2}\sup_{g\in S^{K-1}}\left[%
	\left|\sum_{j\neq i}\frac{(A_{ij}-P_{ij})}{\theta _{i}^{1/2}(d_{i}^{\tau
		})^{1/2}(d_{j}^{\tau })^{1/2}}V_{j}(g)\right|+\left|\frac{-P_{ii}}{\theta
		_{i}^{1/2}( d_{i}^{\tau })^{1/2}(d_{i}^{\tau })^{1/2}}V_{i}(g)\right|\right]
	\notag \\
	\leq & 1712 C_1c_1^{-1/2}\sup_{i}\biggl[\sum_{j\neq i}\frac{%
		(A_{ij}+P_{ij})\theta _{j}^{1/2}}{\theta _{i}^{1/2}(d_{i}^{\tau
		})^{1/2}(d_{j}^{\tau })^{1/2}}+\frac{P_{ii}}{d_{i}^{\tau }}\biggr]\tilde{%
		\gamma}_{n}  \notag \\
	\leq & 1712 C_1c_1^{-1/2}\sup_{i}\biggl[\left|\sum_{j\neq i}\frac{%
		(A_{ij}-P_{ij})\theta _{j}^{1/2}}{\theta _{i}^{1/2}(d_{i}^{\tau
		})^{1/2}(d_{j}^{\tau })^{1/2}}\right|+\sum_{j=1}^{n}\frac{2P_{ij}\theta
		_{j}^{1/2}}{\theta _{i}^{1/2}(d_{i}^{\tau })^{1/2}(d_{j}^{\tau })^{1/2}}%
	\biggr]\tilde{\gamma}_{n}  \notag \\
	\leq & 1712 C_1c_1^{-1/2}\biggl(\frac{3\rho_n^{1/2}\log ^{1/2}(n)\overline{%
			\theta }^{1/4}}{(\mu _{n}^{\tau })^{1/2}\underline{\theta }^{1/4}}+2\rho _{n}%
	\biggr)\tilde{\gamma}_{n}\leq 3425 C_1c_1^{-1/2}\rho _{n}\tilde{\gamma}_{n},
	\notag \\
	&
	\end{align}%
	where the first inequality holds due to the definition of $L_{2}$ norm of a $%
	K\times 1$ vector and Assumption \ref{ass:nk3}, the second inequality holds
	by \eqref{eq:di3}, the third inequality holds by \eqref{eq:II3c}, the fourth
	inequality holds by the triangle inequality, the fifth inequality holds
	because by Bernstein inequality, 
	\begin{equation*}
	\left|\sum_{j\neq i}\frac{(A_{ij}-P_{ij})\theta _{j}^{1/2}}{\theta
		_{i}^{1/2}(d_{i}^{\tau })^{1/2}(d_{j}^{\tau })^{1/2}}\right|\leq 3\left(%
	\frac{\rho_n^{1/2}\log ^{1/2}(n)\overline{\theta }^{1/4}}{(\mu _{n}^{\tau
		})^{1/2}\underline{\theta }^{1/4}} \vee \frac{\log (n)\overline{\theta }%
		^{1/2}}{\mu _{n}^{\tau }\underline{\theta }^{1/2}}\right)\mbox{ }a.s.
	\end{equation*}%
	and 
	\begin{equation*}
	\sum_{j=1}^{n}\frac{P_{ij}\theta _{j}^{1/2}}{\theta _{i}^{1/2}(d_{i}^{\tau
		})^{1/2}(d_{j}^{\tau })^{1/2}}\leq \sum_{j=1}^{n}\frac{\rho
		_{n}(d_{i}d_{j})^{1/2}\theta _{j}}{n(d_{i}^{\tau })^{1/2}(d_{j}^{\tau
		})^{1/2}}\leq \rho _{n},
	\end{equation*}%
	and the last inequality holds because $\frac{\log(n)\overline{\theta}^{1/2}}{%
		\underline{\theta}^{1/2}\mu _{n}^{\tau }\rho_n^{1/2}} \leq 0.0001$.
	
	Combining \eqref{eq:II30DC}, \eqref{eq:II31}, and \eqref{eq:II32DC}, we have 
	\begin{equation}
	T_{2,3}\leq 3432 C_1c_1^{-1/2}\rho _{n}\log ^{1/2}(n)(\mu _{n}^{\tau
	})^{-1/2}|\sigma _{Kn}^{-1}|\biggl[\psi _{n}+\rho _{n}+\frac{(\frac{1}{K}+%
		\frac{\log(5)}{\log(n)})^{1/2}\rho _{n}^{1/2}\overline{\theta }^{1/4}}{%
		\underline{\theta }^{1/4}}\biggr]\mbox{ }a.s.,  \label{eq:II3DCfinal}
	\end{equation}%
	where we use the fact that 
	\begin{equation*}
	\frac{(\log (n)+\log(5)K)\bar{\theta}^{1/2}}{\mu _{n}^{\tau }\underline{
			\theta }^{1/2}} \leq \log^{1/2}(n)(\mu_n^\tau)^{-1/2}|\sigma_{Kn}^{-1}|.
	\end{equation*}
	
	\textbf{Step 5: Bound for $T_1+T_{2,1}+T_{2,2}+T_{2,3}$}\newline
	Combining \eqref{eq:IDCfinal}, \eqref{eq:II1DCfinal}, \eqref{eq:II2c}, and %
	\eqref{eq:II3DCfinal}, we have 
	\begin{align*}
	& \sup_{i}(n_{g_{i}^{0}}^{\tau })^{1/2}(\theta _{i}^{\tau })^{-1/2}\Vert 
	\hat{ \Lambda}_{i}-\Lambda _{i}\Vert \\
	\leq & 3450C_1c_1^{-1/2}\rho _{n}\log ^{1/2}(n)(\mu _{n}^{\tau
	})^{-1/2}|\sigma _{Kn}^{-1}|\biggl[\psi _{n}+\rho _{n}+\frac{(\frac{1}{K}+%
		\frac{\log(5)}{\log(n)})^{1/2}\rho _{n}^{1/2}\overline{\theta } ^{1/4}}{%
		\underline{\theta }^{1/4}}\biggr]\mbox{ }a.s.
	\end{align*}
\end{proof}

In the proof of Lemma \ref{lem:B3} we utilize Lemma \ref{lem:looDC} below
whose proof calls Lemmas \ref{lem:Li-LDC} and \ref{lem:DiDC}.

\begin{lem}
	\label{lem:Li-LDC} Suppose that conditions in Theorem \ref{thm:main_DC}
	hold. Then,  
	\begin{equation*}
	\sup_{i}\Vert L_{\tau }^{(i)}-\mathcal{L}_{\tau }^{\prime }\Vert \leq
	4.4(\log (n)/\mu _{n}^{\tau })^{1/2}\mbox{ }a.s.
	\end{equation*}
\end{lem}

\begin{proof}
	Let $\tilde{L}_{\tau }=\mathcal{D}_{\tau }^{-1/2}A\mathcal{D}_{\tau }^{-1/2}$%
	. Note that $\Vert L_{\tau }^{(i)}-\mathcal{L}_{\tau }^{\prime }\Vert \leq
	\Vert \tilde{L}_{\tau }-\mathcal{L}_{\tau }^{\prime \tau }\Vert +\Vert 
	\tilde{L}_{\tau }-L_{\tau }^{(i)}\Vert .$ In the proof of Lemma \ref{lem:dk3}%
	, we have shown that 
	\begin{equation*}
	\Vert \tilde{L}_{\tau }-\mathcal{L}_{\tau }^{\prime }\Vert \leq 4.39(\log
	(n)/\mu _{n}^{\tau })^{1/2}\text{ }\mbox{ }a.s.
	\end{equation*}%
	It remains to show that, for $n$ sufficiently large, 
	\begin{equation*}
	\sup_{i}\Vert \tilde{L}_{\tau }-L_{\tau }^{(i)}\Vert =\sup_{i}\Vert \mathcal{%
		D}_{\tau }^{-1/2}(A^{(i)}-A)\mathcal{D}_{\tau }^{-1/2}\Vert \leq 0.01(\log
	(n)/\mu _{n}^{\tau })^{1/2}\mbox{ }a.s.
	\end{equation*}%
	By construction, 
	\begin{equation}
	\lbrack A-A^{(i)}]_{st}=%
	\begin{cases}
	0 & \quad s\neq i,t\neq i \\ 
	A_{si}-P_{si} & \quad s\neq i,t=i \\ 
	A_{it}-P_{it} & \quad s=i,t\neq i \\ 
	0 & \quad s=t=i%
	\end{cases}%
	.  \label{AA}
	\end{equation}%
	Then, for $n$ sufficiently large, 
	\begin{align*}
	\sup_{i}\Vert \tilde{L}_{\tau }-L_{\tau }^{(i)}\Vert & \leq \sup_{i}\Vert 
	\mathcal{D}_{\tau }^{-1/2}(A^{(i)}-A)\mathcal{D}_{\tau }^{-1/2}\Vert _{F} \\
	& \leq \sup_{i}\biggl(2\sum_{j\neq i}\frac{(A_{ij}-P_{ij})^{2}}{d_{i}^{\tau
		}d_{j}^{\tau }}\biggr)^{1/2} \\
	& \leq \sup_{i}\biggl(2\sum_{j\neq i}\frac{|A_{ij}-P_{ij}|}{d_{i}^{\tau
		}d_{j}^{\tau }}\biggr)^{1/2} \\
	& \leq \sup_{i}\biggl(2|\sum_{j\neq i}\frac{A_{ij}-P_{ij}}{d_{i}^{\tau
		}d_{j}^{\tau }}|+4\sum_{j=1}^{n}\frac{P_{ij}}{d_{i}^{\tau }d_{j}^{\tau }}%
	\biggr)^{1/2} \\
	& \leq \left(\frac{8.46\log ^{1/2}(n)}{(\mu _{n}^{\tau })^{3/2}}+\frac{%
		8.46\log (n)}{(\mu _{n}^{\tau })^{2}}+\frac{4}{\mu _{n}^{\tau }}\right)^{1/2}
	\\
	& \leq 0.01 (\log (n)/\mu _{n}^{\tau })^{1/2}\mbox{ }a.s.,
	\end{align*}%
	where the first inequality holds because $\Vert A\Vert \leq \Vert A\Vert
	_{F} $ for a generic matrix $A$, the second inequality holds by (\ref{AA}),
	the third inequality holds because $|A_{ij}-P_{ij}|\leq 1$, the fourth
	inequality holds because $|A_{ij}-P_{ij}|\leq A_{ij}-P_{ij}+2P_{ij}$, the
	fifth inequality holds by the fact that 
	\begin{equation*}
	\sum_{j=1}^{n}\frac{P_{ij}}{d_{i}^{\tau }d_{j}^{\tau }}\leq \sum_{j=1}^{n}%
	\frac{P_{ij}}{d_{i}^{\tau }\mu _{n}^{\tau }}=1/\mu _{n}^{\tau },
	\end{equation*}
	and by \eqref{eq:dhat-d0}, 
	\begin{align}  \label{eq:dhat-d}
	|\hat{d}_j^\tau - d_j^\tau| = |(\hat{d}_j^\tau)^{1/2} - (d_j^\tau)^{1/2}||(%
	\hat{d}_j^\tau)^{1/2} + (d_j^\tau)^{1/2}| \leq 2.09 \log^{1/2}(n)\times
	2.0209 (d_j^\tau)^{1/2} = 4.23 (\log(n)d_j^\tau)^{1/2},
	\end{align}
	and 
	\begin{equation*}
	\sup_{i}\left|\sum_{j\neq i}\frac{A_{ij}-P_{ij}}{d_{i}^{\tau }d_{j}^{\tau }}%
	\right|\leq 4.23\left(\frac{\log ^{1/2}(n)}{(\mu _{n}^{\tau })^{3/2}}+\frac{%
		\log (n)}{(\mu _{n}^{\tau })^{2}}\right)\quad a.s.
	\end{equation*}
\end{proof}

\begin{lem}
	\label{lem:looDC} Recall $\mathcal{N}_{n}=$diag$((n_{g_{1}^{0}}^{\tau
	})^{1/2}(\theta _{1}^{\tau })^{-1/2},\cdots ,(n_{g_{n}^{0}}^{\tau
	})^{1/2}(\theta _{n}^{\tau })^{-1/2}).$ Suppose that conditions in Theorem %
	\ref{thm:main_DC} hold and $\Vert \mathcal{N}_{n}\hat{U}_{1n}\Vert
	_{2\rightarrow \infty }\leq \psi _{n}$. Then, 
	\begin{equation*}
	\sup_{i}\Vert \mathcal{N}_{n}[(\hat{U}_{1n}^{(i)})\hat{O}_{n}^{(i)}-\hat{U}%
	_{1n}\hat{O}_{n}]\Vert _{2\rightarrow \infty } \leq\frac{1676C_1^{1/2}%
		\log^{1/2}(n)}{(\mu_n^\tau)^{1/2}|\sigma _{Kn}|}\left(\psi _{n}+\rho _{n}+%
	\frac{(\frac{1}{K} + \frac{\log(5)}{\log(n)})^{1/2} \rho _{n}^{1/2} 
		\overline{\theta }^{1/4}}{\underline{\theta }^{1/4}}\right) \quad a.s.
	\end{equation*}
\end{lem}

\begin{proof}
	Recall the definitions of $\hat{H}_n$, $\hat{O}_n$, $\hat{H}^{(i)}_n$, and $%
	\hat{O}^{(i)}_n$ before Lemma \ref{lem:B3}. Let $\gamma
	_{n}=(\log(n)/\mu_n^\tau)^{1/2}$ and recall that $\Vert L_{\tau }^{\prime }-%
	\mathcal{L}_{\tau }^{\prime }\Vert \leq 7\gamma_n$ a.s. By Lemma 3 in %
	\citet{abbe2017}\footnote{%
		Note that in the notation of \cite{abbe2017}, $(H,\text{sgn}(H))=(\hat{H}
		_{n},\hat{O}_{n})$ (or $(\hat{H}_{n}^{(i)},\hat{O}_{n}^{(i)})$), $U^{\ast
		}=U_{1n}$, $E=L_{\tau }^{\prime }-\mathcal{L}_{\tau }^{\prime }$ (or $%
		L_{\tau }^{(i)}-\mathcal{L}_{\tau }^{\prime }$), and $\bar{\gamma}=7\gamma
		_{n}/(|\sigma _{Kn}|-7\gamma _{n})$ for some absolute constant $c>0$.} and
	Lemma \ref{lem:dk3}, we have 
	\begin{equation}  \label{eq:OH}
	\Vert \hat{O}_{n}-\hat{H}_{n}\Vert ^{1/2} \leq \frac{7r_n/(|\sigma_{Kn}| -
		7r_n)}{1-7r_n/(|\sigma_{Kn}| - 7r_n)} \leq
	7.01(\log(n)/\mu_n^\tau)^{1/2}|\sigma_{Kn}^{-1}| \quad a.s.,
	\end{equation}%
	where we use the fact that 
	\begin{align*}
	(\log(n)/\mu_n^\tau)^{1/2}|\sigma_{Kn}|^{-1} \leq 0.0001.
	\end{align*}
	Similarly, by Lemma \ref{lem:Li-LDC}, we have 
	\begin{equation}  \label{eq:OHi}
	\sup_{i}\Vert \hat{O}_{n}^{(i)}-\hat{H}_{n}^{(i)}\Vert ^{1/2} \leq
	4.41(\log(n)/\mu_n^\tau)^{1/2}|\sigma_{Kn}^{-1}|\mbox{ }a.s.
	\end{equation}%
	Then%
	\begin{align*}
	& \Vert \mathcal{N}_{n}(\hat{U}_{1n}\hat{O}_{n}-\hat{U}_{1n}^{(i)}\hat{O}%
	_{n}^{(i)})\Vert _{2\rightarrow \infty } \\
	\leq & \Vert \mathcal{N}_{n}\hat{U}_{1n}(\hat{O}_{n}-\hat{H}_{n})\Vert
	_{2\rightarrow \infty }+\Vert \mathcal{N}_{n}\hat{U}_{1n}^{(i)}(\hat{O}%
	_{n}^{(i)}-\hat{H}_{n}^{(i)})\Vert _{2\rightarrow \infty }+\Vert \mathcal{N}%
	_{n}(\hat{U}_{1n}\hat{H}_{n}-\hat{U}_{1n}^{(i)}\hat{H}_{n}^{(i)})\Vert
	_{2\rightarrow \infty } \\
	\leq & \Vert \mathcal{N}_{n}\hat{U}_{1n}\Vert _{2\rightarrow \infty }\Vert 
	\hat{O}_{n}-\hat{H}_{n}\Vert +\Vert \mathcal{N}_{n}\hat{U}_{1n}^{(i)}\Vert
	_{2\rightarrow \infty }\Vert \hat{O}_{n}^{(i)}-\hat{H}_{n}^{(i)}\Vert +\Vert 
	\mathcal{N}_{n}(\hat{U}_{1n}\hat{H}_{n}-\hat{U}_{1n}^{(i)}\hat{H}%
	_{n}^{(i)})\Vert _{2\rightarrow \infty } \\
	\leq & 49.15\psi _{n}\gamma _{n}^{2}\sigma _{Kn}^{-2}+19.45\Vert \mathcal{N}%
	_{n}\hat{U}_{1n}^{(i)}\hat{O}_{n}^{(i)}\Vert _{2\rightarrow \infty }\gamma
	_{n}^{2}\sigma _{Kn}^{-2}+\Vert \mathcal{N}_{n}(\hat{U}_{1n}\hat{H}_{n}-\hat{%
		U}_{1n}^{(i)}\hat{H}_{n}^{(i)})\Vert _{2\rightarrow \infty } \\
	\leq & 68.6\psi _{n}\gamma _{n}^{2}\sigma _{Kn}^{-2}+19.45\Vert \mathcal{N}%
	_{n}(\hat{U}_{1n}\hat{O}_{n}-\hat{U}_{1n}^{(i)}\hat{O}_{n}^{(i)})\Vert
	_{2\rightarrow \infty }\gamma _{n}^{2}\sigma _{Kn}^{-2}+\Vert \mathcal{N}%
	_{n}(\hat{U}_{1n}\hat{H}_{n}-\hat{U}_{1n}^{(i)}\hat{H}_{n}^{(i)})\Vert
	_{2\rightarrow \infty },
	\end{align*}%
	where the first inequality holds by the triangle inequality, the second
	inequality holds by the fact that $\Vert AB\Vert _{2\rightarrow \infty }\leq
	\Vert A\Vert _{2\rightarrow \infty }\Vert B\Vert ,$ the third inequality
	holds by \eqref{eq:OH}, \eqref{eq:OHi}, and the assumption that $\Vert 
	\mathcal{N}_{n}\hat{U}_{1n}\Vert _{2\rightarrow \infty }\leq \psi _{n}$, and
	the last inequality holds by the triangle inequality and another use of $%
	\Vert \mathcal{N}_{n}\hat{U}_{1n}\Vert _{2\rightarrow \infty }\leq \psi _{n}$%
	. By rearranging terms and the fact that $\gamma_n|\sigma _{Kn}^{-1}| \leq
	0.0001$, we have 
	\begin{equation}
	\Vert \mathcal{N}_{n}(\hat{U}_{1n}\hat{O}_{n}-\hat{U}_{1n}^{(i)}\hat{O}%
	_{n}^{(i)})\Vert _{2\rightarrow \infty }\leq 68.74\biggl[\log(n)
	(\mu_n^\tau)^{-1}\sigma _{Kn}^{-2}\psi _{n}+\Vert \mathcal{N}_{n}(\hat{U}%
	_{1n}\hat{H}_{n}-\hat{U}_{1n}^{(i)}\hat{H}_{n}^{(i)})\Vert _{2\rightarrow
		\infty }\biggr].  \label{eq:Lem7a}
	\end{equation}%
	In addition, by Lemma 3 in \citet{abbe2017},\footnote{%
		Note that in the notation of \cite{abbe2017}, $H=\hat{H}_{n}$ (or $\hat{H}%
		_{n}^{(i)}$), $E=L_{\tau }^{\prime }-\mathcal{L}_{\tau }^{\prime }$ (or $%
		L_{\tau }^{(i)}-\mathcal{L}_{\tau }^{\prime }$), and $\Lambda =\widehat{%
			\Sigma }_{1n}$ (or $\widehat{\Sigma }_{1n}^{(i)}$) for some absolute
		constant $c>0$.} Lemma \ref{lem:dk3}, and Lemma \ref{lem:Li-LDC}, we have 
	\begin{equation*}
	\Vert \hat{\Sigma}_{1n}\hat{H}_{n}-\hat{H}_{n}\hat{\Sigma}_{1n}\Vert \leq
	2\Vert L_{\tau }^{\prime }-\mathcal{L}_{\tau }^{\prime }\Vert \leq 14
	\gamma_n \mbox{ }a.s.
	\end{equation*}%
	and 
	\begin{equation*}
	\sup_{i}\Vert \hat{\Sigma}_{1n}^{(i)}\hat{H}_{n}^{(i)}-\hat{H}_{n}^{(i)}\hat{%
		\Sigma}_{1n}^{(i)}\Vert \leq 2\Vert L_{\tau }^{(i)}-\mathcal{L}_{\tau
	}^{\prime }\Vert \leq 8.8\gamma_n\mbox{ }a.s.
	\end{equation*}%
	Therefore, 
	\begin{align}
	& \Vert \mathcal{N}_{n}(\hat{U}_{1n}\hat{H}_{n}-\hat{U}_{1n}^{(i)}\hat{H}%
	_{n}^{(i)})\Vert _{2\rightarrow \infty }  \notag \\
	\leq & \Vert \mathcal{N}_{n}\hat{U}_{1n}(\hat{\Sigma}_{1n}\hat{H}_{n}-\hat{H}%
	_{n}\hat{\Sigma}_{1n})\hat{\Sigma}_{1n}^{-1}\Vert _{2\rightarrow \infty
	}+\Vert \mathcal{N}_{n}\hat{U}_{n}^{(i)}(\hat{\Sigma}_{1n}^{(i)}\hat{H}%
	_{n}^{(i)}-\hat{H}_{n}^{(i)}\hat{\Sigma}_{1n}^{(i)})(\hat{\Sigma}%
	_{1n}^{(i)})^{-1}\Vert _{2\rightarrow \infty }  \notag \\
	& +\Vert \mathcal{N}_{n}(\hat{U}_{1n}\hat{\Sigma}_{1n}\hat{H}_{n}\hat{\Sigma}%
	_{1n}^{-1}-\hat{U}_{1n}^{(i)}\hat{\Sigma}_{1n}^{(i)}\hat{H}_{n}^{(i)}(\hat{%
		\Sigma}_{1n}^{(i)})^{-1})\Vert _{2\rightarrow \infty }  \notag \\
	\leq & \Vert \mathcal{N}_{n}\hat{U}_{1n}\Vert _{2\rightarrow \infty }\Vert 
	\hat{\Sigma}_{1n}\hat{H}_{n}-\hat{H}_{n}\hat{\Sigma}_{1n}\Vert \Vert \hat{%
		\Sigma}_{1n}^{-1}\Vert +\Vert \mathcal{N}_{n}\hat{U}_{n}^{(i)}\Vert
	_{2\rightarrow \infty }\Vert \hat{\Sigma}_{1n}^{(i)}\hat{H}_{n}^{(i)}-\hat{H}%
	_{n}^{(i)}\hat{\Sigma}_{1n}^{(i)}\Vert \Vert (\hat{\Sigma}%
	_{1n}^{(i)})^{-1}\Vert  \notag \\
	& +\Vert \mathcal{N}_{n}(L_{\tau }^{\prime }U_{1n}\hat{\Sigma}%
	_{1n}^{-1}-L_{\tau }^{(i)}U_{1n}(\hat{\Sigma}_{1n}^{(i)})^{-1})\Vert
	_{2\rightarrow \infty }  \notag \\
	\leq & \{22.8\psi _{n}+8.8\Vert \mathcal{N}_{n}(\hat{U}_{1n}\hat{O}_{n}-\hat{%
		U}_{1n}^{(i)}\hat{O}_{n}^{(i)})\Vert _{2\rightarrow \infty
	}\}\gamma_n|\sigma _{Kn}^{-1}|  \notag \\
	& +\Vert \mathcal{N}_{n}L_{\tau }^{\prime }U_{1n}(\hat{\Sigma}_{1n}^{-1}-(%
	\hat{\Sigma}_{1n}^{(i)})^{-1})\Vert _{2\rightarrow \infty } +\Vert \mathcal{N%
	}_{n}(L_{\tau }^{(i)}-L_{\tau }^{\prime })U_{1n}(\hat{\Sigma}%
	_{1n}^{(i)})^{-1}\Vert _{2\rightarrow \infty }  \notag \\
	\leq & (22.8\psi _{n}+8.8\Vert \mathcal{N}_{n}(\hat{U}_{1n}\hat{O}_{n}-\hat{U%
	}_{1n}^{(i)}\hat{O}_{n}^{(i)})\Vert _{2\rightarrow \infty })\gamma_n|\sigma
	_{Kn}^{-1}|  \notag \\
	& +5.01\gamma_n|\sigma _{Kn}^{-1}|+1.01\Vert \mathcal{N}_{n}(L_{\tau
	}^{(i)}-L_{\tau }^{\prime })U_{1n}\Vert _{2\rightarrow \infty }|\sigma
	_{Kn}^{-1}|,  \label{eq:Lem7b}
	\end{align}%
	where the first inequality holds by the triangle inequality, the second
	inequality holds by the fact that $\Vert AB\Vert _{2\rightarrow \infty }\leq
	\Vert A\Vert _{2\rightarrow \infty }\Vert B\Vert $, $\hat{H}_{n}=\hat{U}%
	_{1n}^{T}U_{1n}$, and $\hat{H}_{n}^{(i)}=(\hat{U}_{n}^{(i)})^{T}U_{1n}$, the
	third inequality holds by the fact that $\Vert \mathcal{N}_{n}\hat{U}%
	_{1n}\Vert _{2\rightarrow \infty }\leq \psi _{n}$ and 
	\begin{align*}
	\Vert \mathcal{N}_{n}\hat{U}_{1n}^{(i)}\Vert _{2\rightarrow \infty }=& \Vert 
	\mathcal{N}_{n}\hat{U}_{1n}^{(i)}\hat{O}_{n}^{(i)}\Vert _{2\rightarrow
		\infty }\leq \Vert \mathcal{N}_{n}(\hat{U}_{1n}\hat{O}_{n}-\hat{U}_{1n}^{(i)}%
	\hat{O}_{n}^{(i)})\Vert _{2\rightarrow \infty }+\Vert \mathcal{N}_{n}\hat{U}%
	_{1n}\hat{O}_{n}\Vert _{2\rightarrow \infty } \\
	=& \psi _{n}+\Vert \mathcal{N}_{n}(\hat{U}_{1n}\hat{O}_{n}-\hat{U}_{1n}^{(i)}%
	\hat{O}_{n}^{(i)})\Vert _{2\rightarrow \infty }\mbox{ }a.s.,
	\end{align*}%
	and the last inequality holds by Lemma \ref{lem:DiDC}(iii) below. Finally,
	we bound the term $\Vert \mathcal{N}_{n}(L_{\tau }^{(i)}-L_{\tau }^{\prime
	})U_{1n}\Vert _{2\rightarrow \infty }. $ We have 
	\begin{align}
	& \Vert \mathcal{N}_{n}(L_{\tau }^{(i)}-L_{\tau }^{\prime })U_{1n}\Vert
	_{2\rightarrow \infty }  \notag \\
	=& \Vert \mathcal{N}_{n}(\mathcal{D}_{\tau }^{-1/2}A^{(i)}\mathcal{D}_{\tau
	}^{-1/2}-D_{\tau }^{-1/2}AD_{\tau }^{-1/2})U_{1n}\Vert _{2\rightarrow \infty
	}  \notag \\
	\leq & \Vert \mathcal{N}_{n}\mathcal{D}_{\tau }^{-1/2}(A^{(i)}-A)\mathcal{D}%
	_{\tau }^{-1/2}U_{1n}\Vert _{2\rightarrow \infty }+\Vert \mathcal{N}_{n}(%
	\mathcal{D}_{\tau }^{-1/2}-D_{\tau }^{-1/2})A\mathcal{D}_{\tau
	}^{-1/2}U_{1n}\Vert _{2\rightarrow \infty }  \notag \\
	& +\Vert \mathcal{N}_{n}D_{\tau }^{-1/2}A(\mathcal{D}_{\tau }^{-1/2}-D_{\tau
	}^{-1/2})U_{1n}\Vert _{2\rightarrow \infty }  \notag \\
	\leq & \Vert \mathcal{N}_{n}\mathcal{D}_{\tau }^{-1/2}(A^{(i)}-A)\mathcal{D}%
	_{\tau }^{-1/2}U_{1n}\Vert _{2\rightarrow \infty }+4.5\rho _{n}\gamma _{n}%
	\mbox{ }a.s.  \label{eq:Lem7c}
	\end{align}%
	where the first inequality hold by the triangle inequality and the second
	inequality holds by Lemma \ref{lem:DiDC}. In addition, by (\ref{AA}) we have 
	\begin{equation}
	\Vert \lbrack \mathcal{N}_{n}\mathcal{D}_{\tau }^{-1/2}(A^{(i)}-A)\mathcal{D}%
	_{\tau }^{-1/2}U_{1n}]_{s\cdot }\Vert =%
	\begin{cases}
	\Vert (n_{g_{s}^{0}}^{\tau })^{1/2}(\theta _{s}^{\tau })^{-1/2}(d_{s}^{\tau
	}d_{i}^{\tau })^{-1/2}(A_{si}-P_{si})u_{i}\Vert & \quad s\neq i \\ 
	\Vert (n_{g_{s}^{0}}^{\tau })^{1/2}(\theta _{s}^{\tau })^{-1/2}(d_{i}^{\tau
	})^{-1/2}([A]_{i\cdot }-[P]_{i\cdot })\mathcal{D}_{\tau }^{-1/2}U_{1n}\Vert
	& \quad s=i,%
	\end{cases}
	\label{eq:Lem7d}
	\end{equation}%
	where $u_i^T$ is the i's row of $U_{1n}$. By Assumption \ref{ass:nk3} and
	the fact that $||(n_{g_{i}^{0}}^{\tau })^{1/2}(\theta _{i}^{\tau
	})^{-1/2}u_{i} || = 1$, 
	\begin{equation}
	\sup_{i}\Vert (n_{g_{s}^{0}}^{\tau })^{1/2}(\theta _{s}^{\tau
	})^{-1/2}(d_{s}^{\tau }d_{i}^{\tau })^{-1/2}(A_{si}-P_{si})u_{1i}\Vert \leq
	C_1^{1/2}c_1^{-1/2}\overline{\theta }^{1/2}\underline{\theta }%
	^{-1/2}(\mu^\tau _{n})^{-1}\mbox{ }a.s..  \label{eq:Lem7e}
	\end{equation}%
	By Lemma \ref{lem:V1n3} and the facts that 
	\begin{equation*}
	\Vert U_{1n}\Vert _{2\rightarrow \infty }\leq c_1^{-1/2}\overline{\theta }%
	^{1/2}(K/n)^{1/2}\quad \text{and}\quad \Vert U_{1n}\Vert =1\mbox{ }a.s.,
	\end{equation*}%
	we have 
	\begin{align}
	&\sup_{i}(n_{g_{s}^{0}}^{\tau })^{1/2}(\theta _{s}^{\tau
	})^{-1/2}(d_{i}^{\tau })^{-1/2}\Vert ([A]_{i\cdot }-[P]_{i\cdot })\mathcal{D}%
	_{\tau }^{-1/2}U_{1n}\Vert  \notag \\
	\leq & 6C_1^{1/2}\biggl(\frac{c_1^{-1/2}\overline{\theta }^{1/2}(\log(n) +
		\log(5)K)}{\mu _{n}^{\tau }\underline{\theta }^{1/2}} \vee \frac{%
		(\log(n)+\log(5)K)^{1/2}\rho _{n}^{1/2}\overline{\theta }^{1/4}}{(\mu
		_{n}^{\tau })^{1/2}\underline{\theta }^{1/4}K^{1/2}}\biggr)\mbox{ }a.s.
	\label{eq:Lem7f}
	\end{align}%
	Combining (\ref{eq:Lem7d})--(\ref{eq:Lem7f}) with the fact that 
	\begin{align*}
	& \frac{(\log(n)+\log(5)K)^{1/2}\rho _{n}^{1/2}\overline{ \theta }^{1/4}}{%
		(\mu _{n}^{\tau })^{1/2}\underline{\theta }^{1/4}K^{1/2}} \geq \frac{%
		c_1^{-1/2}\overline{\theta }^{1/2}(\log(n) + \log(5)K)}{\mu _{n}^{\tau } 
		\underline{\theta }^{1/2}} \\
	\iff & \frac{\overline{\theta}^{1/2}(\log(n) + \log(5)K)K}{%
		c_1\mu_n^\tau \underline{\theta}^{1/2} \rho_n} \leq 1
	\end{align*}
	under Assumption \ref{ass:rate3} (as $\rho _{n}\geq 1$), we have 
	\begin{equation}
	\Vert \mathcal{N}_{n}\mathcal{D}_{\tau }^{-1/2}(A^{(i)}-A)\mathcal{D}_{\tau
	}^{-1/2}U_{1n}\Vert _{2\rightarrow \infty }\leq \frac{6C_1^{1/2}(\log(n)+%
		\log(5)K)^{1/2}\rho _{n}^{1/2}\overline{ \theta }^{1/4}}{(\mu _{n}^{\tau
		})^{1/2}\underline{\theta }^{1/4}K^{1/2}}\mbox{ }a.s.  \label{eq:Lem7g}
	\end{equation}%
	Substituting (\ref{eq:Lem7g}) into (\ref{eq:Lem7c}), we have 
	\begin{equation}
	\Vert \mathcal{N}_{n}(L_{\tau }^{(i)}-L_{\tau }^{\prime })U_{1n}\Vert
	_{2\rightarrow \infty }\leq 6C_1^{1/2} \gamma_n\left(\rho _{n}+ \frac{(\frac{%
			1}{K} + \frac{\log(5)}{\log(n)})^{1/2} \rho _{n}^{1/2}\overline{\theta }%
		^{1/4}}{\underline{\theta }^{1/4}}\right)\mbox{ }a.s.  \label{eq:Lem7h}
	\end{equation}%
	Combining (\ref{eq:Lem7h}) with (\ref{eq:Lem7a})-(\ref{eq:Lem7b}), we have 
	\begin{align*}
	& \Vert \mathcal{N}_{n}(\hat{U}_{1n}\hat{O}_{n}-\hat{U}_{1n}^{(i)}\hat{O}%
	_{n}^{(i)})\Vert _{2\rightarrow \infty } \\
	\leq & 68.74\biggl[(22.9\psi _{n}+8.8\Vert \mathcal{N}_{n}(\hat{U}_{1n}\hat{O%
	}_{n}-\hat{U}_{1n}^{(i)}\hat{O}_{n}^{(i)})\Vert _{2\rightarrow \infty })
	\gamma_n|\sigma _{Kn}^{-1}|+ 5.01\gamma_n|\sigma _{Kn}^{-1}| \\
	& +6.06C_1^{1/2}\left(\rho _{n}+ \frac{(\frac{1}{K} + \frac{\log(5)}{\log(n)}%
		)^{1/2} \rho _{n}^{1/2} \overline{\theta }^{1/4}}{\underline{\theta }^{1/4}}%
	\right)\gamma_n|\sigma _{Kn}^{-1}|\biggr] \\
	\leq & 604.92 \gamma_n|\sigma _{Kn}^{-1}|\Vert \hat{U}_{1n}\hat{O}_{n}-\hat{U%
	}_{1n}^{(i)}\hat{O}_{n}^{(i)}\Vert _{2\rightarrow \infty } \\
	& + 1574.15C_1^{1/2}\gamma_n|\sigma _{Kn}^{-1}|\left(\psi _{n}+\rho _{n}+%
	\frac{(\frac{1}{K} + \frac{\log(5)}{\log(n)})^{1/2} \rho _{n}^{1/2} 
		\overline{\theta }^{1/4}}{\underline{\theta }^{1/4}}\right)\mbox{ }a.s.
	\end{align*}%
	By rearranging terms and the fact that $\gamma _{n}|\sigma _{Kn}^{-1}| \leq
	0.0001$, we have, 
	\begin{equation*}
	\Vert \mathcal{N}_{n}(\hat{U}_{1n}\hat{O}_{n}-\hat{U}_{1n}^{(i)}\hat{O}%
	_{n}^{(i)})\Vert _{2\rightarrow \infty }\leq 1676C_1^{1/2}\gamma_n|\sigma
	_{Kn}^{-1}|\left(\psi _{n}+\rho _{n}+\frac{(\frac{1}{K} + \frac{\log(5)}{%
			\log(n)})^{1/2} \rho _{n}^{1/2} \overline{\theta }^{1/4}}{\underline{\theta }%
		^{1/4}}\right)\mbox{ }a.s.
	\end{equation*}
\end{proof}

\begin{lem}
	\label{lem:DiDC} Let $\gamma _{n}=(\log (n)/\mu _{n}^{\tau })^{1/2}$.
	Suppose that conditions in Theorem \ref{thm:main_DC} hold. Then, almost
	surely,
	
	(i) $\sup_{i}(n_{g_{i}^{0}}^{\tau })^{1/2}(\theta _{i}^{\tau })^{-1/2}(\hat{d%
	}_{i}^{\tau })^{-1/2}\Vert \lbrack A]_{i\cdot }(D_{\tau }^{-1/2}-\mathcal{D}%
	_{\tau }^{-1/2})U_{1n}\Vert \leq 2.25C_1^{1/2}c_1^{-1/2}\rho _{n}\gamma_n,$
	
	(ii) $\sup_{i}(n_{g_{i}^{0}}^{\tau })^{1/2}(\theta _{i}^{\tau })^{-1/2}\Vert
	((\hat{d}_{i}^{\tau })^{-1/2}-(d_{i}^{\tau })^{-1/2})[A]_{i\cdot }\mathcal{D}%
	_{\tau }^{-1/2})U_{1n}\Vert \leq 2.25C_1^{1/2}c_1^{-1/2}\rho _{n}\gamma_n,$ 
	
	(iii) $\Vert \mathcal{N}_{n}L_{\tau }^{\prime }U_{1n}(\hat{\Sigma}%
	_{1n}^{-1}-(\hat{\Sigma}_{1n}^{(i)})^{-1})\Vert _{2\rightarrow \infty }\leq
	5.01\gamma_n|\sigma _{Kn}^{-1}|.$
\end{lem}

\begin{proof}
	We prove (i) and (iii) as (ii) can be proved in the same manner as (i). In
	fact, (i) and (ii) still hold if $[A]_{i\cdot}$ is replaced by $[P]_{i\cdot}$
	or $([A]_{i\cdot}- [P]_{i\cdot})$ as the proof of (i) suggests that the
	dominant term is given by $[P]_{i\cdot}$. To show (i), let $S^{K-1}=\{g\in
	\Re ^{K},\Vert g\Vert =1\}$. Then, 
	\begin{align*}
	& \sup_{i}(n_{g_{i}^{0}}^{\tau })^{1/2}(\theta _{i}^{\tau })^{-1/2}(\hat{d}%
	_{i}^{\tau })^{-1/2} \Vert [A]_{i\cdot }(D_{\tau }^{-1/2}-\mathcal{D}_{\tau
	}^{-1/2})U_{1n}\Vert \\
	\leq & \sup_{i}\sup_{g\in S^{K-1}}(n_{g_{i}^{0}}^{\tau })^{1/2}(\theta
	_{i}^{\tau })^{-1/2}(\hat{d}_{i}^{\tau })^{-1/2} |[A]_{i\cdot }(D_{\tau
	}^{-1/2}-\mathcal{D}_{\tau }^{-1/2})\mathcal{N}_{n}^{-1}\mathcal{N}%
	_{n}U_{1n}g|
	\end{align*}%
	Let $V(g)=\mathcal{N}_{n}U_{1n}g$, which is an $n\times 1$ vector and $%
	V_{j}(g)$ be the $j$-th element of $V(g)$. Then, we have 
	\begin{equation*}
	\sup_{g\in S^{K-1}}\sup_{j}|V_{j}(g)|\leq \Vert \mathcal{N}_{n}U_{1n}\Vert
	_{2\rightarrow \infty } = 1.
	\end{equation*}%
	Note that 
	\begin{equation*}
	\frac{(A_{ij}-P_{ij})\theta _{j}^{1/2}}{\theta _{i}^{1/2}(d_{i}^{\tau
		})^{1/2}d_{j}^{\tau }} \leq \frac{\bar{\theta}^{1/2}}{\underline{\theta}
		^{1/2}(\mu_n^\tau)^{3/2}}
	\end{equation*}
	and 
	\begin{equation*}
	\sum_{j \neq i}\mathbb{E} \biggl( \frac{(A_{ij}-P_{ij})\theta _{j}^{1/2}}{%
		\theta _{i}^{1/2}(d_{i}^{\tau })^{1/2}d_{j}^{\tau }}\biggr)^2 \leq \sum_{j
		\neq i}\frac{P_{ij}\theta_j}{ \theta_i d_i^\tau (d_j^\tau)^2} \leq \frac{%
		\bar{\theta}}{\underline{\theta}(\mu_n^\tau)^2}.
	\end{equation*}
	Then, by Bernstein inequality, 
	\begin{align}  \label{eq:11}
	\sup_{i}\biggl|\sum_{j\neq i}\frac{(A_{ij}-P_{ij})\theta _{j}^{1/2}}{\theta
		_{i}^{1/2}(d_{i}^{\tau })^{1/2}d_{j}^{\tau }}\biggr|\leq 3\biggl[\frac{ 
		\overline{\theta }^{1/2}\log ^{1/2}(n)}{\underline{\theta }^{1/2}\mu
		_{n}^{\tau }} \vee \frac{\overline{\theta }^{1/2}\log (n)}{\underline{\theta 
		} ^{1/2}(\mu _{n}^{\tau })^{3/2}}\biggr] = 3\frac{ \overline{\theta }%
		^{1/2}\log ^{1/2}(n)}{\underline{\theta }^{1/2}\mu _{n}^{\tau }}\mbox{ }a.s.,
	\end{align}
	where the last equality holds because $\log(n)/\mu_n^\tau \leq 1$.
	
	Therefore, 
	\begin{align*}
	& (n_{g_{i}^{0}}^{\tau })^{1/2}(\theta _{i}^{\tau })^{-1/2}(\hat{d}%
	_{i}^{\tau })^{-1/2}|[A]_{i\cdot }(D_{\tau }^{-1/2}-\mathcal{D}_{\tau
	}^{-1/2})\mathcal{N}_{n}^{-1}V(g)| \\
	\leq & C_1^{1/2}c_1^{-1/2}\theta _{i}^{-1/2}(\hat{d}_{i}^{\tau })^{-1/2}%
	\biggl|\sum_{j\neq i}\frac{A_{ij}V_{j}(g)\theta _{j}^{1/2}(d_{j}^{\tau }-%
		\hat{d}_{j}^{\tau })}{(\hat{d}_{j}^{\tau }d_{j}^{\tau })^{1/2}((\hat{d}%
		_{j}^{\tau })^{1/2}+(d_{j}^{\tau })^{1/2})}\biggr| \\
	\leq & C_1^{1/2}c_1^{-1/2} \frac{1.022\times 4.23}{0.9791\times 1.9791}%
	\sum_{j\neq i}\frac{A_{ij}\theta _{j}^{1/2}\log ^{1/2}(n)}{\theta
		_{i}^{1/2}(d_{i}^{\tau })^{1/2}d_{j}^{\tau }} \\
	\leq & 2.24C_1^{1/2}c_1^{-1/2}\left[\left|\sum_{j\neq i}\frac{%
		(A_{ij}-P_{ij})\theta _{j}^{1/2}}{\theta _{i}^{1/2}(d_{i}^{\tau
		})^{1/2}d_{j}^{\tau }}\right|+\sum_{j=1}^{n}\frac{P_{ij}\theta _{j}^{1/2}}{%
		\theta _{i}^{1/2}(d_{i}^{\tau })^{1/2}d_{j}^{\tau }}\right]\log ^{1/2}(n) \\
	\leq & 2.24C_1^{1/2}c_1^{-1/2}\biggl[3\frac{\overline{\theta }^{1/2}\log
		^{1/2}(n)}{\underline{\theta }^{1/2}\mu _{n}^{\tau }}+\rho _{n}(\mu
	_{n}^{\tau })^{-1/2}\biggr]\log ^{1/2}(n) \\
	\leq & 2.25C_1^{1/2}c_1^{-1/2}\rho _{n}\mu _{n}^{-1/2}\log ^{1/2}(n)\mbox{ }%
	a.s.,
	\end{align*}%
	where the first inequality holds Assumption \ref{ass:nk3}, the second
	inequality holds by \eqref{eq:di3}, \eqref{eq:dhat-d}, and the fact that $%
	\sup_{g\in S^{K-1}}\sup_{j}|V_{j}(g)|\leq 1$, the third inequality is due to
	the triangle inequality, the fourth inequality is due to \eqref{eq:11} and
	the fact that 
	\begin{equation*}
	\sum_{j\neq i}\frac{P_{ij}\theta _{j}^{1/2}}{\theta _{i}^{1/2}(d_{i}^{\tau
		})^{1/2}d_{j}^{\tau }}\leq \sum_{j\neq i}\frac{\rho _{n}\theta _{j}}{%
		n(d_{j}^{\tau })^{1/2}}\leq \rho _{n}(\mu _{n}^{\tau })^{-1/2},
	\end{equation*}%
	and the last inequality holds because 
	\begin{align*}
	3\frac{\overline{\theta }^{1/2}\log ^{1/2}(n)}{\underline{ \theta }^{1/2}\mu
		_{n}^{\tau }} \leq 0.001\rho _{n}(\mu _{n}^{\tau })^{-1/2}.
	\end{align*}
	
	Next, we show (iii). First note that, by Lemmas \ref{lem:dk3} and \ref%
	{lem:Li-LDC}, and the facts that $\Vert \mathcal{N}_{n}U_{1n}\Vert
	_{2\rightarrow \infty }=1$ and $\log^{1/2}(n)(\mu_n^\tau)^{-1/2} \leq
	0.01|\sigma_{Kn}|$, we have 
	\begin{equation}  \label{eq:sigma1}
	\Vert \hat{\Sigma}_{1n}^{-1}-(\hat{\Sigma}_{1n}^{(i)})^{-1}\Vert \leq
	\sup_{i}\sup_{k=1,\cdots ,K}|\frac{\hat{\sigma}_{kn}-\hat{\sigma}_{kn}^{(i)}%
	}{\hat{\sigma}_{kn}\hat{\sigma}_{kn}^{(i)}}|\leq 5\gamma _{n}\sigma
	_{Kn}^{-2}\mbox{ }a.s.
	\end{equation}%
	and 
	\begin{equation}  \label{eq:sigma12}
	\Vert \Sigma _{1n}(\hat{\Sigma}_{1n}^{-1}-(\hat{\Sigma}_{1n}^{(i)})^{-1})%
	\Vert \leq \sup_{i}\sup_{k=1,\cdots ,K}|\frac{\sigma _{kn}(\hat{\sigma}%
		_{kn}- \hat{\sigma}_{kn}^{(i)})}{\hat{\sigma}_{kn}\hat{\sigma}_{kn}^{(i)}}%
	|\leq 5\gamma _{n}|\sigma _{Kn}^{-1}|\mbox{ }a.s..
	\end{equation}
	
	Note that 
	\begin{align}
	& \Vert \mathcal{N}_{n}L_{\tau }^{\prime }U_{1n}(\hat{\Sigma}_{1n}^{-1}-(%
	\hat{\Sigma}_{1n}^{(i)})^{-1})\Vert _{2\rightarrow \infty }  \notag
	\label{eq:LLU0DC} \\
	\leq & \Vert \mathcal{N}_{n}(L_{\tau }^{\prime }-\mathcal{L}_{\tau }^{\prime
	})U_{1n}(\hat{\Sigma}_{1n}^{-1}-(\hat{\Sigma}_{1n}^{(i)})^{-1})\Vert
	_{2\rightarrow \infty }+\Vert \mathcal{N}_{n}\mathcal{L}_{\tau }^{\prime
	}U_{1n}(\hat{\Sigma}_{1n}^{-1}-(\hat{\Sigma}_{1n}^{(i)})^{-1})\Vert
	_{2\rightarrow \infty }  \notag \\
	\leq & \Vert \mathcal{N}_{n}(L_{\tau }^{\prime }-\mathcal{L}_{\tau }^{\prime
	})U_{1n}(\hat{\Sigma}_{1n}^{-1}-(\hat{\Sigma}_{1n}^{(i)})^{-1})\Vert
	_{2\rightarrow \infty }+\Vert \mathcal{N}_{n}U_{1n}\Sigma _{1n}(\hat{\Sigma}%
	_{1n}^{-1}-(\hat{\Sigma}_{1n}^{(i)})^{-1})\Vert _{2\rightarrow \infty } 
	\notag \\
	\leq & \Vert \mathcal{N}_{n}(L_{\tau }^{\prime }-\mathcal{L}_{\tau }^{\prime
	})U_{1n}\Vert _{2\rightarrow \infty }\Vert \hat{\Sigma}_{1n}^{-1}-(\hat{%
		\Sigma}_{1n}^{(i)})^{-1}\Vert +\Vert \mathcal{N}_{n}U_{1n}\Vert
	_{2\rightarrow \infty }\Vert \Sigma _{1n}(\hat{\Sigma}_{1n}^{-1}-(\hat{\Sigma%
	}_{1n}^{(i)})^{-1})\Vert  \notag \\
	\leq & 5\Vert \mathcal{N}_{n}(L_{\tau }^{\prime }-\mathcal{L}_{\tau
	}^{\prime })U_{1n}\Vert _{2\rightarrow \infty }\gamma _{n}\sigma
	_{Kn}^{-2}+5\gamma _{n}|\sigma _{Kn}^{-1}|\mbox{ }a.s.,
	\end{align}%
	where the first inequality holds by the triangle inequality, the second
	inequality holds because $\mathcal{L}_{\tau }^{\prime }U_{1n}=U_{1n}\Sigma
	_{1n}$, the third inequality holds by the fact that $\Vert AB\Vert
	_{2\rightarrow \infty }\leq \Vert A\Vert _{2\rightarrow \infty }\Vert B\Vert 
	$, and the last inequality is due to \eqref{eq:sigma1} and \eqref{eq:sigma12}%
	.
	
	It remains to bound $\Vert \mathcal{N}_{n}(L_\tau^{\prime }-\mathcal{L}%
	_\tau^{\prime })U_{1n}\Vert _{2\rightarrow \infty }$. Note we have 
	\begin{align}
	& \Vert \mathcal{N}_{n}(L_{\tau }^{\prime }-\mathcal{L}_{\tau }^{\prime
	})U_{1n}\Vert _{2\rightarrow \infty }  \notag  \label{eq:IVDC} \\
	=& \Vert \mathcal{N}_{n}(D_{\tau }^{-1/2}AD_{\tau }^{-1/2}-\mathcal{D}_{\tau
	}^{-1/2}P\mathcal{D}_{\tau }^{-1/2})U_{1n}\Vert _{2\rightarrow \infty } 
	\notag \\
	\leq & \Vert \mathcal{N}_{n}(D_{\tau }^{-1/2}-\mathcal{D}_{\tau }^{-1/2})P%
	\mathcal{D}_{\tau }^{-1/2}U_{1n}\Vert _{2\rightarrow \infty }+\Vert \mathcal{%
		N}_{n}D_{\tau }^{-1/2}A(D_{\tau }^{-1/2}-\mathcal{D}_{\tau
	}^{-1/2})U_{1n}\Vert _{2\rightarrow \infty }  \notag \\
	& +\Vert \mathcal{N}_{n}D_{\tau }^{-1/2}(A-P)\mathcal{D}_{\tau
	}^{-1/2}U_{1n}\Vert _{2\rightarrow \infty }  \notag \\
	=:& T_{1}+T_{2}+T_{3}.
	\end{align}%
	By \eqref{eq:di3} and \eqref{eq:dhat-d}, we have 
	\begin{align}
	T_{1}\leq & \sup_{i}(n_{g_{i}^{0}}^{\tau })^{1/2}(\theta _{i}^{\tau
	})^{-1/2}\left|\frac{\hat{d}_{i}^{\tau }-d_{i}^{\tau }}{(\hat{d}_{i}^{\tau
		}d_{i}^{\tau })^{1/2}((\hat{d}_{i}^{\tau })^{1/2}+(d_{i}^{\tau })^{1/2})}%
	\right|\sup_{g\in S^{K-1}}\sum_{j=1}^{n}\frac{P_{ij}|u_{1j}^{T}g|}{%
		(d_{j}^{\tau })^{1/2}}  \notag  \label{eq:IV1DC} \\
	\leq & 2.25C_1^{1/2}c_1^{-1/2}\sum_{j=1}^{n}\frac{P_{ij}\theta _{j}^{1/2}}{%
		\theta _{i}^{1/2}d_{i}^{\tau }(d_{j}^{\tau })^{1/2}}\log ^{1/2}(n)  \notag \\
	\leq & 2.25C_1^{1/2}c_1^{-1/2}\gamma _{n}\rho_n\mbox{ }a.s..
	\end{align}%
	For $T_{2}$, by Lemma \ref{lem:DiDC}(i), we have 
	\begin{align}  \label{eq:IV2DC}
	T_{2}\leq 2.25C_1^{1/2}c_1^{-1/2}\rho _{n}\gamma _{n}\mbox{ }a.s.
	\end{align}
	
	For $T_{3}$, we have 
	\begin{align}
	T_{3}=& \sup_{i}\Vert (n_{g_{i}^{0}}^{\tau })^{1/2}(\theta _{i}^{\tau
	})^{-1/2}(\hat{d}_{i}^{\tau })^{-1/2}([A]_{i\cdot }-[P]_{i\cdot })\mathcal{D}%
	_{\tau }^{-1/2}U_{1n}\Vert  \notag  \label{eq:IV3DC} \\
	\leq & 1.03\sup_{i}\Vert (n_{g_{i}^{0}}^{\tau })^{1/2}(\theta _{i}^{\tau
	})^{-1/2}(d_{i}^{\tau })^{-1/2}([A]_{i\cdot }-[P]_{i\cdot })\mathcal{D}%
	_{\tau }^{-1/2}U_{1n}\Vert  \notag \\
	\leq & 6.18C_1^{1/2}\biggl(\frac{c_1^{-1/2}(\log (n)+\log(5)K)\overline{%
			\theta }^{1/2}}{\mu _{n}^{\tau }\underline{\theta }^{1/2}} \vee \frac{%
		(\log(n)+\log(5)K)^{1/2}\rho _{n}^{1/2}\overline{\theta }^{1/4}}{K^{1/2}(\mu
		_{n}^{\tau })^{1/2}\underline{\theta }^{1/4}}\biggr)  \notag \\
	= & 6.18C_1^{1/2}c_1^{-1/2}\frac{(\log(n)+\log(5)K)^{1/2}\rho _{n}^{1/2}%
		\overline{\theta }^{1/4}}{(\mu _{n}^{\tau })^{1/2}\underline{\theta }%
		^{1/4}K^{1/2}}\mbox{ }a.s.,
	\end{align}%
	where the first inequality holds because of \eqref{eq:di3}, the second
	inequality holds by Lemma \ref{lem:V1n3} and the fact that 
	\begin{equation*}
	\Vert U_{1n}\Vert _{2\rightarrow \infty }=\sup_{i}(\theta _{i}^{\tau
	}/n_{g_{i}^{0}}^{\tau })^{1/2}\leq c_1^{-1/2}(\overline{\theta }K/n)^{1/2},
	\end{equation*}%
	and the last equality holds because 
	\begin{align*}
	\frac{c_1^{-1}\overline{\theta}^{1/2}(\log(n) + \log(5)K)K}{\mu_n^\tau 
		\underline{\theta}^{1/2} \rho_n} \leq 1.
	\end{align*}
	
	Combining \eqref{eq:IVDC}--\eqref{eq:IV3DC}, we have 
	\begin{equation}
	\Vert \mathcal{N}_{n}(L_{\tau }^{\prime }-\mathcal{L}_{\tau }^{\prime
	})U_{1n}\Vert _{2\rightarrow \infty }\leq 6.18C_1^{1/2}c_1^{-1/2}\gamma
	_{n}\left(\rho _{n}+\frac{\rho _{n}^{1/2}\overline{\theta }^{1/4}(\frac{1}{K}
		+ \frac{\log(5)}{\log(n)})^{1/2}}{\underline{\theta }^{1/4}}\right).
	\label{eq:LLUDC}
	\end{equation}%
	Substituting \eqref{eq:LLUDC} into \eqref{eq:LLU0DC}, we have 
	\begin{equation*}
	\Vert LU_{1n}(\hat{\Sigma}_{1n}^{-1}-(\hat{\Sigma}_{1n}^{(i)})^{-1})\Vert
	_{2\rightarrow \infty }\leq 5.01\gamma _{n}|\sigma _{Kn}^{-1}|\mbox{ }a.s.,
	\end{equation*}%
	where we use the fact that 
	\begin{align*}
	5\times 6.18C_1^{1/2}c_1^{-1/2}\frac{\log ^{1/2}(n)\left(\rho _{n}+\frac{%
			\rho _{n}^{1/2}\overline{\theta }^{1/4}(\frac{1}{K} + \frac{\log(5)}{\log(n)}%
			)^{1/2}}{\underline{\theta }^{1/4}}\right)}{(\mu _{n}^{\tau })^{1/2}|\sigma
		_{Kn}|} \leq 0.01
	\end{align*}
	under Assumption \ref{ass:rate3}.
\end{proof}

\section{Additional simulation results\label{sec:addsim}}

In this section, we report some additional simulation results for DGPs 1-4
studied in the paper.

Table \ref{table:1} reports the classification results based on the
eigenvectors corresponding to the largest $K$ eigenvalues of $%
L=D^{-1/2}AD^{-1/2}$. Given an adjacency matrix $A$, $D$ is not invertible
when there exists a node which has degree 0. We also report the percentage
of replications which generate $A$ with strictly positive degrees for each
node in the table, denoted as Ratio. For these realizations, we report the
classification results. In Table \ref{table:1}, \textquotedblleft
CCP\textquotedblright\ indicates the Correct Classification Proportion
criterion; \textquotedblleft NMI\textquotedblright\ means the Normalized
Mutual Information criterion, and \textquotedblleft
kmeans\textquotedblright\ correspond to the classification methods K-means
with default options (Matlab \textquotedblleft kmedoids\textquotedblright ).
We summarize some important findings from Table \ref{table:1}. First, we
have a fair large probability to obtain zero degree for some nodes in DGPs
1--4 because we allow the minimum degree to diverge to infinity at a very
slow rate, namely at rate-$\log (n)$ in DGPs 1 and 3 and rate-$\log
^{5/6}(n) $ in DGPs 2 and 4. Second, the performance of the spectral
classification based on $L$ is not as satisfactory as that based on its
regularized version studied in the paper. This is especially true when $n/K$
is small.

\linespread{1.2}%
\begin{table}[H]%
	\caption{\label{table:1} Classification results based on $L = D^{-1/2} A
		D^{-1/2}$} \centering%
	\begin{tabular}{ccclcc}
		\hline\hline
		DGP & $K$ \  & $n/K$ \  & Ratio & CCP & NMI \\ \hline
		1 & 2 & 50 & \multicolumn{1}{c}{0.646} & 0.9805 & 0.8827 \\ 
		& 2 & 200 & \multicolumn{1}{c}{0.638} & 0.9927 & 0.9476 \\ \hline
		2 & 3 & 50 & \multicolumn{1}{c}{0.364} & 0.9751 & 0.9073 \\ 
		& 3 & 200 & \multicolumn{1}{c}{0.166} & 0.9906 & 0.9585 \\ \hline
		3 & 2 & 50 & \multicolumn{1}{c}{0.104} & 0.9651 & 0.7523 \\ 
		& 2 & 200 & \multicolumn{1}{c}{0.000} & -- & -- \\ \hline
		4 & 3 & 50 & \multicolumn{1}{c}{0.038} & 0.9543 & 0.7458 \\ 
		& 3 & 200 & \multicolumn{1}{c}{0.000} & -- & -- \\ \hline
	\end{tabular}
\end{table}%
\linespread{1.25}%

Figures \ref{fig:dgp_1_1}--\ref{fig:dgp_4_2} report the classification
results based on $L_{\tau }^{\prime }=D_{\tau }^{-1/2}AD_{\tau }^{-1/2}$ and 
$L_{\tau }=D_{\tau }^{-1/2}A_{\tau }D_{\tau }^{-1/2}$ for DGPs 1--2 and DGPs
3--4, respectively. As in the paper, the left column uses the CCP criterion
and the right column uses the NMI criterion to evaluate the classification
performance. The $x$-axis marks the $\tau $ values, i.e., $[10^{-4},(\tau
_{\max })^{0},(\tau _{\max })^{1/18},\ldots ,(\tau _{\max })^{18/18}],$
where $\tau _{\max }$ is the expected average degree. There are two curves
in each subplot. As marked in the legend and explained in the paper, they
represent classification results by using different classification methods.
In each subplot, the green dashed line is the pseudo $\tau $ value as
defined in \cite{JY16}. We summarize some findings from Figures \ref%
{fig:dgp_1_1}--\ref{fig:dgp_4_2}. First, the spectral classification results
first improve and then deteriorate as $\tau $ increases. Second, as Figures %
\ref{fig:dgp_1_1} and \ref{fig:dgp_2_1} suggest, the spectral clustering
based on $L_{\tau }^{\prime }=D_{\tau }^{-1/2}AD_{\tau }^{-1/2}$ with $\tau =%
\bar{d}$ or $\tau ^{\mathrm{JY}}$ is slightly worse than the UPL method.
Third, as Figures \ref{fig:dgp_3_2} and \ref{fig:dgp_4_2} suggest, the
method of \cite{JY16} tends to select too large a regularization parameter,
but still yields classification results that are much better than those of
CPL.

\begin{figure}[]
	\centering
	\includegraphics[scale = 0.5]{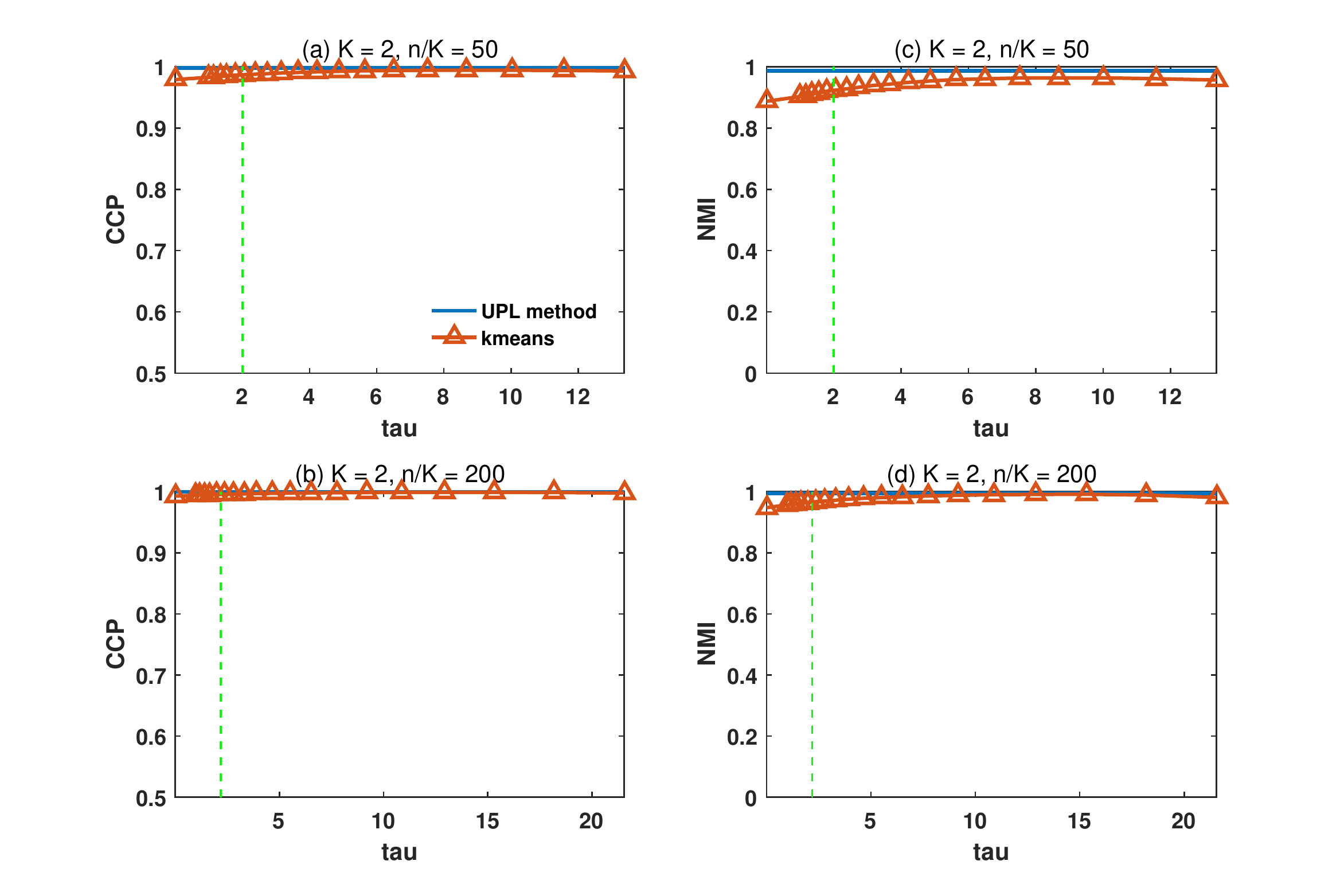}
	\caption{Classification results for CPL and K-means for DGP 1 ($K=2$) based
		on $L_{\protect\tau }^{\prime }=D_{ \protect\tau }^{-1/2}AD_{\protect\tau %
		}^{-1/2}$. The $x$-axis marks the $\protect\tau $ values and the $y$-axis is
		either CCP (left column) or NMI (right column). The green dashed vertical
		line in each subplot indicated the estimated $\protect\tau ^{\mathrm{JY}}$
		value by using the method of \protect\cite{JY16}. The first and second rows
		correspond to $n/K=50$ and 200, respectively.}
	\label{fig:dgp_1_1}
\end{figure}


\begin{figure}[]
	\centering
	\includegraphics[scale = 0.5]{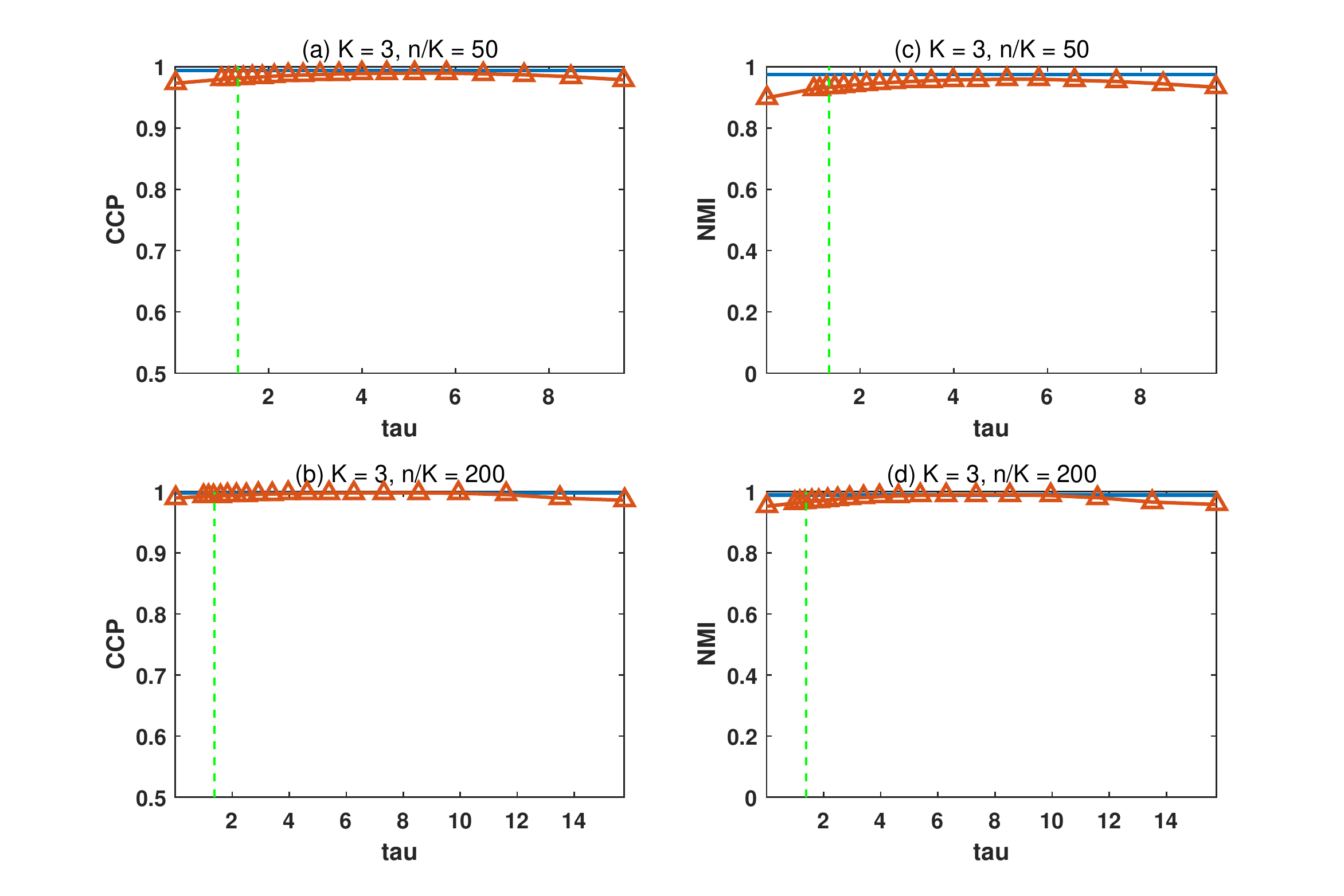}
	\caption{Classification results for DGP 2 ($K=3$) based on $L_{\protect\tau %
		}^{\prime }=D_{\protect\tau }^{-1/2}AD_{ \protect\tau }^{-1/2}$. (See Figure 
		\protect\ref{fig:dgp_1_1} for explanations.)}
	\label{fig:dgp_2_1}
\end{figure}


\begin{figure}[]
	\centering
	\includegraphics[scale = 0.5]{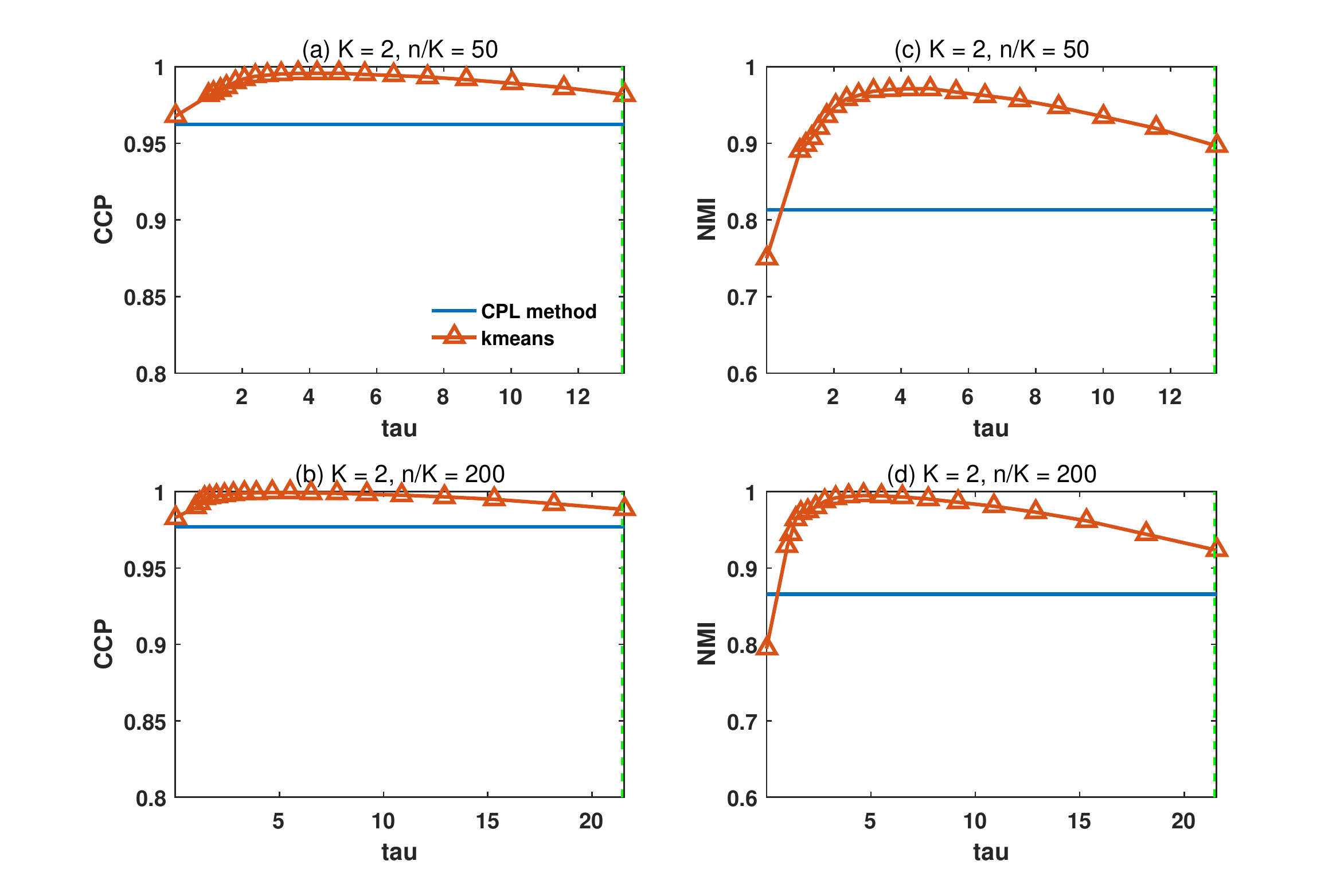}
	\caption{Classification results for DGP 3 ($K=2$, degree-corrected) based on 
		$L_{\protect\tau }=D_{\protect\tau }^{-1/2}A_{\protect\tau }D_{\protect\tau %
		}^{-1/2}$. (See Figure \protect\ref{fig:dgp_1_1} for explanations.)}
	\label{fig:dgp_3_2}
\end{figure}


\begin{figure}[]
	\centering
	\includegraphics[scale = 0.5]{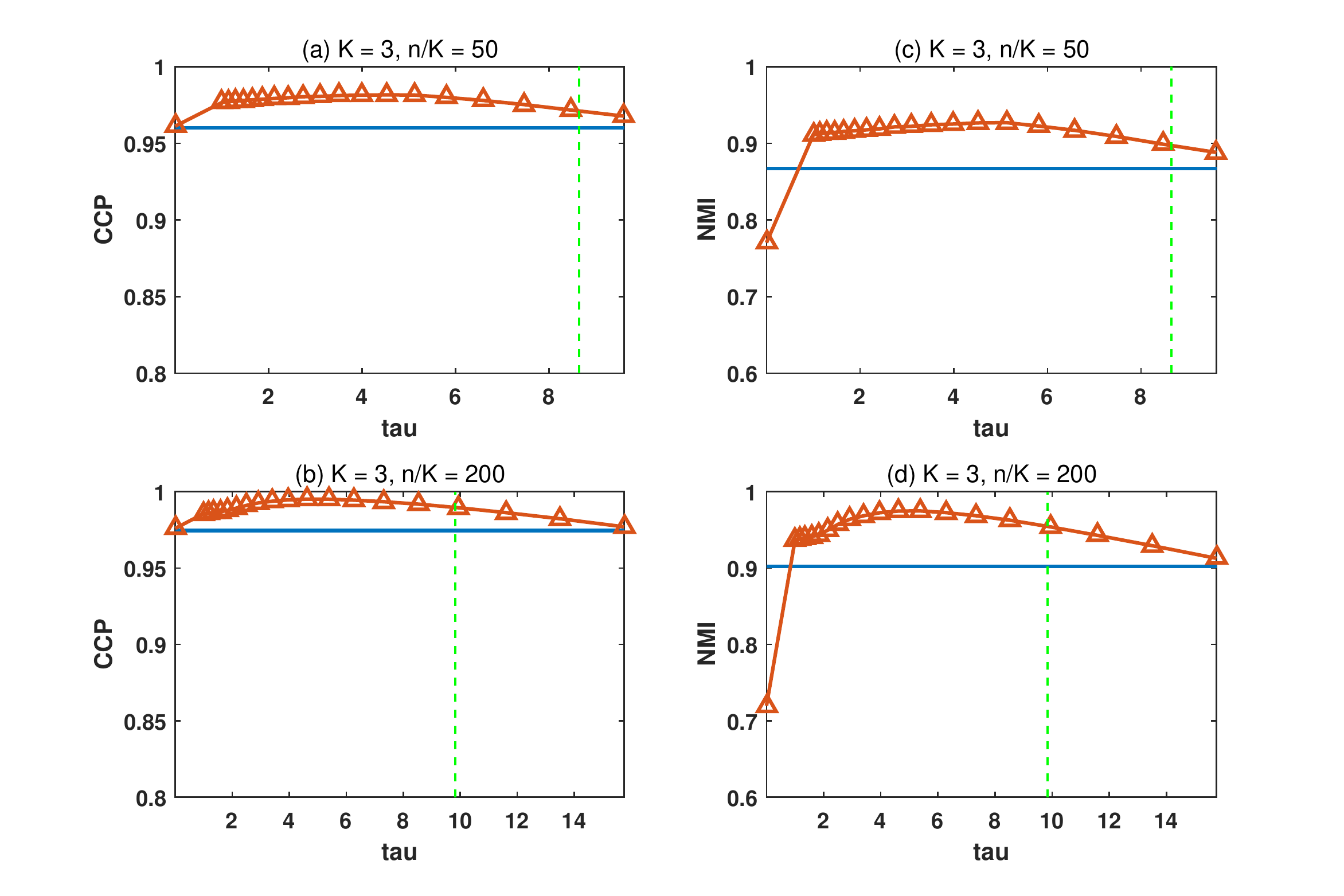}
	\caption{Classification results for DGP 4 ($K=3$, degree-corrected) based on 
		$L_{\protect\tau }=D_{\protect\tau }^{-1/2}A_{\protect\tau }D_{\protect\tau %
		}^{-1/2}$. (See Figure \protect\ref{fig:dgp_1_1} for explanations.)}
	\label{fig:dgp_4_2}
\end{figure}

\bibliographystyle{elsarticle-harv}
\bibliography{CD-1}

\end{document}